\documentclass[superscriptaddress, amsmath,amssymb,
aps, pra, letterpaper, tightenlines, reprint, notitlepages]{revtex4-2}

\usepackage{graphicx}
\usepackage{floatrow}

\usepackage{dcolumn}
\usepackage{bm}

\usepackage{amsfonts}
\usepackage{dsfont}
\usepackage{bm}
\usepackage{bbold}
\usepackage{xcolor}
\usepackage{lipsum,babel}
\usepackage[normalem]{ulem}

\usepackage[utf8]{inputenc}
\usepackage{amsmath}
\usepackage{physics}
\usepackage{hyperref}
\usepackage{amssymb}
\usepackage{braket}
\usepackage{mathtools}
\usepackage{color}
\usepackage{amsthm}
\usepackage{algorithm}
\usepackage{algorithmic}

\usepackage{lipsum}

\DeclareMathOperator*{\E}{\mathbb{E}}

\DeclareMathOperator*{\Prob}{Pr}

\DeclareMathOperator*{\Per}{\rm{Per}}

\DeclareMathOperator*{\poly}{\rm{poly}}
\DeclareMathOperator*{\supp}{\rm{supp}}
\DeclareMathOperator*{\Haar}{\rm{Haar}}

\theoremstyle{plain}
\newtheorem{theorem}{Theorem}
\newtheorem{corollary}{Corollary}

\newtheorem{definition}{Definition}
\newtheorem{lemma}{Lemma}
\newtheorem{conjecture}{Conjecture}
\newtheorem{remark}{Remark}

\theoremstyle{definition}

\begin{document}

\title{Threshold and Parity BosonSampling in the Linear-Mode Regime}

\author{Byeongseon Go}
\email{gbs1997@snu.ac.kr}
\affiliation{NextQuantum Center, Department of Physics and Astronomy, Seoul National University, Seoul 08826, Republic of Korea}
\author{Hyunseok Jeong}
\email{h.jeong37@gmail.com}
\affiliation{NextQuantum Center, Department of Physics and Astronomy, Seoul National University, Seoul 08826, Republic of Korea}
\author{Changhun Oh}
\email{changhun0218@gmail.com}
\affiliation{Department of Physics, Korea Advanced Institute of Science and Technology, Daejeon 34141, Republic of Korea}

\begin{abstract}
BosonSampling is among the most prominent candidates for demonstrating quantum advantage.
However, while the hardness of BosonSampling relies on photon-number-resolving detection, many experimentally relevant settings and applications instead use binary readout based on threshold or parity measurements, whose computational complexity has not yet been rigorously characterized.
In this work, we investigate the computational complexity of BosonSampling with threshold and parity measurements in the linear-mode regime, where the number of modes scales linearly with the number of photons and is most relevant to current experiments.
In particular, we establish average-case \#P-hardness of estimating typical output probabilities in threshold and parity BosonSampling, a crucial ingredient in proving the classical hardness of the corresponding sampling problems.
The resulting imprecision bounds match those obtained in prior hardness results for standard photon-number-resolving BosonSampling in the linear-mode regime.
The key technical ingredient is a Fourier-coefficient extraction method, induced by coherent beam-splitter rotations, that extracts hidden hard components within coarse-grained output probabilities.
These results indicate that the hard output-probability structure of photon-number-resolving BosonSampling can persist under natural binary coarse-grainings, even in collision-dominant regimes.


\end{abstract}

\maketitle

\section{Introduction}
BosonSampling was introduced as a restricted photonic model whose output distribution is expected to be hard to sample classically.
In standard BosonSampling, $n$ indistinguishable photons are injected into an $m$-mode random linear-optical circuit and measured using photon-number-resolving (PNR) detectors~\cite{aaronson2011computational}.
Because of its experimental accessibility and strong complexity-theoretic evidence for classical intractability, BosonSampling has become a prominent platform for photonic quantum advantage~\cite{oh2025recent}.
Since its original proposal, it has driven substantial theoretical and experimental progress, including refined complexity-theoretic foundations~\cite{bouland2022noise, bouland2026complexity, bouland2025exponential, aaronson2016bosonsampling, go2025quantum}, variants based on Gaussian input states~\cite{lund2014boson, hamilton2017gaussian, kruse2019detailed, deshpande2022quantum, grier2022complexity, go2025sufficient}, and large-scale photonic experiments reporting quantum-advantage demonstrations~\cite{zhong2020quantum, zhong2021phase, madsen2022quantum, deng2023gaussian, liu2026gaussian}.

However, full photon-number resolution is experimentally demanding, especially in large-scale photonic experiments.
As a result, many experimentally relevant BosonSampling implementations and applications rely on coarser measurement information than ideal PNR detection, as illustrated in Fig.~\ref{fig: schematic}.
Threshold measurement, which records only whether each output mode is occupied, is the most common example: it is the natural readout of on--off photodetectors and also underlies click-counting and pseudo-PNR schemes used in recent large-scale BosonSampling experiments~\cite{quesada2018gaussian, bressanini2024gaussian, dellios2025validation, bulmer2022threshold, zhong2020quantum, zhong2021phase, madsen2022quantum, deng2023gaussian, liu2026gaussian}.
Parity measurement, which records whether the occupation number in each mode is even or odd, provides another fundamental binary coarse-graining of PNR outcomes.
It arises naturally in atomic BosonSampling settings~\cite{young2024atomic} and superconducting circuit-QED platforms~\cite{wang2020efficient}, and has also attracted interest in photonic learning and generative modeling settings~\cite{bradler2021certain, kolarovszki2026generative, kurkin2026universality}.

Yet, despite strong evidence for the classical intractability of BosonSampling under PNR detection~\cite{aaronson2011computational, bouland2022noise, bouland2026complexity, bouland2025exponential}, it remains unclear whether such binary coarse-grainings preserve that intractability.
This question is especially nontrivial in the \textit{saturated} regime, where the number of modes is comparable to the number of photons, a regime of direct relevance to current large-scale experiments~\cite{bouland2026complexity, mhiri2026boson}.
In the dilute regime, where the number of modes is much larger than the number of photons, these binary measurements nearly coincide with PNR detection, due to the dominance of the collision-free outcomes~\cite{aaronson2011computational, arkhipov2012bosonic, qi2020regimes, go2026computational}.
In the saturated regime, by contrast, collisions occur with high probability, and a single threshold or parity bit string coarse-grains over exponentially many PNR outcomes.
Therefore, given the experimental relevance and extensive use of binary readouts in BosonSampling, it is crucial to understand how such binary coarse-graining affects its computational complexity in the saturated regime.

\begin{figure*}[t]
\includegraphics[width=0.8\linewidth]{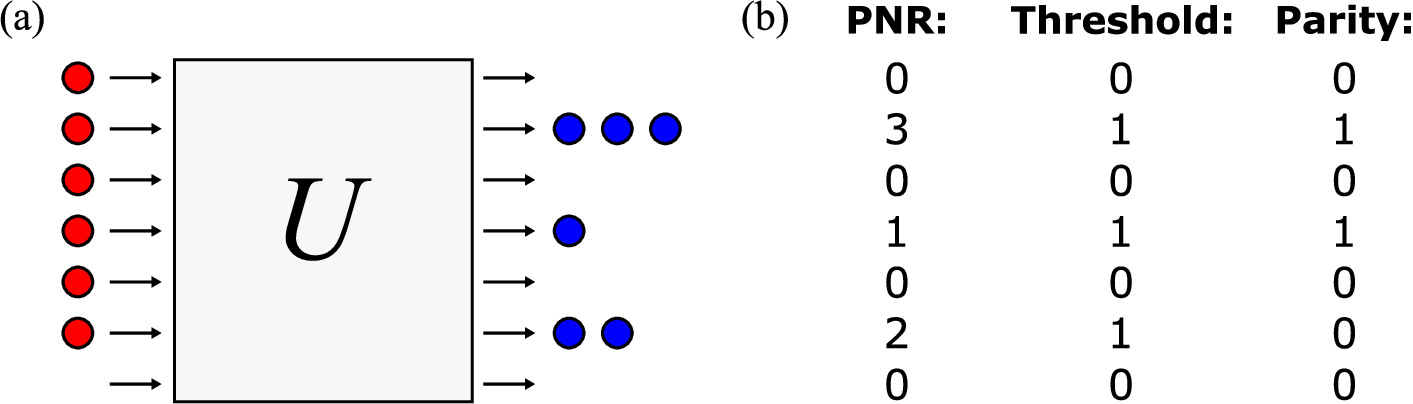}
\caption{(a) Illustration of the basic BosonSampling setup, in which $n$ single photons (red dots) are injected into an $m$-mode linear-optical circuit $U$ and measured at the output.
The blue dots illustrate one possible output occupation pattern, where multiple dots in the same output mode indicate a collision.
(b) The same output occupation pattern represented in three different measurement bases.
PNR detection records the full occupation vector, threshold measurement records whether each output mode is occupied, and parity measurement records each occupation number modulo $2$.
}
\label{fig: schematic}
\end{figure*}

In this work, we investigate the computational complexity of BosonSampling with threshold and parity measurements in the linear-mode regime $m=\alpha n$, focusing on fixed $\alpha>2$.
We prove average-case hardness of output-probability estimation for threshold and parity BosonSampling: estimating the typical binary output probability of a Haar-random linear-optical circuit is \#P-hard up to an explicit additive imprecision.
The resulting imprecision scale matches the best known bound for standard PNR BosonSampling in the saturated regime~\cite{bouland2026complexity}.
Moreover, using a reduction based on Stockmeyer's algorithm, we further show that improving this imprecision level to inverse-polynomial scale would yield the usual approximate-sampling hardness consequences for threshold and parity BosonSampling.
Finally, under an additional anti-concentration assumption for a Fourier coefficient, we present a route toward substantially improving the imprecision bound for both models.

Standard BosonSampling reductions exploit output probabilities that directly isolate a hard permanent associated with a small unitary submatrix.
Binary coarse-graining, however, destroys this direct isolation by combining many PNR outcomes; in the parity setting, the coarse-grained probability can moreover depend on the full input-output block.
To overcome these obstacles, we develop a novel worst-to-average-case reduction based on \textit{Fourier-coefficient extraction}.
The core idea is to apply coherent beam-splitter rotations to carefully chosen output-mode pairs.
These rotations render the corresponding output probability a finite Fourier series in the rotation angle, while preserving the Haar-random circuit distribution. 
A suitably chosen Fourier coefficient then acts as a \textit{row filter}, with a distinct role in each measurement setting.
For threshold measurements, it filters onto contributions with single occupancy of designated output modes, ensuring that the rows encoding the worst-case instance appear exactly once in every surviving permanent.
For parity measurements, the original coarse-grained probability generally depends on the full input-output block, whereas the highest-order Fourier coefficient filters out all configurations containing photons outside the selected rotated modes and thereby localizes the dependence to a small submatrix of the unitary.

This filtering allows a worst-case \#P-hard instance to be embedded into the selected Fourier coefficient by perturbing a Haar-random unitary submatrix.
We recover the encoded quantity from average-case estimates of the coarse-grained output probabilities by combining the robust Fourier-coefficient extraction procedure developed in this work with the polynomial-interpolation framework used in earlier BosonSampling reductions~\cite{bouland2022noise,bouland2026complexity,bouland2025exponential}.
This yields average-case hardness for both threshold and parity readout in the linear-mode regime.
More broadly, we expect that the techniques developed here may also apply to other natural coarse-grained measurement schemes in BosonSampling.


Our paper is organized as follows.
In Sec.~\ref{section: notations}, we define the BosonSampling setup, measurement models, and outcome ensembles used throughout this work.
In Sec.~\ref{section: main results}, we state the average-case probability-estimation problems and hardness results for threshold and parity BosonSampling and explain their common worst-to-average-case workflow and the order of the two proofs.
In Sec.~\ref{section: shared framework}, we develop the technical framework shared by the two proofs.
In Sec.~\ref{Section: threshold BS}, we prove the threshold BosonSampling hardness result.
In Sec.~\ref{Section: parity BS}, we prove the parity BosonSampling hardness result.
In Sec.~\ref{section: stockmeyer}, we generalize Stockmeyer’s reduction to the settings of threshold and parity BosonSampling.
In Sec.~\ref{section: alternative approach}, we present an alternative approach that can improve the imprecision bound under a Fourier coefficient anti-concentration conjecture.
Finally, in Sec.~\ref{section: conclusion}, we conclude the paper and outline several directions for future work.

\section{Setup and notation}\label{section: notations}

This section summarizes the basic notation used throughout the paper.
Let $n$ denote the number of photons and $m \geq n$ the number of modes, and write their relation $m=\alpha n$.

Let $\Haar(m)$ denote the Haar measure on the unitary group $\mathrm{U}(m)$.
Since there is a correspondence between an $m$-mode linear-optical circuit and a unitary matrix $U \in \mathrm{U}(m)$, we use the terms ``linear-optical circuit" and ``unitary matrix" interchangeably to denote $U$.

We first recall the PNR outcome set, for $n$ photons over $m$ modes. 
Explicitly, define
\begin{align}
    \mathcal{S}_{m,n}
    :=
    \left\{
    S=(s_1,\ldots,s_m)
    \;|\;
    s_j\in\mathbb{N}_{\geq 0},\;
    \sum_{j=1}^{m}s_j=n
    \right\}   ,
    \label{eq:main-Smn}
\end{align}
so that each $s_j$ denotes the occupation number of the $j$th mode. 
For $S\in\mathcal{S}_{m,n}$, let $\nu(S)$ denote the length-$n$ multiset of output-mode indices associated with $S$: the index $j\in[m]$ appears exactly $s_j$ times in $\nu(S)$.

For standard BosonSampling, we take the input state to have one photon in each of the first $n$ input modes and vacuum in the remaining modes. 
After propagation through an $m$-mode linear-optical circuit $U\in\mathrm{U}(m)$ and measurement with PNR detectors, the probability of observing an $n$-photon outcome $S\in\mathcal{S}_{m,n}$ is given by~\cite{aaronson2011computational}
\begin{align}
    p_S(U)
    =
    \frac{1}{\prod_{j=1}^{m}s_j!}
    \left|
    \Per\!\left(U_{\nu(S),[n]}\right)
    \right|^2  ,
    \label{eq:main-pnr-probability}
\end{align}
where $U_{A,B}$ denotes the submatrix of $U$ obtained by selecting the rows indexed by $A$ and the columns indexed by $B$; if $A$ is a multiset, rows are repeated according to their multiplicities.
By unitary invariance on the $n$-photon symmetric subspace, equivalently by Schur's lemma, the Haar average of $p_S(U)$ is uniform over $\mathcal{S}_{m,n}$:
\begin{align}
    \E_{U\sim\Haar(m)}[p_S(U)]
    =
    \frac{1}{|\mathcal{S}_{m,n}|}
    =
    \frac{1}{\binom{m+n-1}{n}}   ,
    \label{eq:main-pnr-average}
\end{align}
as also argued in~\cite{bouland2026complexity}.

\subsection{Threshold measurement}

Consider the threshold BosonSampling (TBS) setup, which follows the standard BosonSampling setup except that the output measurement is replaced by threshold measurement. 
In threshold measurement, each output mode is recorded as either $0$ (vacuum) or $1$ (at least a single photon), as illustrated in Fig.~\ref{fig: schematic}.
Throughout, we refer to an occupied mode, corresponding to outcome~1, as a \textit{click}.
For an occupation vector $S$, let $\supp(S):=\{j\in[m]:s_j>0\}$, and for a bit string $x$, let $\supp(x):=\{j\in[m]:x_j=1\}$.
The set of admissible threshold outcomes is then written as
\begin{align}
    \mathcal{T}_{m,n}
    :=
    \left\{
    x\in\{0,1\}^{m}
    \,\middle|\,
    1\leq |x|\leq n
    \right\},
    \label{eq:main-Tmn}
\end{align}
where $|x|$ is the Hamming weight of $x$, i.e., the number of clicked modes.
For $x\in\mathcal{T}_{m,n}$ and $U \in \mathrm{U}(m)$, the output probability of TBS is given by the coarse-grained form:
\begin{align}
    T_x(U)
    =
    \sum_{\substack{S\in\mathcal{S}_{m,n}\\ \supp(S)=\supp(x)}}
    p_S(U).
    \label{eq:main-threshold-probability}
\end{align}
In particular, every permanent contributing to Eq.~\eqref{eq:main-threshold-probability} uses only rows indexed by $\supp(x)$, so $T_x(U)$ depends only on the submatrix $U_{\supp(x),[n]}$.
For a fixed support of size $|x|$, the number of positive occupation vectors summing to $n$ is 
\begin{align}
    \binom{n-1}{|x|-1}  .
\end{align}

Throughout this work, we establish average-case hardness with respect to outcomes that are typical under Haar-random BosonSampling experiments, as also considered in~\cite{bouland2026complexity}.
However, while Ref.~\cite{bouland2026complexity} used the uniform distribution over possible PNR outcomes $\mathcal{S}_{m,n}$ to capture typical PNR outcomes, the analogous uniform distribution over threshold outcomes $\mathcal{T}_{m,n}$ does not naturally capture typical TBS outcomes.
This is because it generally differs from the marginal distribution induced by applying a threshold measurement to the output of a Haar-random unitary.

Accordingly, in our average-case formulation, we take the outcome distribution to be the \textit{Haar-averaged} output probability.
By counting possible $S \in \mathcal{S}_{m,n}$ satisfying $\supp(S)=\supp(x)$, the Haar-averaged probability is given by
\begin{align}
    T_x^{({\rm avg})}
    :=
    \E_{U\sim\Haar(m)}[T_x(U)]
    =
    \frac{
    \binom{n-1}{|x|-1}
    }{
    \binom{m+n-1}{n}
    }.
    \label{eq:main-threshold-average}
\end{align}

We then introduce the outcome ensemble underlying the average-case hardness over threshold outcomes.

\begin{definition}[Random threshold-outcome ensemble]\label{def: random threshold outcome ensemble}
    We define $\mathcal{D}_{m,n}$ to be the distribution on the set of threshold outcomes $\mathcal{T}_{m,n}$ in Eq.~\eqref{eq:main-Tmn}, with probability mass function given by $T_x^{(\mathrm{avg})}$ in Eq.~\eqref{eq:main-threshold-average} for each $x \in \mathcal{T}_{m,n}$.
\end{definition}

In fact, $T_x^{({\rm avg})}$ is the marginal probability of $x$ in experiments with Haar-random circuits, obtained by first drawing $U\sim\Haar(m)$ and then sampling the corresponding outcome.
Thus, $\mathcal{D}_{m,n}$ provides the natural typical-outcome distribution for Haar-random TBS experiments.
Also, by definition, sampling $x \sim \mathcal{D}_{m,n}$ can be done efficiently by first drawing $S\in\mathcal{S}_{m,n}$ uniformly at random and then outputting the unique binary string $x\in\{0,1\}^m$ satisfying $\supp(x)=\supp(S)$.

\subsection{Parity measurement}

Next, consider the parity BosonSampling (PBS) setup, which follows the standard BosonSampling setup but with a parity measurement.
Throughout, we label odd parity by $1$ and even parity by $0$, as illustrated in Fig.~\ref{fig: schematic}.
The admissible parity outcomes are
\begin{align}
    \mathcal{P}_{m,n}
    :=
    \left\{
    x\in\{0,1\}^{m}
    \,\middle|\,
    |x|\leq n,
    \ |x|\equiv n \ ({\rm mod} \ 2)
    \right\}  ,
    \label{eq:main-Pmn}
\end{align}
where $|x|$ here denotes the number of output modes with odd occupation. 
Note that the congruence condition enforces conservation of total photon-number parity.

For $x\in\mathcal{P}_{m,n}$ and linear-optical circuit $U \in \mathrm{U}(m)$, the output probability of PBS is given by 
\begin{align}
    P_x(U)
    =
    \sum_{\substack{S\in\mathcal{S}_{m,n}\\ S\ {\rm mod}\ 2=x}}
    p_S(U).
    \label{eq:main-parity-probability}
\end{align}
Unlike in the threshold setting, the parity constraint does not restrict the occupied modes to $\supp(x)$: modes outside $\supp(x)$ may contain a positive even number of photons, so $P_x(U)$ generally depends on the full $m\times n$ input-output block $U_{[m],[n]}$.
Also, the number of possible $S \in \mathcal{S}_{m,n}$ satisfying $S\ {\rm mod}\ 2=x$ in the summation is
\begin{align}\label{eq: number of S}
    \binom{m + \frac{n-|x|}{2} - 1}{\frac{n-|x|}{2}}   .
\end{align}
By Eq.~\eqref{eq:main-pnr-average} and Eq.~\eqref{eq: number of S}, the Haar-averaged probability for PBS is then given by
\begin{align}
    P_x^{({\rm avg})}
    :=
    \E_{U\sim\Haar(m)}[P_x(U)]
    =
    \frac{
    \binom{m + \frac{n-|x|}{2} - 1}{\frac{n-|x|}{2}}
    }{
    \binom{m+n-1}{n}
    }.
    \label{eq:main-parity-average}
\end{align}

We also define the analogous outcome ensemble for average-case parity outcomes.

\begin{definition}[Random parity-outcome ensemble]\label{def: random parity outcome ensemble}
    We define $\mathcal{Y}_{m,n}$ to be the distribution on the set of parity outcomes $\mathcal{P}_{m,n}$ in Eq.~\eqref{eq:main-Pmn}, with probability mass function given by $P_x^{(\mathrm{avg})}$ in Eq.~\eqref{eq:main-parity-average} for each $x \in \mathcal{P}_{m,n}$.
\end{definition}

Similarly, $P_x^{({\rm avg})}$ is the marginal probability of $x$ in PBS experiments with Haar-random circuits, so $\mathcal{Y}_{m,n}$ provides the natural typical-outcome distribution for PBS experiments.
Also, sampling $x \sim \mathcal{Y}_{m,n}$ is equivalent to first sampling $S$ according to $1/|\mathcal{S}_{m,n}|$ and then outputting $x = S\ {\rm mod}\ 2$, which can be done in classical polynomial time.

\section{Main results and proof architecture}\label{section: main results}

\subsection{Average-case threshold-probability estimation}\label{subsec: threshold main result}

We follow the standard approach in BosonSampling hardness results~\cite{aaronson2011computational, hamilton2017gaussian, kruse2019detailed, deshpande2018dynamical, bouland2019complexity, bouland2022noise, grier2022complexity, bouland2026complexity, bouland2025exponential, go2026computational} and formulate an average-case estimation problem for threshold output probabilities, which we call average-case threshold-output probability estimation (ATE).
We focus on the saturated regime $m=\Theta(n)$, where photon collisions occur with high probability and which is of direct relevance to current experiments~\cite{bouland2026complexity, mhiri2026boson}.

\begin{definition}[Average-case Threshold-probability Estimation: ATE]\label{def: thresholdBS}
    The {\rm ATE} problem is defined as follows.
    Given a pair $(x,U)$ drawn from $x\sim\mathcal{D}_{m,n}$ and $U\sim\Haar(m)$ independently, the task is to output an additive estimate $\hat T$ of $T_x(U)$ such that
    \begin{align}
        \Pr_{\substack{x\sim\mathcal{D}_{m,n} \\ U\sim\Haar(m)}}
        \left[
        |\hat T-T_x(U)|>\varepsilon T_x^{(\mathrm{avg})}
        \right]
        <\delta .
    \end{align}
\end{definition}

Our main hardness result for TBS is the following.

\begin{theorem}[Average-case hardness for TBS]\label{thm: threshold average case main}
    Let $m=\alpha n$ for a constant $\alpha>2$.
    Then the {\rm ATE} problem in Definition~\ref{def: thresholdBS} is $\#\mathrm{P}$-hard under a $\mathrm{BPP}^{\mathrm{NP}}$ reduction with additive imprecision $\varepsilon=\exp(-3n\log n-O(n))$ and failure probability $\delta=\poly(n)^{-1}$.
\end{theorem}

The proof of Theorem~\ref{thm: threshold average case main} is given in Sec.~\ref{Section: threshold BS}.
While this average-case estimation result does not directly imply approximate-sampling hardness at inverse-polynomial total-variation distance, Sec.~\ref{section: stockmeyer} shows that such a sampling-hardness consequence follows if the ATE hardness can be strengthened to inverse-polynomial imprecision and failure probability.

Note that the obtained imprecision bound is comparable to the best known bound for BosonSampling in the saturated regime~\cite{bouland2026complexity}, while remaining weaker than the bound achieved in the dilute regime~\cite{bouland2025exponential}.
In Sec.~\ref{section: alternative approach}, we establish under a threshold Fourier-coefficient anti-concentration conjecture that the imprecision bound for TBS can be improved to $\exp(-O(n^\lambda))$ for any fixed $0<\lambda<1$, together with an analogous PBS result under the corresponding parity conjecture.

\subsection{Average-case parity-probability estimation}\label{subsec: parity main result}

We analogously formulate the corresponding average-case estimation problem for parity output probabilities, which we call average-case parity-output probability estimation (APE).

\begin{definition}[Average-case Parity-probability Estimation: APE]\label{def: parityBS}
    The {\rm APE} problem is defined as follows.
    Given a pair $(x,U)$ drawn from $x\sim\mathcal{Y}_{m,n}$ and $U\sim\Haar(m)$ independently, the task is to output an additive estimate $\hat P$ of $P_x(U)$ such that
    \begin{align}
        \Pr_{\substack{x\sim\mathcal{Y}_{m,n} \\ U\sim\Haar(m)}}
        \left[
        |\hat P-P_x(U)|>\varepsilon P_x^{(\mathrm{avg})}
        \right]
        <\delta .
    \end{align}
\end{definition}

We then establish the corresponding hardness result for PBS.

\begin{theorem}[Average-case hardness for PBS]\label{thm: parity average case main}
    Let $m=\alpha n$ for a constant $\alpha>2$.
    The {\rm APE} problem in Definition~\ref{def: parityBS} is $\#\mathrm{P}$-hard under a $\mathrm{BPP}^{\mathrm{NP}}$ reduction with additive imprecision $\varepsilon=\exp(-3n\log n-O(n))$ and failure probability $\delta=\poly(n)^{-1}$.
\end{theorem}

The proof of Theorem~\ref{thm: parity average case main} is provided in Sec.~\ref{Section: parity BS}.
As in the TBS setting, extending the \#P-hardness guarantee in Theorem~\ref{thm: parity average case main} to inverse-polynomial imprecision $\varepsilon=\poly(n)^{-1}$ would yield the standard approximate-sampling hardness result for PBS; see Sec.~\ref{section: stockmeyer} for details.
Also, the same comparison between the saturated- and dilute-mode imprecision bounds carries over to PBS, and Sec.~\ref{section: alternative approach} presents an alternative route to improving the imprecision bound under an additional assumption.


\subsection{Common worst-to-average-case reduction framework}\label{subsec: common workflow}

At a high level, the proofs of Theorems~\ref{thm: threshold average case main} and~\ref{thm: parity average case main} follow the same worst-to-average-case reduction procedure as the existing reductions in Refs.~\cite{bouland2022noise,deshpande2022quantum,bouland2026complexity,bouland2025exponential}: oracle access to the corresponding average-case estimation problem enables the computation of a worst-case \#P-hard quantity.
Existing worst-to-average-case reductions for standard BosonSampling rely on three main ingredients:
(i) low-degree polynomial dependence of the output probabilities on a perturbation parameter;
(ii) a simple output-probability expression, typically a single permanent term, that depends only on a small submatrix and can encode a worst-case \#P-hard quantity; and
(iii) distributional stability of small Haar-random submatrices under the perturbation.
Concretely, these reductions perturb a small Haar-random submatrix toward a worst-case matrix and recover the encoded permanent from average-case probability estimates through polynomial interpolation~\cite{bouland2022noise,deshpande2022quantum,bouland2026complexity,bouland2025exponential}.


However, a direct extension of the standard argument encounters the following substantial obstacles.
First, in both models, a binary-output probability is a coarse-grained sum over many PNR outcomes, including collision events, so the permanent encoding the hard instance is not isolated as a single contribution.
In the threshold setting, $T_x(U)$ sums over all positive occupation-number compositions supported on $\supp(x)$.
The reduction must therefore filter this sum to retain only configurations in which each designated clicked mode contains exactly one photon, thereby enabling the encoding of a single worst-case permanent.
For PBS, there is an additional obstacle: the parity constraint permits modes outside $\supp(x)$ to contain a positive even number of photons, so the coarse-grained probability generally depends on all entries of the relevant $m\times n$ input-output block of $U$.
This global dependence prevents the direct application of perturbative stability bounds that control only selected Haar-random submatrices.

To overcome these obstacles, we introduce Fourier coefficients induced by coherent beam-splitter rotations.
A suitably chosen top Fourier coefficient serves as a row filter, but its role differs between the two models.
For TBS, the top Fourier coefficient enforces exactly one photon in each designated unrotated clicked mode.
This single-occupation constraint ensures that the hard rows appear exactly once and enables us to encode a single worst-case permanent within the Fourier coefficient, rather than inside a sum over many occupation patterns.
For PBS, the original coarse-grained probability generally has global row dependence, and the top Fourier coefficient forces all photons into the selected rotated modes.
This localizes the coefficient to a selected submatrix of $U$, thereby enabling the small-block perturbation required by the Haar-submatrix stability argument.

After these model-specific filtering steps, both reductions encode a worst-case permanent into the relevant Fourier coefficient for typical outcomes.
The corresponding worst-case matrix is then embedded into a selected Haar-random submatrix through the shared perturbation-stability framework.
Finally, we recover the encoded value from average-case probability estimates through Fourier-coefficient extraction followed by polynomial interpolation.

We develop the technical ingredients common to both hardness proofs in Sec.~\ref{section: shared framework}, including paired beam-splitter rotations, polynomial interpolation, Fourier-coefficient extraction, and perturbation-stability bounds.
We then treat the two models separately.
In Sec.~\ref{Section: threshold BS}, we prove the threshold result, showing how to choose the rotated modes so that a single worst-case permanent is encoded into the top Fourier coefficient.
In Sec.~\ref{Section: parity BS}, we prove the parity result, showing how the top Fourier coefficient localizes the dependence to a small submatrix of the unitary and how a single worst-case permanent can be encoded into this coefficient.

\section{Shared technical ingredients}\label{section: shared framework}

In this section, we argue the common technical ingredients underlying the threshold and parity hardness proofs.
We begin by defining the paired beam-splitter rotations that induce the Fourier expansion used in the reductions.
We next formulate the robust polynomial-interpolation used in previous hardness results~\cite{bouland2022noise, deshpande2022quantum, bouland2026complexity, bouland2025exponential}, and Fourier-coefficient extraction procedures that constitute a new ingredient of our approach.
Finally, we outline the perturbation-stability bounds and conditional-completion arguments that allow these procedures to be applied to Haar-random unitaries.

\subsection{Paired beam-splitter rotations}\label{subsec: paired beam-splitter rotations}

For $\theta\in[0,2\pi)$, define the real beam-splitter rotation
\begin{align}
B(\theta):=
\begin{pmatrix}
\cos\theta & \sin\theta\\
-\sin\theta & \cos\theta
\end{pmatrix}.
\end{align}
By Euler's formula, $B(\theta)$ above can alternatively be expressed as
\begin{align}
B(\theta)=e^{i\theta}B^+ + e^{-i\theta}B^-
,\quad 
B^\pm := \frac12
\begin{pmatrix}
1 & \mp i\\
\pm i & 1
\end{pmatrix}  .
\end{align}

Let $R\subseteq[m]$ be an active set of size $2s$, together with a fixed ordered pairing
\begin{align}
    R=\{r_1^{(1)},r_1^{(2)}\}\sqcup\cdots\sqcup\{r_s^{(1)},r_s^{(2)}\}.
\end{align}
The ordering within each pair is part of the data defining the rotation.

Given this paired active set $R$, we define $V_R(\theta)\in\mathrm{U}(m)$ as the unitary that applies $B(\theta)$ to each ordered pair $(r_j^{(1)},r_j^{(2)})$ and acts as the identity on $R^c=[m]\setminus R$.

To present a more concrete construction of $V_R(\theta)$, let
\begin{align}\label{eq:VRpm-active}
V_s^\pm := \bigoplus_{j=1}^s B^\pm \in \mathbb{C}^{2s \times 2s}   .
\end{align}
Let $P_R$ be the permutation matrix that orders the modes as
$(r_1^{(1)},r_1^{(2)},r_2^{(1)},r_2^{(2)},\dots,r_s^{(1)},r_s^{(2)})$
followed by the modes in $R^c := [m] \setminus R$.
We extend $V_s^{\pm}$ to $m\times m$ matrices supported on $R$ by
\begin{align}\label{eq:Vpm-full}
\begin{split}
V_R^\pm := P_R^\top\Big(V_s^\pm\oplus 0_{m-2s}\Big)P_R, 
\end{split}
\end{align}
and define the projector matrix onto the unaffected modes $R^c$
\begin{align}
\Pi_{R^c} := P_R^\top\Big(0_{2s}\oplus I_{m-2s}\Big)P_R.
\end{align}
Then we obtain the expression for $V_R(\theta)$ by
\begin{align}\label{eq:V-decomp-full}
V_{R}(\theta):=\Pi_{R^c} + e^{i\theta} V_R^+ + e^{-i\theta}V_R^-   .
\end{align}

For later use, write $U_R(\theta):=V_R(\theta)U$.
For a fixed $S \in \mathcal{S}_{m,n}$, let $n_R(S):=\sum_{i\in R}s_i$ be the number of photons detected in the output modes indexed by $R$.
Let $U_j$ denote the $j$th row of $U$. 
Then, by Eq.~\eqref{eq:V-decomp-full}, we have
\begin{align}\label{eq: rows of rotated circuit}
(U_R(\theta))_j 
=
\begin{cases}
e^{i\theta}(V_R^+U)_j + e^{-i\theta}(V_R^-U)_j \!\! \! & j\in R,\\
U_j & j\notin R.
\end{cases}
\end{align}
Accordingly, the submatrix $U_R(\theta)_{\nu(S),[n]}$ contains exactly $n_R(S)$ $\theta$-dependent rows, counted with multiplicity, while all other rows are independent of $\theta$.

A key property we use is the \textit{row multilinearity} of the permanent.
If $M_{j\leftarrow r}$ denotes the matrix obtained from $M$ by replacing its $j$th row by the row vector $r$, then for any row vectors $r_1,r_2$ and scalars $a,b$,
\begin{align}
    \Per\!\left(M_{j\leftarrow ar_1+br_2}\right)
    =
    a\,\Per\!\left(M_{j\leftarrow r_1}\right)
    +
    b\,\Per\!\left(M_{j\leftarrow r_2}\right).
\end{align}
Applying this identity to the $\theta$-dependent $n_R(S)$ rows of $U_R(\theta)_{\nu(S),[n]}$, counted with multiplicity, shows that $\Per(U_R(\theta)_{\nu(S),[n]})$ has degree at most $n_R(S)$ in $e^{\pm i\theta}$.
Consequently, its squared modulus is a trigonometric polynomial of degree at most $2n_R(S)$.
The threshold and parity band-limit statements follow by combining this observation with the corresponding constraints on the contributing PNR outcomes.

\subsection{Robust polynomial interpolation}\label{subsec: robust polynomial interpolation}

For polynomial interpolation, we use the robust Berlekamp--Welch algorithm, which has been employed repeatedly in recent hardness proofs for BosonSampling~\cite{bouland2022noise, deshpande2022quantum, bouland2026complexity, bouland2025exponential}.

\begin{theorem}[Robust Berlekamp--Welch algorithm~\cite{bouland2022noise}]
\label{thm:probabilistic-robust-BW-extrapolation}
Let $p(t)\in\mathbb C[t]$ be a polynomial of degree at most $d$.
Let $0<\Delta<1$, and let $\{t_j\}_{j = 1}^{100d^2}$ be equally-spaced points in $[0,\Delta]$. 
Then there exists a ${\rm P^{NP}}$ algorithm that, given data points $\{(t_j, y_j)\}_{j = 1}^{100d^2}$ satisfying 
\begin{align}
\Prob\left[
| y_j -p(t_j)|> \epsilon
\right]<\frac14  
\label{eq:BW-input-failure}
\end{align}
for every $j$, outputs an estimate $\widehat p(1)$ such that 
\begin{align}
\Prob\left[
|\widehat p(1)-p(1)|
>
\epsilon\exp\left(d\log\Delta^{-1}+O(d)\right)
\right]
<
\frac13.
\label{eq:BW-output-failure}
\end{align}
\end{theorem}

\subsection{Robust Fourier-coefficient extraction}
\label{subsec: coeff extract}

To estimate a Fourier coefficient from noisy probability estimates, we develop a new robust Fourier extraction procedure for trigonometric polynomials.

\begin{theorem}[Robust Fourier coefficient extraction algorithm]
\label{thm:probabilistic-robust-fourier-extraction}
Fix any failure probability $p_{\mathrm{fail}}<1/8$. 
Let $d>1$, and let $f(\theta) = \sum_{\ell=-d}^{d}c_\ell e^{i\ell\theta}$ be a trigonometric polynomial of degree at most $d$.
Let $L := C_F d$ for $C_F:=\left\lceil\max\left\{6,3/(1-8p_{\mathrm{fail}})\right\}\right\rceil$, and set $\theta_j:= {2\pi j}/L$ for $j=0,\ldots,L-1$.
Then, there exists a ${\rm P^{NP}}$ algorithm that, given data points $\{(\theta_j, y_j)\}_{j=0}^{L-1}$ satisfying
\begin{align}
\Prob\left[
|y_j-f(\theta_j)|>\epsilon
\right]\le p_{\mathrm{fail}}   
\label{eq:fourier-input-failure-prob}
\end{align}
for every $j$, outputs an estimate $\widehat c_{d}$ of the top Fourier coefficient $c_{d}$ of $f$ such that 
\begin{align}
\Prob\left[
|\widehat c_{d}-c_{d}|
>
2\epsilon\left(\frac{eC_F}{2}\right)^{2d}
\right]
<
\frac14 .
\label{eq:fourier-output-failure-prob}
\end{align}
\end{theorem}

\begin{proof}
See Appendix~\ref{appendix: section: robust Fourier extraction}. 
\end{proof}

\subsection{Perturbation-stability and unitary completion}\label{subsec: shared perturbation framework}

To carry out the worst-to-average-case reductions, we must query the corresponding estimation oracle on instances that remain close to Haar-random while still containing a small imprint of the worst-case block.
In each model-specific proof, let $A$ denote the selected submatrix of a Haar-random unitary and let $A_\star$ denote the corresponding block encoding the worst-case quantity.
We interpolate between them along the affine line
\begin{align}
    A(t)=(1-t)A+tA_\star .
\end{align}
The reductions use only small values of $t$ for oracle queries, so the distribution of $A(t)$ must remain close to the original Haar-submatrix distribution, as is likewise required in related worst-to-average-case reductions~\cite{aaronson2011computational, bouland2019complexity, bouland2022noise, bouland2026complexity, movassagh2023hardness, bouland2025exponential, kondo2022quantum, krovi2022average, go2026computational}.

We first analyze this stability under perturbation at the level of the selected block.
Appendix~\ref{appendix: section: translation invariance} derives the density of rectangular Haar-random submatrices and proves the general shift-and-scale stability result used in both model-specific proofs.
While the shift-and-scale stability of square Haar-random submatrices has already been established in~\cite{bouland2026complexity} and the TBS specialization uses a square $n\times n$ block, the submatrix arising in the PBS setting generally need not be square.
We therefore independently derive the corresponding result for rectangular Haar-random submatrices, which covers both model-specific proofs and thereby makes the present analysis self-contained.

Next, since both the ATE and APE problems take a full unitary matrix as input, rather than only the selected block, we lift the block-level closeness to a distributional closeness statement over full unitaries.
Appendix~\ref{appendix: section: conditional-haar-completion} constructs the required randomized conditional-completion procedure and shows that applying the same completion procedure does not increase the total variation distance between the input block distributions, up to an exponentially small contribution from the boundary event.
The threshold and parity sections state the corresponding dimension-dependent perturbation-stability and completion lemmas used in their respective reductions.

\section{Threshold BosonSampling}\label{Section: threshold BS}

\subsection{Proof overview}\label{subsec: threshold proof overview}

We first give a proof overview for Theorem~\ref{thm: threshold average case main}, building on the common technical framework established in Sec.~\ref{section: shared framework}.
The proof implements the worst-to-average-case reduction outlined in Sec.~\ref{subsec: common workflow}.
Specifically, coherent beam-splitter rotations induce a Fourier expansion, a suitably chosen Fourier coefficient provides the filtering structure needed to encode a worst-case permanent, and the encoded value is recovered by perturbing a Haar-random submatrix toward the worst-case block, followed by robust Fourier-coefficient extraction and polynomial interpolation.
The overall construction is illustrated in Fig.~\ref{fig: outlineTBS}.

\begin{figure*}[t]
\includegraphics[width=0.9\linewidth]{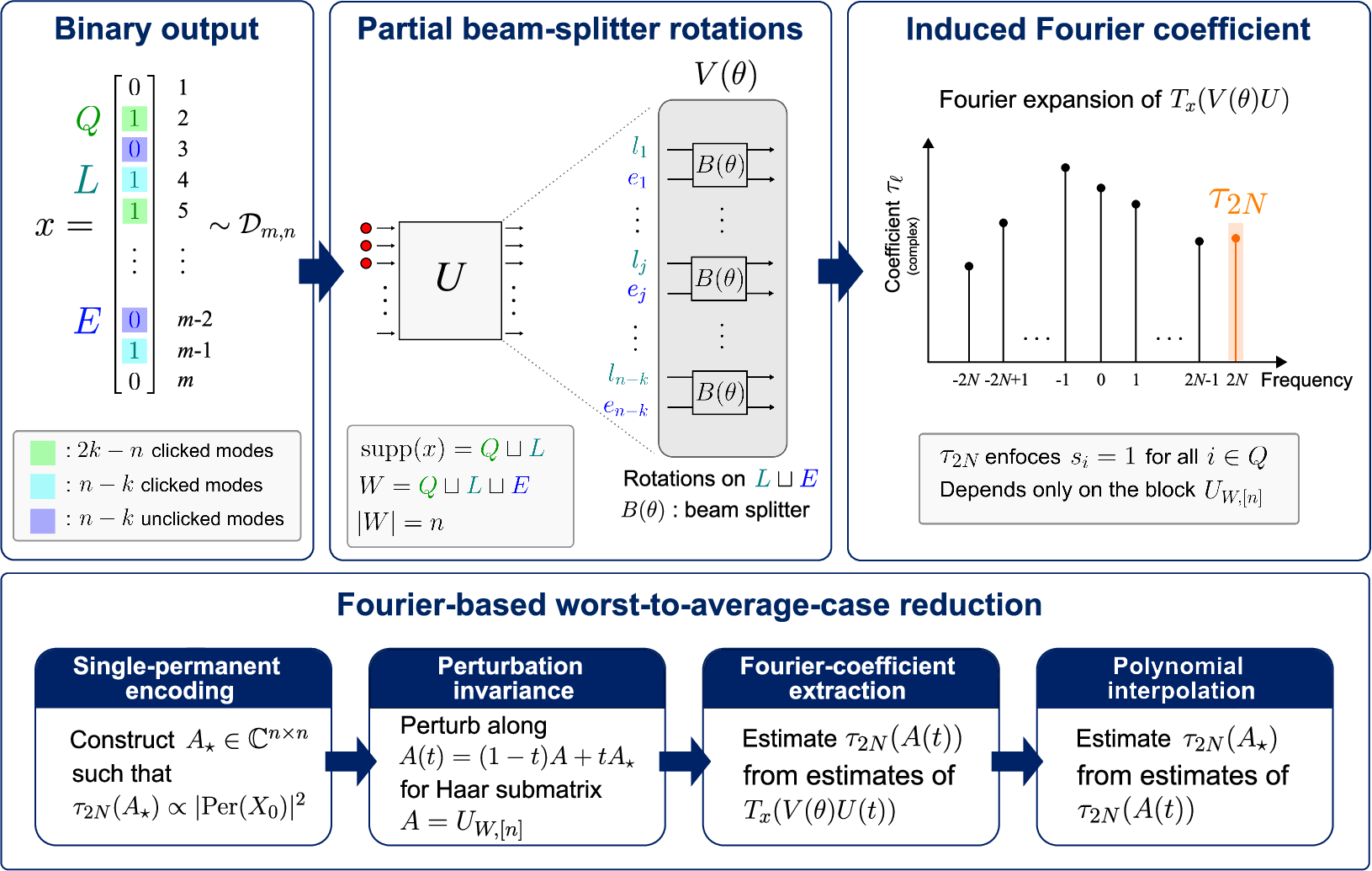}
\caption{Proof overview of Theorem~\ref{thm: threshold average case main}.
Starting from a sampled threshold outcome $x\sim\mathcal D_{m,n}$ with $k=|x|$ clicks, we decompose the clicked set as $\supp(x)=Q\sqcup L$, where $|Q|=2k-n$ and $|L|=n-k$, and choose an unclicked set $E\subseteq[m]\setminus\supp(x)$ with $|E|=n-k$.
Only the modes in $L$ are paired with the unclicked modes in $E$ and rotated by beam splitters $B(\theta)$, while the modes in $Q$ are left unrotated.
The resulting threshold probability is a Fourier series in $\theta$ of degree at most $2N$, where $N=2n-2k$.
Its top Fourier coefficient $\tau_{2N}$ filters out all terms except those in which every unrotated clicked mode $i\in Q$ contains exactly one photon; the remaining photons occupy the clicked modes in $L$, whose rotated rows depend on the paired block $L\sqcup E$.
This single-occupation condition enables the single-permanent encoding.
This partial-rotation construction keeps the relevant square Haar submatrix of size $n\times n$.
The lower row sketches the worst-to-average-case reduction we establish:
construct a worst-case instance $A_\star\in\mathbb{C}^{n\times n}$ such that the corresponding Fourier coefficient $\tau_{2N}$ is proportional to $\Per(X_0)^2$ (Lemma~\ref{lem:tbs-block-permanent-encoding});
use shift-and-scale stability and conditional unitary completion to query instances $A(t)=(1-t)A+tA_\star$ whose distributions remain close to the Haar-submatrix distribution (Lemmas~\ref{lem:tbs-TV} and~\ref{lem:tbs-completion});
extract the Fourier coefficient from estimates of the rotated threshold probabilities (Theorem~\ref{thm:probabilistic-robust-fourier-extraction});
and finally recover the worst-case quantity by polynomial interpolation (Theorem~\ref{thm:probabilistic-robust-BW-extrapolation}).
}
\label{fig: outlineTBS}
\end{figure*}

The threshold-specific feature is that we use a \textit{partial} beam-splitter rotation scheme: for $|x|$ clicked modes, we leave $2|x|-n$ clicked modes unrotated and rotate only the remaining $n-|x|$ clicked modes, each paired with an unclicked mode. 
This choice is designed so that the union of the unrotated clicked modes and all rotated pairs contains exactly $n$ modes, leaving a square $n\times n$ block. When this block is sampled from a Haar-random linear-optical circuit in the worst-to-average-case reduction, the corresponding perturbation margin is $m-2n=\Omega(n)$ for every fixed $\alpha>2$.
The highest-order Fourier coefficient then selects only those contributions in which each of the $2|x|-n$ unrotated clicked modes contains exactly one photon.
This single-occupation constraint is essential for encoding a single worst-case permanent into the corresponding Fourier coefficient.



In the following, we present a more detailed proof outline for the worst-to-average-case reduction underlying Theorem~\ref{thm: threshold average case main}.
The reduction starts from an arbitrary $k_0\times k_0$ binary matrix $X_0$ and shows that oracle access to ATE in Definition~\ref{def: thresholdBS} enables the exact computation of $\Per(X_0)^2$.
Since the exact computation of the permanent of a binary matrix is \#P-hard in the worst case~\cite{valiant1979complexity}, this proves the desired hardness result.

We first establish a typical-outcome property for the threshold ensemble. Specifically, for $x\sim\mathcal D_{m,n}$, the number of clicked modes concentrates near $\alpha n/(\alpha+1)$, as shown in Lemma~\ref{lemma: most threshold outcomes have many clicks}. For fixed $\alpha>2$, a typical outcome therefore has more than $n/2$ clicks. Writing $k=|x|$ and $O=\supp(x)$, we decompose the clicked set as $O=Q\sqcup L$, where $|Q|=2k-n$ and $|L|=n-k$. We then choose a set $E\subseteq[m]\setminus O$ of $n-k$ unclicked modes and fix a pairing between the modes in $L$ and those in $E$. Denoting the union of these paired modes by $R=L\sqcup E$, we further define $W:=Q\sqcup L\sqcup E$. By construction, $|W|=(2k-n)+2(n-k)=n$, so the rows indexed by $W$, together with the first $n$ input columns, define a square $n\times n$ block.

For these fixed mode sets, consider an arbitrary linear-optical circuit $U$. We apply identical coherent beam-splitter rotations to the paired modes in $R=L\sqcup E$, while leaving the modes in $Q$ unchanged, and consider the resulting threshold probability $T_x(V_R(\theta)U)$ as a function of the rotation angle $\theta$. Because every PNR configuration contributing to the threshold outcome $x$ occupies only the clicked modes $Q\sqcup L$, while the rotations mix the rows in $L$ only with their paired rows in $E$, the entire rotated probability depends only on the square block $U_{W,[n]}$.
Moreover, $T_x(V_R(\theta)U)$ is a finite Fourier series of degree at most $2N$, where $N=2n-2k$, because at least one photon must occupy each mode in $Q$, so at most $N$ photons can occupy the rotated clicked modes in $L$. 
The crucial role of the highest-order Fourier coefficient $\tau_{2N}$ is then to filter the coarse-grained sum: it retains only those occupation patterns in which every unrotated clicked mode in $Q$ contains exactly one photon.
This single-occupation constraint is the key property that allows the worst-case permanent to be encoded in the next step.

We now use this single-occupation constraint to encode the worst-case instance $X_0$. Choose a subset $H\subseteq Q$ of $k_0$ modes to carry the rows of $X_0$. Since every mode in $Q$ is singly occupied in each term contributing to $\tau_{2N}$, each row indexed by $H$ appears exactly once in every surviving permanent. We then construct an $n\times n$ block $A_\star$ on the row set $W$ that embeds $X_0$ on the rows indexed by $H$, while its other entries are chosen so that every surviving permanent reduces, up to row permutations, to the permanent of $\kappa\left(X_0\oplus \mathbf{1}_{(n-k_0)\times(n-k_0)}\right)$. Consequently, every surviving permanent equals $\kappa^n\Per(X_0)(n-k_0)!$, independently of the remaining occupation pattern. Summing over these patterns therefore contributes only the explicitly known and efficiently computable factor $G_{n,k}$, so the resulting top (highest-order) Fourier coefficient is proportional to $\Per(X_0)^2$ with known prefactor $\kappa^{2n}((n-k_0)!)^2G_{n,k}$ (Lemma~\ref{lem:tbs-block-permanent-encoding}). The condition $k_0<|Q|$ required by this construction holds for the chosen $k_0=o(n)$ and for typical $x\sim\mathcal D_{m,n}$.



Having encoded the worst-case instance into the block $A_\star$, we now connect this construction to average-case Haar-random instances.
We draw a Haar-random unitary $U \sim \Haar(m)$ and set $A=U_{W,[n]}$ for the selected square block.
We then consider the affine path $A(t)=(1-t)A+tA_\star$.
For sufficiently small $t$, Lemma~\ref{lem:tbs-TV} shows that the distribution of $A(t)$ remains close in total variation distance to the original Haar-submatrix distribution.
Lemma~\ref{lem:tbs-completion} then lifts this block-level closeness to full unitaries by a randomized completion procedure, producing a distribution on unitaries that remains close to $\Haar(m)$.
We denote a unitary sampled from this completed distribution by $U(t)$.

For each interpolation point $t$, we query the ATE oracle on the rotated circuits $V_R(\theta)U(t)$ at several Fourier grid points $\theta$.
Since the distribution of $U(t)$ is close to Haar and left multiplication by $V_R(\theta)$ preserves Haar measure, the oracle's average-case guarantee applies to these queries up to the total-variation distance error.
The robust Fourier-coefficient extraction algorithm of Theorem~\ref{thm:probabilistic-robust-fourier-extraction} then combines these probability estimates to output an additive estimate of $\tau_{2N}(x,U(t),R)$ with high probability.

Because the top Fourier coefficient depends only on the prescribed block $A(t)$, its value is independent of the particular unitary completion. Denoting this coefficient by $p(t)$, we obtain a polynomial in $t$ of degree at most $2n$. Repeating the Fourier-extraction step for many values of $t$ gives noisy evaluations of $p(t)$ near $t=0$. We then use the robust Berlekamp--Welch interpolation algorithm of Theorem~\ref{thm:probabilistic-robust-BW-extrapolation} to extrapolate to $t=1$, where $A(1)=A_\star$ and hence $p(1)$ is proportional to $\Per(X_0)^2$.
After dividing by the known factor $\kappa^{2n}((n-k_0)!)^2G_{n,k}$, the final additive error is chosen to be below $1/2$. Since $\Per(X_0)^2$ is an integer, rounding recovers it exactly, and standard repetition boosts the success probability.

The remainder of this section proceeds as follows.
In Sec.~\ref{subsec: most threshold outcomes}, we show that most threshold outcomes have the number of clicks concentrated near $\alpha n/(\alpha+1)$.
In Sec.~\ref{subsec: threshold fourier}, we introduce our partial beam-splitter rotation scheme for TBS, define the induced threshold Fourier coefficients, establish their band limitation, and derive the top-coefficient single-occupation filter.
In Sec.~\ref{subsec: threshold encoding single permanent}, we encode a worst-case permanent into this coefficient.
In Sec.~\ref{subsec: threshold translation inv}, we specialize the shared shift-and-scale stability and unitary-completion framework to the threshold setting.
In Sec.~\ref{subsec: threshold coeff extract}, we explain how the shared extraction procedures are applied.
Finally, in Sec.~\ref{subsec: proof of threshold average}, we complete the proof of the TBS hardness theorem.

\subsection{Click-count concentration}\label{subsec: most threshold outcomes}

We first observe that, for $x\sim\mathcal{D}_{m,n}$, the number of clicks $|x|$ concentrates near $\alpha n/(\alpha+1)$.

\begin{lemma}[Click-count concentration]\label{lemma: most threshold outcomes have many clicks}
Let $\alpha=m/n \geq 1$, and fix $\eta\in(0,1/(\alpha + 1))$.
Let $\mathcal{D}_{m,n}$ be the distribution in Definition~\ref{def: random threshold outcome ensemble}.
Then there exist constants $C_{\alpha,\eta},c_{\alpha,\eta}>0$ such that
\begin{align}
\Pr_{x\sim\mathcal{D}_{m,n}}\left[
\left| |x|-\frac{\alpha}{\alpha+1}n\right|\ge \eta n
\right]
\le C_{\alpha,\eta}e^{-c_{\alpha,\eta}n}.
\end{align}
\end{lemma}

\begin{proof}
    See Appendix~\ref{appendix: section: click count concentration}.
\end{proof}

In what follows, we fix $\alpha>2$ and choose $\eta>0$ small enough so that
\begin{align}
    \frac{\alpha}{\alpha+1}-\eta>\frac12 .
\label{eq:tbs-click-more-than-half}
\end{align}
Then Lemma~\ref{lemma: most threshold outcomes have many clicks} implies that, with probability at least $1-\exp(-\Theta(n))$, a sampled threshold outcome $x\sim\mathcal{D}_{m,n}$ satisfies $|x|>n/2$.
Writing $k=|x|$, this ensures that both $2k-n$ and $n-k$ are nonnegative, which is the parameter range required for the partial-rotation construction below.

\subsection{Partial rotations, threshold Fourier coefficients, and band limitation}\label{subsec: threshold fourier}

Let $x\in\mathcal{T}_{m,n}$ be a fixed threshold outcome with $k:=|x|>n/2$, and let $O:=\supp(x)$ be the set of clicked modes.
We choose a decomposition
\begin{align}\label{eq: def of O Q L}
    O=Q\sqcup L,
    \qquad
    |Q|= 2k - n,
    \qquad
    |L|= n - k.
\end{align}
For typical outcomes satisfying the concentration event of Lemma~\ref{lemma: most threshold outcomes have many clicks}, both $|Q|=2k-n$ and $|L|=n-k$ are $\Theta(n)$.
Modes in $Q$ will be left unrotated, while the modes in $L$ will be paired with unclicked modes.
We choose an arbitrary set
\begin{align}
    E\subseteq [m]\setminus O,
    \qquad
    |E|= n-k,
\end{align}
and fix a pairing
\begin{align}\label{eq: def of R for TBS}
    R:=L\sqcup E
    =\{l_1,e_1\}\sqcup\cdots\sqcup\{l_{n-k},e_{n-k}\},
\end{align}
for $l_j\in L$ and $e_j\in E$.
Finally, we define the mode-index set
\begin{align}\label{eq: def of W}
    W:=Q\sqcup L\sqcup E  ,
\end{align}
whose size is given by $|W| = (2k-n)+2(n-k) = n$.
Since every $S$ contributing to $T_x(U)$ satisfies $\supp(S)=O=Q\sqcup L$, and $V_R(\theta)$ mixes the rows indexed by $L$ only with the rows indexed by $E$, the full rotated probability $T_x(V_R(\theta)U)$ depends only on the $n\times n$ submatrix $U_{W,[n]}$ for every $\theta$.
Thus, this selected-block dependence is already present before taking any particular Fourier coefficient.

We now apply the beam-splitter rotation $V_R(\theta)$ defined in Eq.~\eqref{eq:V-decomp-full}, with the rotated pairs given by $(l_j,e_j)$ as defined in Eq.~\eqref{eq: def of R for TBS}.
This gives rise to the following Fourier coefficient for TBS.

\begin{definition}[Threshold Fourier coefficient]\label{def: threshold fourier coeff}
Given $x\in\mathcal{T}_{m,n}$, $U\in\mathrm{U}(m)$, and the paired set $R \subseteq [m]$ as above, define
\begin{align}
    \tau_\ell(x,U,R)
    :=
    \frac{1}{2\pi}\int_{0}^{2\pi}T_x(V_R(\theta)U)e^{-i\ell\theta}d\theta,
    \qquad \ell\in\mathbb{Z}.
\label{eq:tbs-fourier-def}
\end{align}
\end{definition}

\subsubsection{Band limitation}

We first show that the rotated threshold probability is band-limited: its Fourier support lies in $\{-2N,\ldots,2N\}$, where $N:=n-|Q|=2n-2k$.

\begin{lemma}\label{lem:tbs-bandlimit}
Let $x\in\mathcal{T}_{m,n}$ be a threshold outcome with $k:=|x|>n/2$, let $U\in\mathrm{U}(m)$, and let $O=\supp(x)$.
Let $Q$, $L$, $E$, $R$, and $W$ denote the mode-index sets defined in Eqs.~\eqref{eq: def of O Q L}--\eqref{eq: def of W}, and define $N:=n-|Q|=2n-2k$.
Then $T_x(V_R(\theta)U)$ is a trigonometric polynomial in $\theta$ of degree at most $2N$.
Equivalently,
\begin{align}
T_x(V_R(\theta)U)
=
\sum_{\ell=-2N}^{2N}\tau_\ell(x,U,R)e^{i\ell\theta},
\label{eq:tbs-finite-fourier}
\end{align}
where the coefficients are given in Definition~\ref{def: threshold fourier coeff}.
\end{lemma}

\begin{proof}
Write $U_R(\theta):=V_R(\theta)U$.
By the definition of threshold output probability,
\begin{align}
T_x(U_R(\theta))
=
\sum_{\substack{S\in\mathcal{S}_{m,n}\\ \supp(S)=O}}
\frac{1}{\prod_i s_i!}
\left|\Per\left(U_R(\theta)_{\nu(S),[n]}\right)\right|^2.
\end{align}
For a fixed $S$ with $\supp(S)=O$, define $n_L(S):=\sum_{i\in L}s_i$ as the number of photons occupying the output modes indexed by $L$.
Since every mode in $Q$ belongs to $\supp(S)$, each such mode contains at least one photon.
Therefore, we have
\begin{align}
    n_L(S)\le n-|Q|=N.
\label{eq:tbs-nL-upper}
\end{align}

By construction, rows indexed by $Q$ are not affected by $V_R(\theta)$, whereas rows indexed by $j\in L$ have the form
\begin{align}\label{eq: row of rotated circuit TBS}
(U_R(\theta))_j
= e^{i\theta}\,(V_R^+U)_j + e^{-i\theta}\,(V_R^-U)_j.
\end{align}
Thus, by multilinearity of the permanent, $\Per(U_R(\theta)_{\nu(S),[n]})$ is a trigonometric polynomial in $\theta$ of degree at most $n_L(S)$, and its squared magnitude has degree at most $2n_L(S)$.
By Eq.~\eqref{eq:tbs-nL-upper}, $T_x(U_R(\theta))$ is therefore a trigonometric polynomial in $\theta$ of degree at most $2N$.
Equivalently, there exist complex coefficients $a_{-2N},\ldots,a_{2N}\in\mathbb{C}$ such that
\begin{align}
T_x(U_R(\theta))
=
\sum_{\ell=-2N}^{2N}a_\ell e^{i\ell\theta},
\label{eq:tbs-finite-fourier-a}
\end{align}
for all $\theta\in[0,2\pi)$.
Here, the coefficients $a_\ell$ coincide with the Fourier coefficients $\tau_\ell(x,U,R)$ defined in Eq.~\eqref{eq:tbs-fourier-def}, as follows immediately from the orthogonality of the Fourier basis.
\end{proof}

\subsubsection{Top Fourier coefficient and single-occupation filter}

We focus on the highest-frequency coefficient
\begin{align}
\tau_{2N}(x,U,R)
=
\frac{1}{2\pi}\int_0^{2\pi}
T_x(V_R(\theta)U)e^{-i2N\theta}\,d\theta,
\label{eq:tbs-top-fourier-def}
\end{align}
which we also refer to as the top Fourier coefficient.
By Lemma~\ref{lem:tbs-bandlimit}, this coefficient saturates the bound $n_L(S)\le N$ and thereby forces single occupation on every unrotated clicked mode in $Q$.

\begin{lemma}\label{lem:tbs-rowfilter}
Let $x\in\mathcal{T}_{m,n}$ be a threshold outcome with $k:=|x|>n/2$, let $U\in\mathrm{U}(m)$, and let $O=\supp(x)$.
Let $Q$, $L$, $E$, $R$, and $W$ denote the mode-index sets defined in Eqs.~\eqref{eq: def of O Q L}--\eqref{eq: def of W}, and define $N:=n-|Q|=2n-2k$.
Then the top Fourier coefficient satisfies
\begin{align}
\tau_{2N}(x,U,R)
&=
\!\!\!\!\!\!\sum_{\substack{S\in\mathcal{S}_{m,n}\\ \supp(S)=O\\ s_i=1\;\forall i\in Q}}
\!\!\!\!\frac{1}{\prod_i s_i!}
\Per\left(\left((\Pi_Q+V_R^+)U\right)_{\nu(S),[n]}\right)
\nonumber \\
&\qquad\qquad\quad\times \Per\left(\left((\Pi_Q+V_R^-)U\right)_{\nu(S),[n]}\right)^*  ,
\label{eq:tbs-top-fourier-final}
\end{align}
where $\Pi_Q$ denotes the projector matrix onto the modes in $Q$.
\end{lemma}

The use of the top Fourier coefficient $\tau_{2N}(x,U,R)$ is crucial: lower-frequency coefficients also receive contributions from occupation patterns satisfying $n_L(S)<N$, for which the modes in $Q$ need not all be singly occupied.
Hence, the primary TBS-specific consequence of the top coefficient is the single-occupation condition $s_i=1$ for every $i\in Q$.
In the encoding below, we choose a subset $H\subseteq Q$ to carry the hard matrix $X_0$; this condition ensures that every row indexed by $H$ appears exactly once in each contributing permanent.


\begin{proof}[Proof of Lemma~\ref{lem:tbs-rowfilter}]
As in the proof of Lemma~\ref{lem:tbs-bandlimit}, write $U_R(\theta):=V_R(\theta)U$ and let $n_L(S):=\sum_{i\in L}s_i$.
By Eq.~\eqref{eq:tbs-nL-upper}, every $S$ contributing to the threshold probability satisfies $n_L(S)\le N$.
We now extract the coefficient of $e^{i2N\theta}$, which filters out all terms except those with $n_L(S)=N$.
Since $S$ has total photon number $n$ and every mode in $Q$ is occupied at least once, the equality $n_L(S)=N=n-|Q|$ is equivalent to
\begin{align}
    s_i=1
    \qquad
    \forall i\in Q  ,
\label{eq:tbs-Q-forced-one}
\end{align}
and all terms violating Eq.~\eqref{eq:tbs-Q-forced-one} vanish in the top coefficient.

For $S \in \mathcal{S}_{m,n}$ satisfying Eq.~\eqref{eq:tbs-Q-forced-one}, the largest exponent $e^{iN\theta}$ in $\Per(U_R(\theta)_{\nu(S),[n]})$ is obtained by choosing the $e^{i\theta}$ contribution from every row indexed by $L$ in Eq.~\eqref{eq: row of rotated circuit TBS}, while retaining the unchanged rows indexed by $Q$.
Therefore, the coefficient of $e^{iN\theta}$ is $\Per\left(\left((\Pi_Q+V_R^+)U\right)_{\nu(S),[n]}\right)$.
Similarly, the coefficient of the smallest exponent $e^{-iN\theta}$ is $\Per\left(\left((\Pi_Q+V_R^-)U\right)_{\nu(S),[n]}\right)$.
Therefore, the coefficient of $e^{i2N\theta}$ in the squared magnitude is the product of the former coefficient and the complex conjugate of the latter.
Integrating against $e^{-i2N\theta}$, as in Eq.~\eqref{eq:tbs-top-fourier-def}, yields Eq.~\eqref{eq:tbs-top-fourier-final}.
\end{proof}

Since the expression in Eq.~\eqref{eq:tbs-top-fourier-final} depends only on the rows indexed by $W=Q\sqcup L\sqcup E$, or equivalently, only on the $n\times n$ submatrix $U_{W,[n]}$, we next extend the definition of the top Fourier coefficient from $n\times n$ unitary submatrices to arbitrary $n\times n$ matrix inputs.
This extension permits its evaluation on the encoded block used in the worst-case construction.
Concretely, for an arbitrary matrix $A\in\mathbb{C}^{n\times n}$, let $\widetilde A\in\mathbb{C}^{m\times n}$ be any completion satisfying $\widetilde A_{W,[n]}=A$.
Since $\Pi_Q+V_R^\pm$ is supported only on the modes indexed by $W$, whenever $A=U_{W,[n]}$ we have
\begin{align}
\left((\Pi_Q+V_R^\pm)U\right)_{[m],[n]}
=
(\Pi_Q+V_R^\pm)\widetilde A.
\label{eq:tbs-block-support-identity}
\end{align}
This identity also shows that the following quantity is independent of the chosen completion:
\begin{align}
\tau_{2N}^{\mathrm{blk}}(x,A,R)
&:=
\!\!\!\!\!\!\sum_{\substack{S\in\mathcal{S}_{m,n}\\ \supp(S)=O\\ s_i=1\;\forall i\in Q}}
\!\!\!\!\frac{1}{\prod_i s_i!}
\Per\left(\left((\Pi_Q+V_R^+)\widetilde A\right)_{\nu(S),[n]}\right)
\nonumber \\
&\qquad\qquad\quad \times \Per\left(\left((\Pi_Q+V_R^-)\widetilde A\right)_{\nu(S),[n]}\right)^* .
\label{eq:tbs-block-top-coefficient}
\end{align}
Accordingly, whenever $A=U_{W,[n]}$ for a unitary $U$, Lemma~\ref{lem:tbs-rowfilter} gives
\begin{align}
    \tau_{2N}^{\mathrm{blk}}(x,U_{W,[n]},R)=\tau_{2N}(x,U,R).
\label{eq:tbs-block-equals-fourier}
\end{align}

\subsection{Encoding a single permanent into the threshold Fourier coefficient}\label{subsec: threshold encoding single permanent}

We now encode the worst-case quantity $\Per(X_0)^2$ into the block-extended coefficient in Eq.~\eqref{eq:tbs-block-top-coefficient}.
The single-occupation condition in Lemma~\ref{lem:tbs-rowfilter} is the key threshold-specific property used in the encoding.
We choose a subset $H\subseteq Q$ to carry the hard matrix $X_0$; since every mode in $Q$ is singly occupied in each term contributing to the top Fourier coefficient, every row indexed by $H$ appears exactly once in each contributing permanent.
This enables the block-diagonal permanent collapse used below.
The block extension then allows the corresponding selected $n\times n$ block to be replaced by a worst-case matrix $A_\star$ while keeping the coefficient well-defined.
The hard instance is a binary matrix $X_0\in\{0,1\}^{k_0\times k_0}$ with $k_0=o(n)$, whose squared permanent is \#P-hard to compute exactly in the worst case~\cite{valiant1979complexity}.

\begin{lemma}[Single-permanent encoding for TBS]\label{lem:tbs-block-permanent-encoding}
Let $x\in\mathcal{T}_{m,n}$ with $k=|x|>n/2$, let $O = \supp(x)$, and let $Q,L,E,R$, and $W$ denote the mode-index sets defined in Eqs.~\eqref{eq: def of O Q L}--\eqref{eq: def of W}.
Let $\tau^{\mathrm{blk}}_{2N}$ be the block extension of the top Fourier coefficient defined in Eq.~\eqref{eq:tbs-block-top-coefficient}.
Given a matrix $X_0\in\{0,1\}^{k_0\times k_0}$ with $k_0 < |Q| = 2k-n$, there exists an $n\times n$ matrix $A_\star\in\mathbb{C}^{n\times n}$ that can be constructed in polynomial time and satisfies
\begin{align}
\tau_{2N}^{\mathrm{blk}}(x,A_\star,R)
=
\kappa^{2n}\left((n-k_0)!\right)^2G_{n,k}\Per(X_0)^2,
\label{eq:tbs-encoding-main}
\end{align}
where $\kappa>0$ is a free normalization parameter to be chosen later, and the factor
\begin{align}
G_{n,k}
:=
\sum_{\substack{r_1,\ldots,r_{n-k}\ge 1\\ r_1+\cdots+r_{n-k}=2(n-k)}}
\prod_{j=1}^{n-k}\frac{1}{r_j!} ,
\label{eq:tbs-Gs-def}
\end{align}
which can be computed in polynomial time.
\end{lemma}

\begin{figure*}[t]
\includegraphics[width=0.85\linewidth]{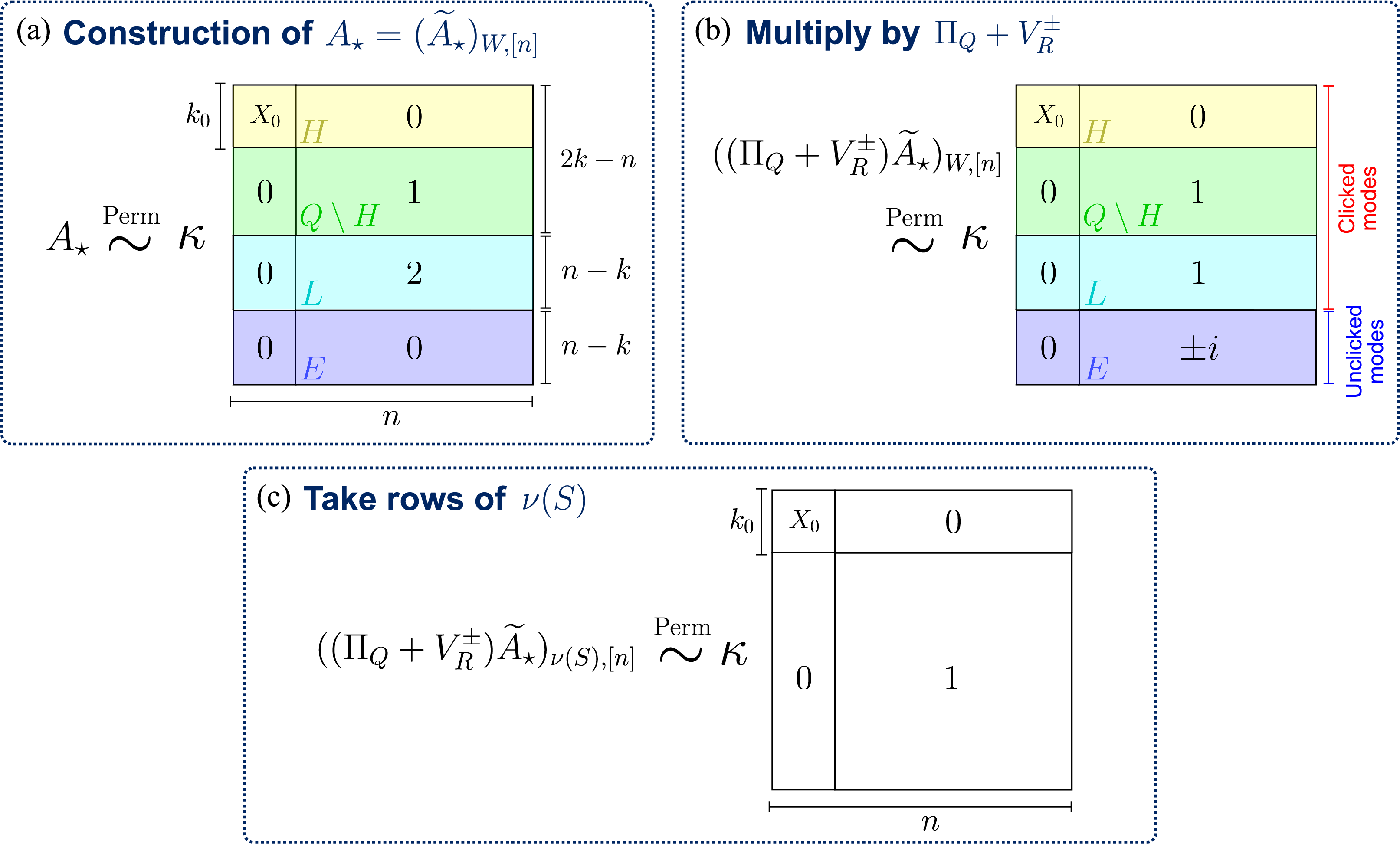}
\caption{
Illustration of the matrix constructions used in the proof of Lemma~\ref{lem:tbs-block-permanent-encoding}.
(a) Schematic of the matrix $A_\star=(\widetilde A_\star)_{W,[n]}$ in Eq.~\eqref{eq:tbs-Astar-tilde},
constructed from a worst-case matrix $X_0\in\{0,1\}^{k_0\times k_0}$ with $k_0<|Q|=2k-n$
by adjoining all-zero blocks, an all-ones block, and an all-twos block, together with a normalization factor $\kappa>0$.
The notation $\overset{\mathrm{Perm}}{\sim}$ denotes equivalence up to row permutations.
The row indices corresponding to $H$, $Q\setminus H$, $L$, and $E$ are indicated by different colors.
(b) Schematic of the matrices $((\Pi_Q+V_R^\pm)\widetilde A_\star )_{W,[n]}$ in Eq.~\eqref{eq:tbs-Astar-transformed}.
(c) Schematic of the matrix $((\Pi_Q+V_R^\pm)\widetilde A_\star )_{\nu(S),[n]}$ in Eq.~\eqref{eq:tbs-Astar-permanent-collapse}, obtained by imposing the constraint $\supp(S) = \supp(x)$ and $s_i=1\;\forall i\in Q$ to (b).
}
\label{fig: encoding TBS}
\end{figure*}

\begin{proof}

We now construct the matrix $A_\star$; an illustration of the construction is provided in Fig.~\ref{fig: encoding TBS}.

We choose an arbitrary subset $H\subseteq Q$ of size $|H|=k_0$ to encode the hard instance $X_0$ (see Fig.~\ref{fig: encoding TBS}(a)).
We define the $m\times n$ completion $\widetilde A_\star$ by specifying the following row blocks, setting all unspecified entries to zero:
\begin{align}\label{eq:tbs-Astar-tilde}
\begin{split}
    &(\widetilde A_\star)_{H,\{1,\ldots,k_0\}}=\kappa X_0,
    \qquad \\
    &(\widetilde A_\star)_{Q\setminus H,\{k_0+1,\ldots,n\}}=\kappa 1_{(2k-n-k_0)\times(n-k_0)},
    \qquad \\
    &(\widetilde A_\star)_{L,\{k_0+1,\ldots,n\}}=2\kappa 1_{(n-k)\times(n-k_0)} , 
\end{split}
\end{align}
where $1_{a \times b}$ denotes the $a \times b$ all-one matrix. 
This construction specifies $A_\star=(\widetilde A_\star)_{W,[n]}$.

We next apply the matrix $\Pi_Q+V_R^\pm$ to $\widetilde A_\star$.
The rows indexed by $Q$ remain unchanged because of the projector $\Pi_Q$, whereas the rows indexed by $R=L\sqcup E$ are transformed according to $V_R^\pm$ in Eq.~\eqref{eq:Vpm-full}.
Thus, we obtain
\begin{align}\label{eq:tbs-Astar-transformed}
\begin{split}
&\left((\Pi_Q+V_R^\pm)\widetilde A_\star\right)_{H,\{1,\ldots,k_0\}}
=\kappa X_0,
\quad \\
&\left((\Pi_Q+V_R^\pm)\widetilde A_\star\right)_{Q\setminus H,\{k_0+1,\ldots,n\}}
=\kappa 1_{(2k-n-k_0)\times(n-k_0)},
\quad \\
&\left((\Pi_Q+V_R^\pm)\widetilde A_\star\right)_{L,\{k_0+1,\ldots,n\}}
=\kappa 1_{(n-k)\times(n-k_0)}
\quad \\
&\left((\Pi_Q+V_R^\pm)\widetilde A_\star\right)_{E,\{k_0+1,\ldots,n\}}
= \pm i \kappa 1_{(n-k)\times(n-k_0)}  ,
\end{split}
\end{align}
where all unspecified entries are set to zero.

Now consider any $S$ contributing to Eq.~\eqref{eq:tbs-block-top-coefficient}; such an $S$ satisfies $\supp(S)=O$ and $s_i=1$ for all $i\in Q$.
For the modes $l_j \in L$, write
\begin{align}\label{eq: def of r_j}
    s_{l_j}=r_j, \qquad j=1,\ldots,n-k.
\end{align}
Because the clicked modes in $\supp(S)=O=Q\sqcup L$ have at least one photon, we have $r_j\ge 1$ for every $j$.
Moreover, because the total number of photons outside $Q$ is $N= 2n - 2k$, the occupations satisfy
\begin{align}
    r_1+\cdots+r_{n-k}=2n-2k.
\label{eq:tbs-rj-constraint}
\end{align}
For such an $S$, the rows indexed by $H$ appear exactly once and are supported only on the columns indexed by $\{1,\ldots,k_0\}$.
All other rows appearing in the permanent, namely rows from $Q\setminus H$ and repeated rows from $L$, are supported only on the columns indexed by $\{k_0+1,\ldots,n\}$ (note that $E$ rows do not contribute to the permanent).
The number of these rows is
\begin{align}
    |Q\setminus H|+\sum_{j=1}^{n-k}r_j
    =(2k - n -k_0)+2n-2k
    =n-k_0.
\end{align}
Hence, up to row permutations, the relevant matrix is the block diagonal $\kappa (X_0 \oplus 1_{(n-k_0) \times (n-k_0)} )$, and thus
\begin{align}
\Per\left(\left((\Pi_Q+V_R^\pm)\widetilde A_\star\right)_{\nu(S),[n]}\right)
=
\kappa^n\Per(X_0)(n-k_0)!   .
\label{eq:tbs-Astar-permanent-collapse}
\end{align}

We now evaluate the photon-number factorial factor in Eq.~\eqref{eq:tbs-block-top-coefficient}.
Since all modes in $Q$ contain exactly one photon, they contribute a factor of $1$.
Hence, the factorial factor depends only on the occupation numbers $r_j$ in Eq.~\eqref{eq: def of r_j}, and is given by
\begin{align}
    \prod_i s_i! = \prod_{j=1}^{n-k}r_j!.
\end{align}
Summing over all $r_1,\ldots,r_{n-k}$ satisfying Eq.~\eqref{eq:tbs-rj-constraint} gives exactly the factor $G_{n,k}$ in Eq.~\eqref{eq:tbs-Gs-def}, yielding Eq.~\eqref{eq:tbs-encoding-main}.

Lastly, one can rewrite $G_{n,k}$ by
\begin{align}
G_{n,k}
&=
[z^{2(n-k)}](e^z-1)^{n-k} \\
&=
\frac{1}{(2(n-k))!}
\sum_{j=0}^{n-k}
(-1)^{n-k-j}
\binom{n-k}{j}
j^{2(n-k)},
\end{align}
which can be evaluated exactly in polynomial time.
\end{proof}

The next lemma shows that the multiplicative factor $G_{n,k}$ in Lemma~\ref{lem:tbs-block-permanent-encoding} is lower bounded by $\exp(-O(n))$.

\begin{lemma}\label{lem:tbs-Gs-lower-bound}
The multiplicative factor $G_{n,k}$ defined in Eq.~\eqref{eq:tbs-Gs-def} of Lemma~\ref{lem:tbs-block-permanent-encoding} is lower bounded by
\begin{align}
    G_{n,k}\ge 2^{-n+k}=e^{-O(n)}.
\end{align}
\end{lemma}

\begin{proof}
Since all summands in $G_{n,k}$ are nonnegative, we may lower bound $G_{n,k}$ by restricting the summation to the single occupation pattern
\begin{align}
    r_1=r_2=\cdots=r_{n-k}=2.
\end{align}
This pattern satisfies $r_1+\cdots+r_{n-k}= 2n - 2k$, and contributes
\begin{align}
    \prod_{j=1}^{n-k}\frac{1}{2!}=2^{-n+k}.
\end{align}
Therefore, $G_{n,k}\ge 2^{-n+k}$.
\end{proof}

\subsection{Shift-and-scale stability and conditional unitary completion}\label{subsec: threshold translation inv}

In the proof of Theorem~\ref{thm: threshold average case main}, we interpolate between the hard block $A_\star$ constructed above and a Haar-random unitary submatrix.
Let $A:=U_{W,[n]}$ for $U\sim\Haar(m)$, and set $A(t):=(1-t)A+tA_\star$.
We denote by $\mu$ the distribution of $A$, and by $\mu_t$ the distribution of $A(t)$.
The following lemma establishes the required shift-and-scale stability bound for square submatrices of Haar-random unitaries.


\begin{lemma}[Shift-and-scale stability for TBS]\label{lem:tbs-TV}
Let $\mu$ be the distribution of $A=U_{W,[n]}$ for $U\sim\Haar(m)$ and $|W| = n$, and let $\mu_t$ be the distribution of $A(t)=(1-t)A+tA_\star$ for $A_\star$ constructed in Lemma~\ref{lem:tbs-block-permanent-encoding}.
Suppose that $m-2n=\Omega(n)$.
There exist constants $c_0,c_1,c_2>0$ such that, whenever
\begin{align}
    0\le t\le \frac{c_0}{mn(1+\kappa\sqrt n)},
\label{eq:tbs-t-bound}
\end{align}
the total variation distance between the distributions $\mu$ and $\mu_t$ is bounded by
\begin{align}
    \|\mu_t-\mu\|_{\mathrm{TVD}}
    \le
    c_1mn(1+\kappa\sqrt n)t+e^{-c_2(m-2n)}.
\label{eq:tbs-TV-bound}
\end{align}
\end{lemma}

\begin{proof}

See Appendix~\ref{appendix: section: translation invariance}.

\end{proof}

The block-level perturbation must also be lifted to the full unitary group, because an instance of the ATE problem in Definition~\ref{def: thresholdBS} is specified by a complete circuit unitary.
To this end, we introduce a conditional unitary-completion procedure that extends the perturbed submatrix to a full unitary while preserving the corresponding total-variation-distance guarantee.

\begin{lemma}[Conditional unitary completion for TBS]\label{lem:tbs-completion}
Let $\mu$ be the distribution of the block $A=U_{W,[n]}$ under $U\sim\Haar(m)$, and let $\mu_t$ be the distribution of the block $A(t) = (1-t)A + tA_\star$.
Suppose that $m-2n=\Omega(n)$ and $t$ satisfy Eq.~\eqref{eq:tbs-t-bound}.
There exists a randomized polynomial-time procedure $\mathcal{C}$ that, with probability at least $1-\exp(-\Omega(m-2n))$, completes the input matrix to a unitary containing it as the prescribed submatrix.
Moreover, the corresponding pushforward measure $\mathcal H_t:=\mathcal C_{\#}\mu_t$ satisfies
\begin{align}
    \|\mathcal H_t-\Haar(m)\|_{\mathrm{TVD}}
    \le
    \|\mu_t-\mu\|_{\mathrm{TVD}}+e^{-\Theta(m-2n)}.
\label{eq:tbs-completion-bound}
\end{align}
\end{lemma}

\begin{proof}

See Appendix~\ref{appendix: section: conditional-haar-completion}.

\end{proof}

\subsection{Application of the extraction procedures}\label{subsec: threshold coeff extract}

We briefly describe how the two extraction procedures established in Theorems~\ref{thm:probabilistic-robust-BW-extrapolation} and~\ref{thm:probabilistic-robust-fourier-extraction}, developed in Secs.~\ref{subsec: robust polynomial interpolation} and~\ref{subsec: coeff extract}, respectively, are combined to recover the encoded worst-case quantity from average-case estimates of threshold probabilities.
Fix a Haar-random block $A=U_{W,[n]}$ with $U\sim\Haar(m)$, and let $A_\star$ be the block constructed in Lemma~\ref{lem:tbs-block-permanent-encoding}.
For each sufficiently small $t\in[0,\Delta]$, where $\Delta$ is specified below, define $A(t)=(1-t)A+tA_\star$, and let $U(t)$ be a unitary sampled according to the completion procedure of Lemma~\ref{lem:tbs-completion} applied to $A(t)$.
From estimates of $T_x(V_R(\theta)U(t))$, we first estimate the top Fourier coefficient $\tau_{2N}(x,U(t),R)=\tau_{2N}^{\mathrm{blk}}(x,A(t),R)$ using Theorem~\ref{thm:probabilistic-robust-fourier-extraction}.
Here, the trigonometric degree is at most $2N\le 2n$, by Lemma~\ref{lem:tbs-bandlimit}.
By Eq.~\eqref{eq:tbs-block-top-coefficient}, the quantity $p(t):=\tau_{2N}^{\mathrm{blk}}(x,A(t),R)$ is a polynomial in $t$ of degree at most $2n$, and Lemma~\ref{lem:tbs-block-permanent-encoding} shows that its value at $t=1$ is proportional to $\Per(X_0)^2$.
The robust Berlekamp--Welch algorithm in Theorem~\ref{thm:probabilistic-robust-BW-extrapolation} is then applied to the resulting estimates over $t\in[0,\Delta]$ to recover $p(1)=\tau_{2N}^{\mathrm{blk}}(x,A_\star,R)$.
Thus, the application proceeds as $A(t)\to U(t)\to T_x(V_R(\theta)U(t))\to\tau_{2N}(x,U(t),R)\to p(t)\to p(1)$, with the error and failure bounds supplied in the proof below.

\subsection{Proof of Theorem~\ref{thm: threshold average case main}}\label{subsec: proof of threshold average}

We are now ready to prove Theorem~\ref{thm: threshold average case main}, which establishes the average-case hardness of TBS.
More explicitly, we prove the following quantitative statement.

\begin{theorem}\label{thm: threshold average case quantitative}
Let $m=\alpha n$ for a fixed constant $\alpha>2$.
Then {\rm ATE} problem in Definition~\ref{def: thresholdBS} is $\#\mathrm{P}$-hard under a $\rm{BPP^{NP}}$ reduction, for additive imprecision
\begin{align}
\varepsilon \leq \exp(-3n\log n -  n\log \alpha - O(n)),
\end{align}
and failure probability $\delta<1/8$.
\end{theorem}

\begin{proof}
Let $\mathcal O$ be an oracle for ATE such that, on input $x\in\mathcal{T}_{m,n}$ and $U\in\mathrm{U}(m)$, $\mathcal O$ outputs an estimate satisfying
\begin{align}
\Pr_{x,U}\left[
|\mathcal O(x,U)-T_x(U)|>\varepsilon T_x^{(\mathrm{avg})}
\right]
<\delta,
\label{eq:tbs-oracle}
\end{align}
where the probability is taken over $x\sim\mathcal{D}_{m,n}$ in Definition~\ref{def: random threshold outcome ensemble} and $U\sim\Haar(m)$.
We prove that oracle access to $\mathcal{O}$ suffices to compute $\Per(X_0)^2$ exactly for any binary matrix $X_0\in\{0,1\}^{k_0\times k_0}$.

Given $X_0\in\{0,1\}^{k_0\times k_0}$, choose a constant $0<\lambda<1$ and set $ n=\left\lceil k_0^{1/\lambda} \right\rceil$, so that $k_0=o(n)$.
Let $m=\alpha n $ for a fixed $\alpha>2$.
We sample $x\sim\mathcal{D}_{m,n}$ and set $k=|x|$.
By Lemma~\ref{lemma: most threshold outcomes have many clicks}, and by choosing $\eta>0$ small enough so that Eq.~\eqref{eq:tbs-click-more-than-half} holds, the event
\begin{align}
\left|k-\frac{\alpha}{\alpha+1}n\right|\le \eta n
\label{eq:tbs-typical-k}
\end{align}
occurs with probability at least $1-\exp(-\Theta(n))$.
Conditioned on this event, we have $k>n/2$, and both $2k-n=\Theta(n)$ and $n-k=\Theta(n)$ hold.
In particular, for all sufficiently large $n$ we have $2k - n > k_0$, as required for the construction in Lemma~\ref{lem:tbs-block-permanent-encoding}.

Conditioned on the event in Eq.~\eqref{eq:tbs-typical-k}, let $O=\supp(x)$, and choose the mode-index sets $Q,L,E,R$, and $W$ according to Eqs.~\eqref{eq: def of O Q L}--\eqref{eq: def of W}.
Given $X_0$, we construct $A_\star$ as in Lemma~\ref{lem:tbs-block-permanent-encoding}, while keeping $\kappa$ unspecified for now.
We then sample $U\sim\Haar(m)$, set $A:=U_{W,[n]}$, define $A(t):=(1-t)A+tA_\star$, and consider the polynomial
\begin{align}
    p(t):=\tau_{2N}^{\mathrm{blk}}(x,A(t),R).
\label{eq:tbs-proof-p-def}
\end{align}
By Eq.~\eqref{eq:tbs-block-top-coefficient}, $p(t)$ is a polynomial in $t$ of degree at most $2n$.
Furthermore, Lemma~\ref{lem:tbs-block-permanent-encoding} implies
\begin{align}
    p(1)
    =
    \kappa^{2n}\left((n-k_0)!\right)^2G_{n,k}\Per(X_0)^2.
\label{eq:tbs-proof-p1}
\end{align}

We next explain how to estimate $p(t)$ for a fixed $t\in[0,\Delta]$.
Given $A(t)$, we sample a unitary $U(t)\in\mathrm{U}(m)$ satisfying
\begin{align}\label{eq:tbs-completion}
    U(t)_{W,[n]}=A(t)
\end{align}
by the completion procedure in Lemma~\ref{lem:tbs-completion}, which succeeds with probability at least $1 - \exp(-\Omega(n))$.
Let $\mathcal H_t$ denote the distribution of the completed unitary.

Choose $\Delta:={c_\delta}/{mn(1+\kappa\sqrt n)}$ for a sufficiently small constant $c_\delta>0$.
Since $m-2n=(\alpha-2)n=\Omega(n)$, Lemma~\ref{lem:tbs-TV} and Lemma~\ref{lem:tbs-completion} imply that, for every $t\in[0,\Delta]$,
\begin{align}
\|\mathcal H_t-\Haar(m)\|_{\mathrm{TVD}}
\le
\frac{1/8-\delta}{2}+e^{-\Theta(n)}.
\label{eq:tbs-Ht-close}
\end{align}
Combining all the preceding arguments, for all sufficiently large $n$, the failure probability for estimating $T_x(U(t))$ with additive imprecision at most $\varepsilon T_x^{(\mathrm{avg})}$ is bounded by
\begin{align}
&\delta+e^{-\Theta(n)}+e^{-\Omega(n)} +\frac{1/8-\delta}{2}+e^{-\Theta(n)}
<\frac18.
\end{align}
Here, $\delta$ is the oracle failure probability in Eq.~\eqref{eq:tbs-oracle}, the first exponentially small term is the failure probability of the typical event in Eq.~\eqref{eq:tbs-typical-k}, and the second is the failure probability of the completion procedure associated with Eq.~\eqref{eq:tbs-completion}.
The last two terms are respectively the block-to-Haar total variation contribution and the boundary contribution in Eq.~\eqref{eq:tbs-Ht-close}.
Therefore, for all sufficiently large $n$,
\begin{align}
\Pr\left[
|\mathcal O(x,U(t))-T_x(U(t))|>\varepsilon T_x^{(\mathrm{avg})}
\right]
<\frac18    .
\label{eq:tbs-oracle-fixed-t}
\end{align}
The same bound also holds with $U(t)$ replaced by $V_R(\theta)U(t)$ for any fixed $\theta$, since left multiplication by $V_R(\theta)$ preserves Haar measure and does not increase total variation distance from Haar.

We now apply the robust Fourier coefficient extraction algorithm.
Since $T_x(V_R(\theta)U(t))$ is a trigonometric polynomial with degree at most $2N\le 2n$, Theorem~\ref{thm:probabilistic-robust-fourier-extraction} yields an estimate $\widehat p(t)$ satisfying
\begin{align}
\Pr\left[
|\widehat p(t)-p(t)|>
\varepsilon T_x^{(\mathrm{avg})}e^{O(n)}
\right]
<\frac14.
\label{eq:tbs-fourier-output}
\end{align}

We repeat this procedure at the interpolation points $t_j\in[0,\Delta]$ specified in Theorem~\ref{thm:probabilistic-robust-BW-extrapolation}.
The robust Berlekamp--Welch algorithm then yields an estimate $\widehat p(1)$ satisfying
\begin{align}
\Pr\left[
|\widehat p(1)-p(1)|>
\varepsilon T_x^{(\mathrm{avg})}
\exp(2n\log\Delta^{-1}+O(n))
\right]
<\frac13.
\label{eq:tbs-BW-output}
\end{align}
Using Eq.~\eqref{eq:tbs-proof-p1}, we divide the obtained estimate $\widehat p(1)$ by the known factor $\kappa^{2n}((n-k_0)!)^2G_{n,k}$ to obtain an estimate of $\Per(X_0)^2$ within additive imprecision
\begin{align}
    \varepsilon'
    :=
    \frac{\varepsilon T_x^{(\mathrm{avg})}\exp(2n\log\Delta^{-1}+O(n))}
    {\kappa^{2n}\left((n-k_0)!\right)^2G_{n,k}}.
\label{eq:tbs-epsilon-prime}
\end{align}

It remains only to bound the above quantity.
We now choose $\kappa = 1/\sqrt n$, so that $\kappa^{-2n}=n^n=\exp(n\log n)$.
Also, by definition, we have $\Delta^{-1}=O(\alpha n^2)$, and hence
\begin{align}
    \exp(2n\log\Delta^{-1})
    \le
    \exp(4n\log n+2n\log\alpha+O(n)).
\label{eq:tbs-interp-overhead}
\end{align}
Since $k_0= O(n^{\lambda})$ with $\lambda < 1$, Stirling's formula gives
\begin{align}
    \left((n-k_0)!\right)^2
    =
    \exp(2n\log n-O(n)).
\label{eq:tbs-factorial-bound}
\end{align}
By Lemma~\ref{lem:tbs-Gs-lower-bound}, we have $G_{n,k}^{-1}=e^{O(n)}$.
Finally, for fixed $\alpha>2$, Lemma~\ref{lem:typical-threshold-Haar-weight-upper} in Appendix~\ref{subsec: threshold Haar upper} gives
\begin{align}
    T_x^{(\mathrm{avg})}
    \le
    \exp(-n\log\alpha+O(n)).
\label{eq:tbs-TH-bound-proof}
\end{align}
Substituting Eq.~\eqref{eq:tbs-interp-overhead}--Eq.~\eqref{eq:tbs-TH-bound-proof} into Eq.~\eqref{eq:tbs-epsilon-prime}, we obtain
\begin{align}
\varepsilon'
&\le
\varepsilon
\exp(3n\log n+n\log\alpha+O(n)).
\label{eq:tbs-epsilon-prime-final}
\end{align}

Therefore, if
\begin{align}
    \varepsilon
    \le
    \exp(-3n\log n- n\log\alpha-O(n)),
\label{eq:tbs-epsilon-final}
\end{align}
then the constant hidden in $O(n)$ can be chosen so that $\varepsilon'<1/2$.
Rounding the resulting estimate to the nearest integer therefore yields $\Per(X_0)^2$ exactly.
Moreover, the overall procedure succeeds with probability at least $2/3$, and repetition and majority voting can arbitrarily boost this success probability.
All steps of the reduction are polynomial-time procedures with access to an NP oracle.
Therefore, oracle access to ATE with $\varepsilon$ satisfying Eq.~\eqref{eq:tbs-epsilon-final} and failure probability $\delta<1/8$ allows the exact computation of $\Per(X_0)^2$ for arbitrary $X_0\in\{0,1\}^{k_0\times k_0}$, where the overall procedure is in $\rm{BPP^{NP}}$.

\end{proof}

\section{Parity BosonSampling}\label{Section: parity BS}

We now prove Theorem~\ref{thm: parity average case main}.
The proof follows the common worst-to-average-case reduction described in Sec.~\ref{subsec: common workflow}, while retaining the parity-specific beam-splitter rotations and permanent-encoding construction, illustrated in Fig.~\ref{fig: outline}.

\subsection{Proof overview}\label{subsec: proof sketch for PBS}

\begin{figure*}[t]
\includegraphics[width=0.9\linewidth]{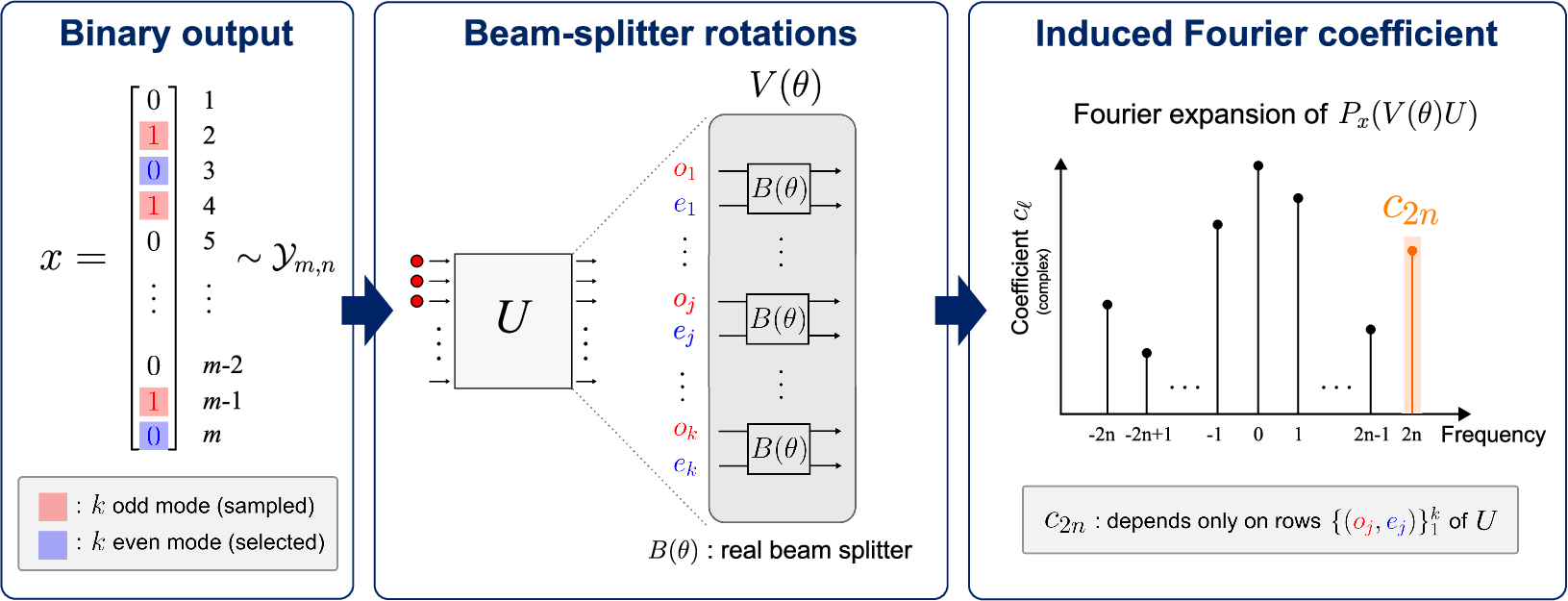}
\caption{Parity-specific Fourier-filtering construction used in the proof of Theorem~\ref{thm: parity average case main}.
Starting from a sampled parity outcome $x\sim \mathcal Y_{m,n}$ with $|x| = k$ odd parities, we choose $k$ even-parity modes and apply paired beam-splitter rotations $B(\theta)$ to all odd-even pairs.
The resulting probability becomes a Fourier series in $\theta$ of degree at most $2n$.
Its top Fourier coefficient, denoted by $c_{2n}$, filters out all contributions except those in which all $n$ photons occupy the rotated modes, and therefore depends only on the corresponding $2k\times n$ submatrix $U_{R,[n]}$.
The shared worst-to-average-case reduction procedure is described in Secs.~\ref{subsec: common workflow} and~\ref{section: shared framework}.
}
\label{fig: outline}
\end{figure*}

We provide an overview of the worst-to-average-case reduction proving Theorem~\ref{thm: parity average case main}.
As in the TBS case, the reduction starts from an arbitrary $k_0\times k_0$ binary matrix $X_0$ and shows that oracle access to APE in Definition~\ref{def: parityBS} enables the exact computation of $\Per(X_0)^2$.
Since exact computation of the permanent of a binary matrix is \#P-hard in the worst case~\cite{valiant1979complexity}, this establishes the desired hardness result.

We first establish a typical-outcome property for the parity ensemble.
For $x\sim\mathcal Y_{m,n}$ with $m=\alpha n$, Lemma~\ref{lemma: most parity outcomes have many odd parities} shows that the number of odd-parity modes concentrates near $\alpha n/(\alpha+2)$.
Writing $k=|x|$ and $O=\supp(x)$, a typical outcome therefore has $k=\Theta(n)$ odd-parity modes.
We then choose a set $E\subseteq[m]\setminus O$ of $k$ even-parity modes and fix a pairing between the modes in $O$ and those in $E$.
Denoting the union of these paired modes by $R:=O\sqcup E$, we have $|R|=2k$.
For fixed $\alpha>2$, the same typical-outcome event also guarantees the margin $m-n-2k=\Omega(n)$ required later for the Haar-submatrix perturbation argument.

For these fixed mode sets, consider an arbitrary linear-optical circuit $U$.
We apply identical coherent beam-splitter rotations to all paired modes in $R$ and examine the resulting parity probability $P_x(V_R(\theta)U)$ as a function of the rotation angle $\theta$.
This probability is a finite Fourier series of degree at most $2n$.
Unlike in the threshold setting, however, the original parity probability is not determined by $U_{R,[n]}$: modes outside $R$ have even parity and may still contain a positive even number of photons, so $P_x(U)$ generally depends on the full $m\times n$ input-output block.
The crucial role of the highest-order Fourier coefficient $c_{2n}$ is to remove this global dependence.
To contribute at frequency $2n$, all $n$ photons must occupy the rotated modes in $R$.
Consequently, $c_{2n}$ depends only on the rectangular block $U_{R,[n]}$ (Lemma~\ref{lem:rowfilter}), providing the localized block into which the worst-case instance can be encoded.

In Lemma~\ref{lem:large-block-permanent-encoding}, we use this localization to encode the worst-case instance $X_0$.
Choose $k_0$ of the odd-parity modes to carry the rows of $X_0$.
We construct a $2k\times n$ block $A_\star$ that embeds $X_0$ on these selected odd-mode rows, while its other entries are chosen so that every nonvanishing contribution to the top Fourier coefficient contains each hard row exactly once.
For every surviving occupation pattern, the corresponding permanent reduces, up to known phase factors, to $\kappa^n\Per(X_0)(n-k_0)!$, with these phases canceling between the two permanent factors in the top Fourier coefficient.
Lastly, summing over the remaining occupation patterns contributes only the explicitly known and efficiently computable factor $Q_{n,k,k_0}$.
Thus, the resulting top Fourier coefficient is proportional to $\Per(X_0)^2$ with known prefactor $\kappa^{2n}((n-k_0)!)^2Q_{n,k,k_0}$ (Lemma~\ref{lem:large-block-permanent-encoding}).
The condition $k_0<k$ required by the construction holds for the chosen $k_0=o(n)$ and for typical $x\sim\mathcal Y_{m,n}$.

The remainder of the reduction then follows the common workflow in Sec.~\ref{subsec: common workflow} and invokes the technical results established in Sec.~\ref{section: shared framework}.
Lemma~\ref{lem:sparse-TV} bounds the distributional change induced by the affine block perturbation, and Lemma~\ref{lem:completion} lifts this perturbation to a full unitary through conditional completion.
Finally, Theorem~\ref{thm:probabilistic-robust-fourier-extraction} extracts the relevant Fourier coefficient from APE estimates, and Theorem~\ref{thm:probabilistic-robust-BW-extrapolation} extrapolates this coefficient to the endpoint encoding the worst-case quantity.

The rest of this section proceeds as follows.
In Sec.~\ref{subsec: most parity outcomes}, we show that most parity outcomes have the number of odd parities concentrated near $\alpha n/(\alpha+2)$.
In Sec.~\ref{subsec: fourier}, we introduce the parity Fourier coefficients, establish their band limitation, and derive the top-coefficient row filter.
In Sec.~\ref{subsec: encoding single permanent}, we show how to encode the worst-case \#P-hard quantity into the parity Fourier coefficient for most parity outcomes.
In Sec.~\ref{subsec: translation inv}, we specialize the shift-and-scale stability and unitary-completion framework to the parity setting.
Finally, in Sec.~\ref{subsec: proof of parity average}, we complete the proof of Theorem~\ref{thm: parity average case main}.


\subsection{Odd-parity count concentration}\label{subsec: most parity outcomes}

We first show that, for $x\sim\mathcal{Y}_{m,n}$, the number of odd-parity modes $|x|$ concentrates near $\alpha n/(\alpha+2)$.

\begin{lemma}[Odd-parity concentration]\label{lemma: most parity outcomes have many odd parities}
    Let $\alpha = m/n \geq 1$, and fix $\eta \in (0,\min\left\{ \alpha/(\alpha + 2), 2/(\alpha + 2) \right\})$.
    Let $\mathcal{Y}_{m,n}$ be the distribution in Definition~\ref{def: random parity outcome ensemble}.
    Then, for all sufficiently large $n$, there exist constants $C_{\alpha,\eta},c_{\alpha,\eta}>0$ such that
    \begin{align}
        \Pr_{x\sim \mathcal{Y}_{m,n}}
        \left[ \left| |x| - \frac{\alpha}{\alpha+2}n \right| \geq \eta n\right]
        \le C_{\alpha,\eta}e^{-c_{\alpha,\eta}n}
    \end{align}
\end{lemma}

\begin{proof}
    See Appendix~\ref{appendix: section: odd parity concentration}.
\end{proof}

Lemma~\ref{lemma: most parity outcomes have many odd parities} implies that, for $x \sim \mathcal{Y}_{m,n}$, the event $|x| = \Theta(n)$ occurs with probability at least $1 - \exp(-\Theta(n))$.
Moreover, for any fixed $\alpha > 2$, choosing a sufficiently small $\eta>0$ yields
\begin{align}
m-n-2|x|
&\ge
\left(
\alpha-1-\frac{2\alpha}{\alpha+2}-2\eta
\right)n=
\Omega(n),
\label{eq:parity-margin-rigorous}
\end{align}
again with probability at least $1 - \exp(-\Theta(n))$.
These two properties are the typical-outcome inputs exploited in the proof of Theorem~\ref{thm: parity average case main}.

\subsection{Parity Fourier coefficients and band limitation}\label{subsec: fourier}

Using the paired beam-splitter rotations defined in Sec.~\ref{subsec: paired beam-splitter rotations}, we now define the Fourier coefficients of the parity probability as a function of the rotation angle.
Let $R\subseteq[m]$ be a paired active set as defined in Sec.~\ref{subsec: paired beam-splitter rotations}.

\begin{definition}[Fourier coefficient]\label{def: fourier coeff}
Given $x\in\mathcal{P}_{m,n}$, $U\in\mathrm{U}(m)$, and a paired active set $R$ as above, define
\begin{align}\label{eq: def of fourier-coeff}
c_\ell(x,U,R)
:=
\frac{1}{2\pi}\int_0^{2\pi}
P_x(V_R(\theta)U)e^{-i\ell\theta}\,d\theta,
\quad \ell\in\mathbb{Z}.
\end{align}
\end{definition}

We next show that this function of $\theta$ is band-limited, such that its Fourier support lies in $\{-2n,\ldots,2n\}$.

\begin{lemma}\label{lem:bandlimit}
For any fixed $x\in\mathcal{P}_{m,n}$, $U\in\mathrm{U}(m)$, and paired active set $R$, the parity output probability $P_x(V_R(\theta)U)$ is a trigonometric polynomial in $\theta$ of degree at most $2n$.
Equivalently,
\begin{align}\label{eq:bandlimit-expansion}
    P_x(V_R(\theta)U)
    =
    \sum_{\ell=-2n}^{2n} c_\ell(x,U,R)e^{i\ell\theta},
\end{align}
where the coefficients are given in Definition~\ref{def: fourier coeff}.
\end{lemma}

\begin{proof}
Write $U_R(\theta):=V_R(\theta)U$ throughout the proof.
For a fixed $S \in \mathcal{S}_{m,n}$, let $n_R(S):=\sum_{i\in R}s_i$ be the number of photons detected in the output modes indexed by $R$.
By Eq.~\eqref{eq: rows of rotated circuit}, the submatrix $U_R(\theta)_{\nu(S),[n]}$ contains exactly $n_R(S)$ $\theta$-dependent rows, counted with multiplicity, while all other rows are independent of $\theta$.
By the row multilinearity of the permanent established in Sec.~\ref{subsec: paired beam-splitter rotations}, $\Per(U_R(\theta)_{\nu(S),[n]})$ has degree at most $n_R(S)$ in $e^{\pm i\theta}$.
Consequently, its squared modulus is a trigonometric polynomial of degree at most $2n_R(S)$.
Summing over $S$ as in Eq.~\eqref{eq:main-parity-probability} preserves this bound, and since $n_R(S)\le n$, the parity probability $P_x(U_R(\theta))$ is a trigonometric polynomial in $\theta$ of degree at most $2n$.

\end{proof}

Now fix $\alpha=m/n>2$ and a parity string $x\in\mathcal{P}_{m,n}$ with weight $|x|=k$.
In the reduction we will restrict to typical outcomes for which $k=\Theta(n)$ and $m-n-2k=\Omega(n)$.
By Lemma~\ref{lemma: most parity outcomes have many odd parities}, these conditions hold with probability at least $1-\exp(-\Theta(n))$ when $x\sim\mathcal{Y}_{m,n}$.

For the fixed $x$, let $O:=\supp(x)$ be the set of odd-parity modes; thus $|O|=k$.
Choose an arbitrary set $E\subseteq[m]\setminus O$ with $|E|=k$, and fix a perfect matching between $O$ and $E$:
\begin{align}\label{eq: def of R}
    R:=O\sqcup E
    =
    \{o_1,e_1\}\sqcup\cdots\sqcup\{o_k,e_k\},
\end{align}
where $o_j\in O$ and $e_j\in E$.
We regard this pairing as part of the data defining $R$.

Applying the beam-splitter rotation $V_R(\theta)$ to these paired output modes gives the Fourier coefficients $c_\ell(x,U,R)$ from Definition~\ref{def: fourier coeff}.
We focus on the highest-frequency coefficient
\begin{align}\label{eq:top}
c_{2n}(x,U,R)
=
\frac{1}{2\pi}\int_0^{2\pi}
P_x(V_R(\theta)U)e^{-i2n\theta}\,d\theta ,
\end{align}
which we also refer to as the top Fourier coefficient.

Importantly, while $P_x(U)$ itself depends on the entire $m \times n$ block of $U$ (first $n$ columns), the top coefficient depends only on a $2k \times n$ submatrix of $U$, as shown in the following lemma.

\begin{lemma}\label{lem:rowfilter}
Let $x \in \mathcal{P}_{m,n}$, $U \in {\rm U}(m)$, and $R = O \sqcup E$ for $O = \supp(x)$ and any choice of $E \subseteq [m]\setminus O$ with $|E| = |O|$.
Then only configurations satisfying $n_R(S):=\sum_{i\in R}s_i=n$, equivalently $\supp(S)\subseteq R$, contribute to the top Fourier coefficient, and
\begin{align}\label{eq:c2n-final}
&c_{2n}(x,U,R) \nonumber \\
&=
\!\!\!\!\!\sum_{\substack{ S \in \mathcal{S}_{m,n} \\ S\;{\rm mod}\; 2 = x \\ n_R(S)=n }}
\!\!\!\!\!\frac{1}{\prod_i s_i!}\,
\Per\!\big((V_R^+U)_{\nu(S),[n]}\big)
\Per\!\big((V_R^-U)_{\nu(S),[n]}\big)^*.
\end{align}
Moreover, since $V_R^\pm$ has support only on rows indexed by $R$, the coefficient $c_{2n}(x,U,R)$ depends only on the $2k\times n$ submatrix $U_{R,[n]}$ of $U$.
\end{lemma}

This isolation property plays a central role in the hardness proof, particularly in the PBS setting, because it allows us to apply perturbation-stability results for Haar-random submatrices when embedding the worst-case matrix as a small perturbation, as in related constructions~\cite{bouland2022noise,bouland2026complexity,bouland2025exponential,deshpande2022quantum,grier2022complexity,go2026computational}.

\begin{proof}[Proof of Lemma~\ref{lem:rowfilter}]
As in the proof of Lemma~\ref{lem:bandlimit}, we denote $U_R(\theta) = V_{R}(\theta)U$ throughout this proof.
By definition, the top coefficient can be expressed as
\begin{align}\label{eq:c2n-sum-int}
&c_{2n}(x,U,R) \nonumber  \\
&= \!\!\!\!\!\!\sum_{\substack{ S \in \mathcal{S}_{m,n} \\ S\;\text{mod 2} =  x   } }
\!\!\!\!\frac{1}{\prod_i s_i!}\cdot
\frac{1}{2\pi}\int_{0}^{2\pi}
\Big|\Per\big(U_R(\theta)_{\nu(S),[n]}\big)\Big|^2\,e^{-i2n\theta}\,d\theta.
\end{align}

For a given $S \in \mathcal{S}_{m,n}$, let $n_R(S):=\sum_{i\in R} s_i$ be the number of photons occupying the output modes $R$.
As argued in the proof of Lemma~\ref{lem:bandlimit}, Eq.~\eqref{eq: rows of rotated circuit} and multilinearity of permanent implies that $\big|\Per(U_{R}(\theta)_{\nu(S),[n]})\big|^2$ is a trigonometric polynomial in $\theta$ of degree at most $2n_R(S)$.
Therefore, the integrand in Eq.~\eqref{eq:c2n-sum-int} contains no $e^{i2n\theta}$ term unless $2n_R(S)\ge 2n$, which, by constraint $n_R(S) \leq n$, forces $n_R(S)=n$.
Accordingly, for every $S$ contributing to $c_{2n}(x,U,R)$, all photons must occupy the modes confined in $R$, and for such $S$, every row of $U_{R}(\theta)_{\nu(S),[n]}$ has the form:
\begin{align}
(U_R(\theta)_{\nu(S),[n]})_j
&= e^{i\theta}\,((V_R^+U)_{\nu(S),[n]})_j  \nonumber \\
&+ e^{-i\theta}\,((V_R^-U)_{\nu(S),[n]})_j.
\end{align}
By multilinearity over the $n$ rows, expanding $\Per(U_{R}(\theta)_{\nu(S),[n]})$ yields a sum of $2^n$ terms, one for each choice of whether to pick the $e^{i\theta}$ or $e^{-i\theta}$ part in each row.
Here, the largest exponent $e^{in\theta}$ occurs when we pick the $e^{i\theta}$ part from every row, and the corresponding coefficient is exactly $\Per\big((V_R^+U)_{\nu(S),[n]}\big)$.
Similarly, the smallest exponent $e^{-in\theta}$ occurs when we pick the $e^{-i\theta}$ part from every row, and its coefficient is $\Per\big((V_R^-U)_{\nu(S),[n]}\big)$.
Therefore, the coefficient of $e^{i2n\theta}$ in $\big|\Per(U_R(\theta)_{\nu(S),[n]})\big|^2$ is given by
\begin{align}
\Per\big((V_R^+U)_{\nu(S),[n]}\big)\cdot {\Per\big((V_R^-U)_{\nu(S),[n]}\big)}^*  .
\end{align}
Since the integration over $\theta$ in Eq.~\eqref{eq:c2n-sum-int} extracts exactly the coefficient of $e^{i2n\theta}$ by the orthogonality, this yields Eq.~\eqref{eq:c2n-final}.

\end{proof}

\subsection{Encoding a single permanent into the parity Fourier coefficient}\label{subsec: encoding single permanent}
\label{subsubsec:large-block-encoding}

Similarly to the threshold case, we generalize the definition of the top Fourier coefficient to take a rectangular matrix $A\in\mathbb{C}^{2k \times n}$ as an input, instead of a unitary matrix.
To this end, we use a convention $\widetilde{A} \in \mathbb{C}^{m \times n}$ to denote any matrix completion of $A$ that satisfies $\widetilde{A}_{R,[n]} = A $.
This allows us to define the block extension of the top Fourier coefficient in Lemma~\ref{lem:rowfilter} by
\begin{align}
&c^{\mathrm{blk}}_{2n}(x,A,R) \nonumber \\
&:=
\!\!\!\!\!\sum_{\substack{S\in \mathcal{S}_{m,n}\\ S \bmod 2=x}}
\!\!\!\!\frac{1}{\prod_i s_i!}
\Per\left((V_R^+\widetilde{A})_{\nu(S),[n]}\right)
\Per\left((V_R^-\widetilde{A})_{\nu(S),[n]}\right)^*
.
\label{eq:block-top-coefficient}
\end{align}
Accordingly, by Lemma~\ref{lem:rowfilter}, whenever $A=U_{R,[n]}$ for a unitary $U$, we have
\begin{align}
c^{\mathrm{blk}}_{2n}(x,U_{R,[n]},R)=c_{2n}(x,U,R).
\label{eq:block-top-equals-fourier}
\end{align}

We now describe how to encode a worst-case \#P-hard quantity into this top Fourier coefficient.
Specifically, let $X_0\in\{0,1\}^{k_0\times k_0}$ be an arbitrary worst-case instance, with $k_0=o(n)$ in the reduction.
In the following lemma, we construct a $2k \times n$ worst-case matrix $A_\star \in \mathbb{C}^{2k \times n}$ such that $c^{\mathrm{blk}}_{2n}(x,A_\star,R)$ encodes $\Per(X_0)^2$ up to a known multiplicative factor, so that estimating $c^{\mathrm{blk}}_{2n}(x,A_\star,R)$ immediately yields an estimate of $\Per(X_0)^2$.

\begin{lemma}[Single-permanent encoding]
\label{lem:large-block-permanent-encoding}
Fix $\alpha = m/n > 2$.
Let $x \in \mathcal{P}_{m,n}$ with $|x| = k$, and let $R = O \sqcup E$ for $O = \supp(x)$ and any choice of $E \subseteq [m]\setminus O$ with $|E| = |O|$.
Let $c^{\mathrm{blk}}_{2n}$ be the block extension of the top Fourier coefficient in Eq.~\eqref{eq:block-top-coefficient}.
Given a matrix $X_0\in\{0,1\}^{k_0\times k_0}$ with $k_0 < k$, there exists a $2k \times n$ matrix $A_\star \in \mathbb{C}^{2k \times n}$ that can be constructed in polynomial time and satisfies
\begin{align}
c^{\mathrm{blk}}_{2n}(x,A_\star,R)
=
\kappa^{2n}
\left((n-k_0)!\right)^2
Q_{n,k,k_0}
\Per(X_0)^2,
\label{eq:large-block-encoding}
\end{align}
where $\kappa>0$ is a free normalization parameter one can arbitrarily choose, and the factor
\begin{align}
Q_{n,k,k_0}
:=
\!\!\!\!\!\!\!\!\!\!\sum_{\substack{a_1,\ldots,a_{k-k_0}\ge 0\\ b_1,\ldots,b_{k-k_0}\ge 0\\
\sum_{j=1}^{k-k_0}\left((2a_j+1)+2b_j\right)=n-k_0}}
\!\!\!\!\!\!\!\!\!\prod_{j=1}^{k-k_0}
\frac{1}{(2a_j+1)!(2b_j)!} ,
\label{eq:Q-nkk-definition}
\end{align}
which can be computed in polynomial time.
\end{lemma}

\begin{figure*}[t]
\includegraphics[width=0.85\linewidth]{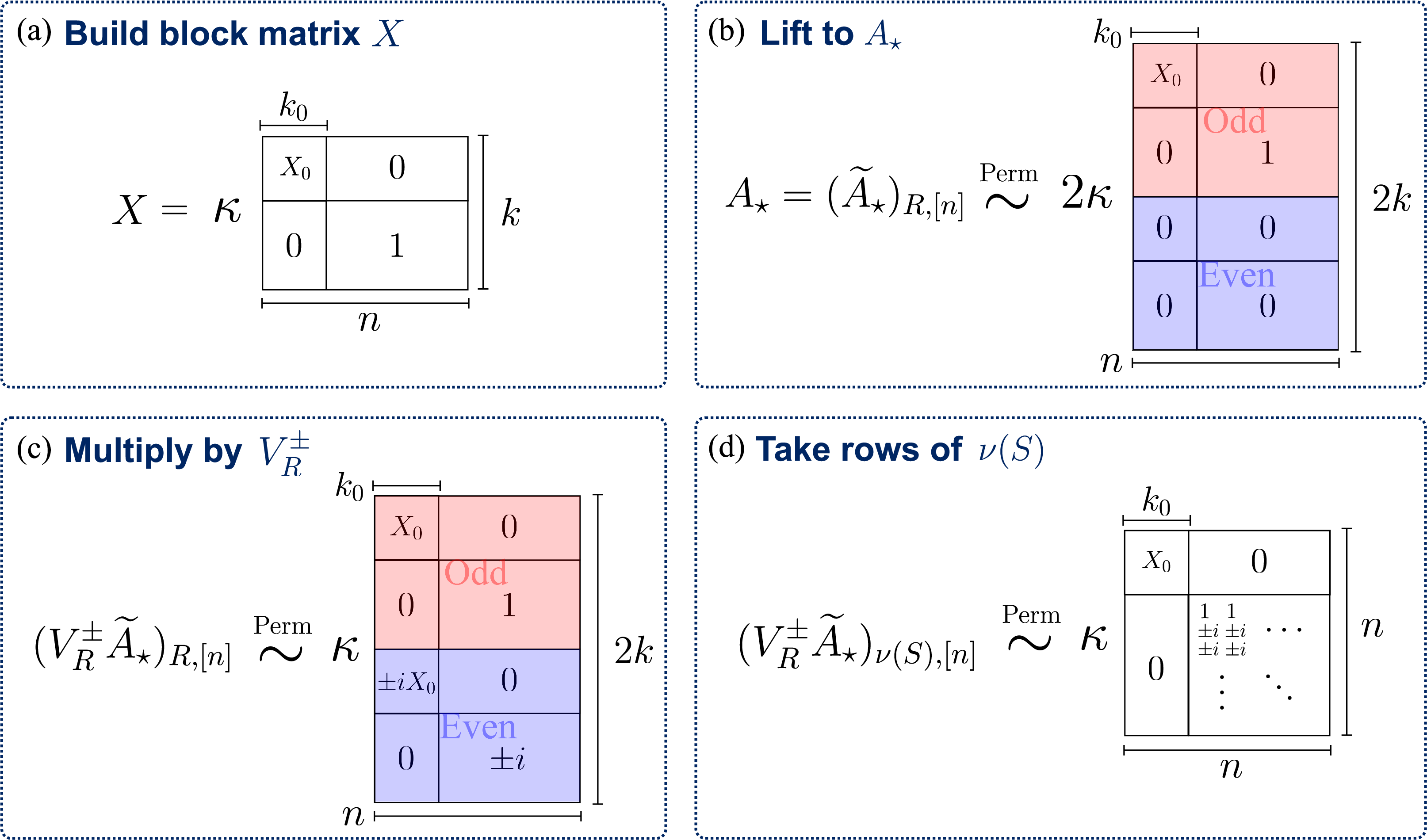}
\caption{
Illustration of the matrix constructions in the proof of Lemma~\ref{lem:large-block-permanent-encoding}.
(a) Schematic of the block matrix $X$ in Eq.~\eqref{eq:large-block-wj}, constructed from a worst-case matrix $X_0\in\{0,1\}^{k_0\times k_0}$ (where $k_0 < k$) by adjoining all-zero blocks and an all-ones block, together with a normalization factor $\kappa>0$.
(b) Schematic of the matrix $A_\star=(\widetilde A_\star)_{R,[n]}$ in Eq.~\eqref{eq:large-block-Astar}, obtained by embedding $X$ into the rows indexed by $R = O \sqcup E$ and the first $n$ columns of a larger matrix $\widetilde A_\star$.
In the figure, the labels ``odd” and ``even” inside the blocks indicate the rows indexed by $O$ and $E$, respectively, corresponding to the modes where odd and even parities are measured.
The notation $\overset{\mathrm{Perm}}{\sim}$ denotes equivalence up to row permutations.
(c) Schematic of the matrices $(V_R^\pm \widetilde{A}_\star)_{R,[n]}$ in Eq.~\eqref{eq:large-block-Wpm}.
(d) Schematic of the matrix $(V_R^\pm\widetilde A_\star)_{\nu(S),[n]}$ used in Eq.~\eqref{eq:large-block-proof-start}, obtained from (c) after imposing the constraint $S \bmod 2 = x$.
}
\label{fig: encoding PBS}
\end{figure*}

\begin{proof}

We construct the matrix $A_\star$ below; an illustration of the construction is provided in Fig.~\ref{fig: encoding PBS}.

Given a matrix $X_0\in\{0,1\}^{k_0\times k_0}$, define the $k \times n$ matrix
\begin{align}
X :=
\kappa
\begin{pmatrix}
X_0 & 0_{k_0\times(n-k_0)}\\
0_{(k-k_0)\times k_0} & 1_{(k-k_0)\times(n-k_0)}
\end{pmatrix}  ,
\label{eq:large-block-wj}
\end{align}
where $\kappa$ is a free parameter one can arbitrarily choose, and $0_{a\times b}$ and $1_{a \times b}$ denote the $a \times b$ all-zero and all-one matrices, respectively.
Provided that $R = O \sqcup E$ for odd and even modes $O$ and $E$, one can construct a matrix $A_\star \in \mathbb{C}^{2k \times n}$ whose $m \times n$ completion $\widetilde{A}_\star$ satisfies
\begin{align}
(\widetilde{A}_\star)_{O, [n]} = 2X   ,
\quad
(\widetilde{A}_\star)_{E, [n]} = 0  .
\label{eq:large-block-Astar}
\end{align}
By Eq.~\eqref{eq:Vpm-full}, this also gives
\begin{align}
(V_R^\pm \widetilde{A}_\star)_{O,[n]} = X,
\quad
(V_R^\pm \widetilde{A}_\star)_{E, [n]}= \pm i X  .
\label{eq:large-block-Wpm}
\end{align}

Based on this construction, consider the corresponding top Fourier coefficient:
\begin{align}
c^{\mathrm{blk}}_{2n}(x,A_\star,R)
&=\!\!\!\!\!\!\sum_{\substack{S\in \mathcal{S}_{m,n}\\ S \bmod 2=x }}
\!\!\!\!\!\frac{1}{\prod_i s_i!}
\Per\left((V_R^+\widetilde{A}_\star)_{\nu(S),[n]}\right) \nonumber \\
& \qquad\qquad\qquad\times \Per\left((V_R^-\widetilde{A}_\star)_{\nu(S),[n]}\right)^*
.
\label{eq:large-block-proof-start}
\end{align}
Recall that $o_j \in O$ and $e_j \in E$ denote the odd and even mode indices, respectively.
By the block-diagonal structure of $X$ in Eq.~\eqref{eq:large-block-wj}, for each $j=1,\ldots,k_0$, the rows of $V_R^\pm \widetilde{A}_\star$ indexed by $o_j$ and $e_j$ are supported on the first $k_0$ columns.
The parity constraint forces each hard odd mode $o_j$ to contain at least one photon.
If any additional photon occupied one of the hard pairs, then the corresponding matrix would contain more than $k_0$ rows supported only on the first $k_0$ columns, so its permanent would vanish.
Moreover, any occupation outside $R$ gives a zero row in $V_R^\pm\widetilde A_\star$ and therefore also gives no contribution.
Thus, every contributing occupation pattern must satisfy
\begin{align}
s_{o_j}=1,
\quad
s_{e_j}=0,
\quad
\label{eq:hard-pairs-forced-large-block}
\end{align}
for all $j=1,\ldots,k_0$.
This implies that all remaining $n-k_0$ photons must occupy the modes outside these pairs.
More concretely, one can write
\begin{align}
s_{o_{k_0+j}}=2a_j+1,
\qquad
s_{e_{k_0+j}}=2b_j,
\label{eq:sink-occupations-large-block}
\end{align}
for $j=1,\ldots,k-k_0$, where $a_j,b_j\ge 0$ and
\begin{align}
\sum_{j=1}^{k-k_0}\left((2a_j+1)+2b_j\right)=n-k_0.
\end{align}

For such $S$, the matrix $(V_R^+\widetilde{A}_\star)_{\nu(S),[n]}$ (and likewise $(V_R^-\widetilde{A}_\star)_{\nu(S),[n]}$) admits a block-diagonal form whose diagonal blocks are $\kappa X_0$ and a phase-shifted version of $\kappa 1_{(n-k_0)\times(n-k_0)}$ (by the factor $\pm i $ in Eq.~\eqref{eq:large-block-Wpm}).
More concretely, one can write
\begin{align}
\Per\left( (V_R^+\widetilde{A}_\star)_{\nu(S),[n]} \right)
=
i^{\sum_{j=1}^{k-k_0}2b_j}
\kappa^n
\Per(X_0)
(n-k_0)!   ,
\label{eq:large-block-plus-permanent}
\end{align}
and similarly,
\begin{align}
\Per\left((V_R^-\widetilde{A}_\star)_{\nu(S),[n]} \right)
=
(-i)^{\sum_{j=1}^{k-k_0}2b_j}
\kappa^n
\Per(X_0)
(n-k_0)!.
\label{eq:large-block-minus-permanent}
\end{align}
Since $\sum_j2b_j$ is always even, we obtain
\begin{align}
\Per\left((V_R^+\widetilde{A}_\star)_{\nu(S),[n]}\right)
\Per\left((V_R^-\widetilde{A}_\star)_{\nu(S),[n]}\right)^*
\nonumber \\
=
\kappa^{2n}\Per(X_0)^2\left((n-k_0)!\right)^2.
\end{align}

Also, for $S \in \mathcal{S}_{m,n}$ satisfying Eq.~\eqref{eq:hard-pairs-forced-large-block} and Eq.~\eqref{eq:sink-occupations-large-block}, the photon-number factorial factor is given by
\begin{align}
\prod_{i=1}^{m} s_i ! = \prod_{j=1}^{k-k_0}(2a_j+1)!(2b_j)!.
\end{align}
Summing over all possible $S$ finally yields Eq.~\eqref{eq:large-block-encoding}.

Lastly, one can rewrite $Q_{n,k,k_0}$ by
\begin{align}
Q_{n,k,k_0}
&=
[z^{n-k_0}](\sinh z\cosh z)^{k-k_0} \\
&=
\frac{1}{4^{k-k_0} (n-k_0)!}
\sum_{j=0}^{k-k_0}
(-1)^j
\binom{k-k_0}{j}
\nonumber \\
&\qquad\qquad\qquad\qquad\qquad\times\left(2(k-k_0-2j)\right)^{n-k_0} ,
\end{align}
where $[z^N]f(z)$ denotes the coefficient of $z^N$ in the power-series expansion of $f(z)$.
Since this expression contains only $O(n)$ terms, $Q_{n,k,k_0}$ can be evaluated exactly in polynomial time.
This completes the proof.

\end{proof}

We additionally show that the multiplicative factor $Q_{n,k,k_0}$ introduced in Lemma~\ref{lem:large-block-permanent-encoding} is lower-bounded by $\exp(-O(n))$ whenever $k_0 = o(n)$ and $k = \Omega(n)$, where the latter holds for typical parity outcomes over $x \sim \mathcal{Y}_{m,n}$ by Lemma~\ref{lemma: most parity outcomes have many odd parities}.

\begin{lemma}[Lower bound on the normalization factor]
\label{lem:Q-nkk-lower-bound}
Suppose $k_0 = o(n)$ and
\begin{align}\label{eq: k lwer bound}
k\geq \left( \frac{\alpha}{\alpha+2}  - \eta \right) n  ,
\end{align}
for any $\alpha > 2$ and $0 < \eta \leq 1/6$.
Then, the multiplicative factor $Q_{n,k,k_0}$ in Lemma~\ref{lem:large-block-permanent-encoding} is lower bounded by
\begin{align}
Q_{n,k,k_0}
\ge
\exp\left[
-\left(
\frac{1}{\alpha+2}+\frac{\eta}{2}
\right)(\log 2)n
\right].
\label{eq:Q-nkk-lower-bound-explicit}
\end{align}
\end{lemma}

\begin{proof}
    See Appendix~\ref{appendix: section: single permanent into coefficient - PBS}.
\end{proof}

\subsection{Shift-and-scale stability and conditional unitary completion}\label{subsec: translation inv}

In the parity setting, let $A=U_{R,[n]}$ for $U\sim\Haar(m)$, and let $A_\star$ be the block constructed in Lemma~\ref{lem:large-block-permanent-encoding}.
Also, let $\mu$ denote the distribution of $A=U_{R,[n]}$, and let $\mu_t$ denote the distribution of the affine map $A(t) = (1-t)A+tA_\star$.
The following lemma establishes the required shift-and-scale stability bound for this PBS setting, where the relevant objects are rectangular submatrices of Haar-random unitaries.

\begin{lemma}[Shift-and-scale stability for PBS]
\label{lem:sparse-TV}
Let $\mu$ be the distribution of $A = U_{R,[n]}$ for $U \sim \Haar(m)$ and $|R| = 2k$, and let $\mu_t$ be the distribution of $A(t) = (1-t)A + tA_\star$ for $A_\star$ constructed in Lemma~\ref{lem:large-block-permanent-encoding}.
Suppose that $m - n - 2k = \Omega(n)$.
There exist constants $c_0, c_1, c_2 > 0$ such that, whenever
\begin{align}\label{eq: bound for t}
  0 \leq t \leq \frac{c_0}{mn(1 + \kappa \sqrt{k} )} ,
\end{align}
the total variation distance between the distributions $\mu$ and $\mu_t$ is bounded by
\begin{align}
  \|\mu_t-\mu\|_{\mathrm{TVD}}
  \le
  c_1 mn(1 + \kappa \sqrt{k} ) t + e^{-c_2(m-n-2k)}.
\end{align}
\end{lemma}

\begin{proof}
See Appendix~\ref{appendix: section: translation invariance}.
\end{proof}

Here, the odd-parity concentration from Lemma~\ref{lemma: most parity outcomes have many odd parities} applies.
For $x\sim\mathcal{Y}_{m,n}$ and fixed $\alpha>2$, Eq.~\eqref{eq:parity-margin-rigorous} shows that the required margin $m-n-2k=\Omega(n)$ holds with probability at least $1-\exp(-\Theta(n))$.

We next lift the block-level perturbation to a full unitary matrix, since the APE problem in Definition~\ref{def: parityBS} takes a complete circuit unitary as input.

\begin{lemma}[Conditional unitary completion for PBS]
\label{lem:completion}
Let $\mu$ be the distribution of the block $A=U_{R,[n]}$ under $U\sim\Haar(m)$, and let $\mu_t$ be the distribution of the block $A(t) = (1-t)A + tA_\star$.
Suppose that $m - n - 2k = \Omega(n)$ and $t$ satisfy Eq.~\eqref{eq: bound for t}.
There exists a randomized polynomial-time procedure $\mathcal{C}$ that, with probability at least $1 - \exp({-\Omega(m-n-2k)})$, completes the input matrix to a unitary containing it as the prescribed submatrix.
Moreover, the corresponding pushforward measure $\mathcal H_t:=\mathcal C_{\#}\mu_t$ satisfies
\begin{align}
  \|\mathcal{H}_t- \Haar(m) \|_{\mathrm{TVD}} \le \|\mu_t-\mu\|_{\mathrm{TVD}} + e^{-\Theta(m-n-2k)}.
\end{align}
\end{lemma}

\begin{proof}
See Appendix~\ref{appendix: section: conditional-haar-completion}.
\end{proof}

\subsection{Proof of Theorem~\ref{thm: parity average case main}}\label{subsec: proof of parity average}

We finally prove Theorem~\ref{thm: parity average case main}.
To this end, we present the following quantitative statement.

\begin{theorem}
\label{thm:parity-no-fourier-AC}
Let $m=\alpha n$ for a fixed constant $\alpha>2$.
Then {\rm APE} problem in Definition~\ref{def: parityBS} is $\#\mathrm{P}$-hard under a ${\rm BPP^{NP}}$ reduction, for additive imprecision
\begin{align}
\varepsilon \leq \exp(-3n\log n - 1.01 n\log \alpha - O(n))  ,
\end{align}
and failure probability $\delta < 1/8$.
\end{theorem}

\begin{proof}
Let $\mathcal O$ be an oracle for APE such that, on input $x\in \mathcal{P}_{m,n}$ and $U\in \mathrm{U}(m)$, it outputs an estimate $\mathcal{O}(x,U)$ of $P_x(U)$ satisfying
\begin{align}
\Prob_{x,U}
\left[
|\mathcal O(x,U)-P_x(U)|>\varepsilon P_x^{(\mathrm{avg})}
\right]
<\delta,
\label{eq:oracle-no-AC}
\end{align}
where the probability is taken over $x\sim \mathcal{Y}_{m,n}$ in Definition~\ref{def: random parity outcome ensemble} and $U\sim\Haar(m)$.
We follow the proof of Theorem~\ref{thm: threshold average case quantitative} and explain only the parity-specific ingredients.


Given $X_0\in\{0,1\}^{k_0\times k_0}$, choose any fixed $0<\lambda<1$ and set $n=\left\lceil k_0^{1/\lambda}\right\rceil$ and $m=\alpha n$.
We sample $x\sim\mathcal{Y}_{m,n}$ and set $k=|x|$.
By Lemma~\ref{lemma: most parity outcomes have many odd parities}, we may choose $0<\eta\leq\min\{0.01,1/(\alpha+2)\}$ sufficiently small so that
\begin{align}
\left|k-\frac{\alpha}{\alpha+2}n\right|
\leq
\eta n
\label{eq:typical-k-no-AC}
\end{align}
with probability at least $1-\exp(-\Theta(n))$, and this event implies
\begin{align}
k=\Theta(n),
\qquad
k>k_0,
\qquad
m-n-2k=\Omega(n).
\label{eq:margin-no-AC}
\end{align}
Conditioned on this event, let $O=\supp(x)$, choose $E\subseteq[m]\setminus O$ with $|E|=k$, set $R=O\sqcup E$, and fix a pairing between $O$ and $E$.
Construct $A_\star$ as in Lemma~\ref{lem:large-block-permanent-encoding}, sample $U\sim\Haar(m)$, set $A:=U_{R,[n]}$, and define $A(t):=(1-t)A+tA_\star$.
Then
\begin{align}
p(t):=c_{2n}^{\mathrm{blk}}(x,A(t),R)
\label{eq:p-polynomial-proof-no-AC}
\end{align}
is a polynomial of degree at most $2n$, and
\begin{align}
p(1)
=
\kappa^{2n}\left((n-k_0)!\right)^2Q_{n,k,k_0}\Per(X_0)^2
\label{eq:p1-normalized-no-AC}
\end{align}
by Lemma~\ref{lem:large-block-permanent-encoding}.
These are the parity counterparts of Eqs.~\eqref{eq:tbs-proof-p-def} and~\eqref{eq:tbs-proof-p1} in the threshold proof.

We now apply the same completion and extraction argument as in the proof of Theorem~\ref{thm: threshold average case quantitative}, replacing the square threshold block and margin by the rectangular parity block $U_{R,[n]}$ and the margin in Eq.~\eqref{eq:margin-no-AC}.
For each fixed $t$, Lemma~\ref{lem:completion} produces a unitary $U(t)$ containing $A(t)$ as its $2k \times n$ submatrix, except with exponentially small failure probability.
Choose
\begin{align}
\Delta
:=
\frac{c_\delta}{mn(1+\kappa\sqrt{k})}
\end{align}
for a sufficiently small constant $c_\delta>0$.
Lemmas~\ref{lem:sparse-TV} and~\ref{lem:completion}, together with the oracle guarantee in Eq.~\eqref{eq:oracle-no-AC}, imply that the rotated probabilities $P_x(V_R(\theta)U(t))$ can be estimated with additive imprecision $\varepsilon P_x^{(\mathrm{avg})}$ and failure probability below $1/8$ for every fixed $t\in[0,\Delta]$ and $\theta$.
Applying Theorem~\ref{thm:probabilistic-robust-fourier-extraction} followed by Theorem~\ref{thm:probabilistic-robust-BW-extrapolation} therefore gives an estimate $\widehat p(1)$ satisfying
\begin{align}
\Prob\left[
|\widehat p(1)-p(1)|
>
\varepsilon P_x^{(\mathrm{avg})}
\exp(2n\log\Delta^{-1}+O(n))
\right]
<
\frac13.
\label{eq:BW-final-prob-main}
\end{align}
Thus, after dividing by the known factor in Eq.~\eqref{eq:p1-normalized-no-AC}, the additive imprecision in the estimate of $\Per(X_0)^2$ is
\begin{align}
\varepsilon'
=
\frac{
\varepsilon P_x^{(\mathrm{avg})}
\exp(2n\log\Delta^{-1}+O(n))
}{
\kappa^{2n}\left((n-k_0)!\right)^2Q_{n,k,k_0}
}.
\label{eq: def of varepsilon'}
\end{align}

It remains to bound the parity-specific factors in Eq.~\eqref{eq: def of varepsilon'}.
Choose $\kappa=1/\sqrt{k}$.
Since $k=\Theta(n)$ on the event in Eq.~\eqref{eq:typical-k-no-AC},
\begin{align}
\Delta^{-1}=O(\alpha n^2),
\qquad
\kappa^{-2n}=\exp(n\log n+O(n)).
\end{align}
Moreover, Lemma~\ref{lem:Q-nkk-lower-bound} gives $Q_{n,k,k_0}^{-1}=e^{O(n)}$, while Stirling's formula gives
\begin{align}
\left((n-k_0)!\right)^2
=
\exp(2n\log n-O(n)).
\end{align}
Finally, Lemma~\ref{lem:typical-parity-Haar-weight-upper} in Appendix~\ref{appendix: section: upper bound of avg parity} gives
\begin{align}
P_x^{(\mathrm{avg})}
\le
\exp(-(1-\eta)n\log\alpha+O(n)).
\end{align}
Consequently,
\begin{align}
\varepsilon'
&\leq
\varepsilon
\exp\left(3n\log n+(1+\eta)n\log\alpha+O(n)\right)
\nonumber\\
&\leq
\varepsilon
\exp\left(3n\log n+1.01n\log\alpha+O(n)\right).
\end{align}

Therefore, if
\begin{align}
\varepsilon
\leq
\exp(-3n\log n-1.01n\log\alpha-O(n)),
\label{eq:epsilon bound final}
\end{align}
then the constant hidden in $O(n)$ can be chosen so that $\varepsilon'<1/2$, and rounding recovers $\Per(X_0)^2$ exactly.
The same repetition and majority-voting argument as in the threshold proof boosts the success probability, and every step is polynomial time with access to an NP oracle.
Hence the overall reduction lies in ${\rm BPP^{NP}}$.
\end{proof}

\section{Generalized Stockmeyer's reduction}\label{section: stockmeyer}

We have established average-case hardness of output-probability estimation for TBS and PBS in the regime $m=\Theta(n)$. 
A natural next question is whether these results can be translated into classical hardness results for the corresponding approximate-sampling tasks.

The standard route from probability-estimation hardness to approximate-sampling hardness is a reduction based on Stockmeyer’s approximate-counting algorithm~\cite{stockmeyer1985approximation}, as exploited in prior works~\cite{aaronson2011computational, bouland2019complexity, movassagh2023hardness, bouland2022noise, bouland2026complexity, bouland2025exponential, go2026computational}.
We therefore generalize this reduction to our TBS and PBS settings.
We first state the reduction for TBS: if there exists an approximate sampler for TBS within a suitably bounded total variation distance, then ATE in Definition~\ref{def: thresholdBS} can be solved by a $\rm{BPP^{NP}}$ procedure.

\begin{theorem}[Generalized Stockmeyer reduction]\label{thm: revised stockmeyers}
Suppose there exists a classical sampler whose output distribution $\widetilde{T}_x(U)$ satisfies, for every input unitary $U$, 
\begin{align}
    \frac{1}{2}
    \sum_{x\in\mathcal{T}_{m,n}}
    \left|
    T_x(U)-\widetilde{T}_x(U)
    \right|
    \leq
    \xi   .
\end{align}
Then for any polynomial $f(n)=\poly(n)$ and any failure probability $\delta>0$, there exists a $\mathrm{BPP}^{\mathrm{NP}}$ procedure that solves {\rm ATE} in Definition~\ref{def: thresholdBS} with failure probability at most $\delta$ and additive imprecision level
\begin{align}\label{eq: epsilon in Stockmeyer}
    \varepsilon
    =
    \frac{1}{\delta}
    \left(
    4\xi+\frac{1}{f(n)}
    \right).
\end{align}
\end{theorem}

We then state the corresponding corollary for the PBS setting and the APE problem.

\begin{corollary}[Stockmeyer reduction for PBS]\label{cor: companion stockmeyer PBS}
Suppose there exists a classical sampler whose output distribution $\widetilde{P}_x(U)$ satisfies, for every input unitary $U$,
\begin{align}
    \frac{1}{2}
    \sum_{x\in\mathcal{P}_{m,n}}
    \left|
    P_x(U)-\widetilde{P}_x(U)
    \right|
    \leq
    \xi   .
\end{align}
Then for any polynomial $f(n)=\poly(n)$ and any failure probability $\delta>0$, there exists a $\mathrm{BPP}^{\mathrm{NP}}$ procedure that solves {\rm APE} in Definition~\ref{def: parityBS} with failure probability at most $\delta$ and additive imprecision level $\varepsilon$ in Eq.~\eqref{eq: epsilon in Stockmeyer}.
\end{corollary}

The proof of Theorem~\ref{thm: revised stockmeyers} is provided in Appendix~\ref{appendix: section: stockmeyer}. 
Although the proof closely parallels that of~\cite{bouland2026complexity}, we include it here to keep the analysis self-contained.
The same result also applies to the PBS setting upon replacing the threshold probabilities and outcome ensemble with their parity counterparts.

From Eq.~\eqref{eq: epsilon in Stockmeyer} in Theorem~\ref{thm: revised stockmeyers}, for any $\varepsilon, \delta = \poly(n)^{-1}$, one can choose $\xi=\poly(n)^{-1}$ by taking $f(n)$ to be a sufficiently large polynomial.
Hence, by Theorem~\ref{thm: revised stockmeyers} and Toda's theorem $\rm{PH} \subseteq {\rm P^{\# P}}$~\cite{toda1991pp}, improving the imprecision bound in Theorem~\ref{thm: threshold average case main} to $\varepsilon = \poly(n)^{-1}$ would imply approximate-sampling hardness for TBS within total variation distance $\xi = \poly(n)^{-1}$, unless the polynomial hierarchy collapses to $\mathrm{BPP}^{\mathrm{NP}}$.
An analogous conclusion holds for PBS: strengthening Theorem~\ref{thm: parity average case main} to inverse-polynomial imprecision would, through the same reduction, establish approximate-sampling hardness for PBS.

Therefore, we arrive at the following conclusion.

\begin{theorem}\label{thm: hardness of APE implies sampling hardness}
    Suppose {\rm ATE} (\,{\rm APE}) in Definition~\ref{def: thresholdBS} (Definition~\ref{def: parityBS}) is proven to be $\#\mathrm{P}$-hard for $\varepsilon,\delta = \poly(n)^{-1}$ under any finite $\mathrm{PH}$ reduction.
    Then no classical algorithm can simulate TBS (PBS) within $\poly(n)^{-1}$ total variation distance, unless $\mathrm{PH}$ collapses to a finite level. 
\end{theorem}

This identifies the improvement of the allowable imprecision bound in the \#P-hardness results for ATE and APE as the remaining key challenge toward establishing full classical hardness of TBS and PBS.

\section{Alternative approach: leading-coefficient extraction}\label{section: alternative approach}

This section presents an alternative approach that can improve the imprecision bound of our average-case hardness results.
Throughout our main proofs, the worst-case instance is introduced via a $\Theta(n)$-dimensional block perturbation, which, in turn, makes the relevant Fourier coefficient a polynomial of degree $\Theta(n)$.
Polynomial interpolation followed by evaluation at the endpoint therefore incurs a precision loss of $\exp(\Theta(n\log n))$.

Here, we develop an alternative reduction strategy, inspired by~\cite{bouland2025exponential}, in which the worst-case instance is introduced through a more ``dilute'' lower-dimensional block perturbation.
As a result, the relevant Fourier coefficient becomes a lower-degree polynomial whose leading coefficient equals the target worst-case \#P-hard quantity multiplied by another Fourier coefficient.
One can therefore extract this leading coefficient directly, rather than polynomial interpolation followed by endpoint evaluation.
Under an anti-concentration assumption for the remaining Fourier coefficient, the approach yields an exponential improvement in the allowable additive imprecision.

We first state the result for TBS, which is the principal setting of this work.
The precise threshold Fourier-coefficient anti-concentration conjecture, including the randomized partial-rotation rule, is stated in Appendix~\ref{appendix: leading coefficient extraction}.

\begin{theorem}[Average-case hardness for TBS under anti-concentration]\label{thm:main-leading}
Fix any constant $0 < \lambda < 1$, and let $m=\alpha n$ for a fixed $\alpha>2$.
Assume that the threshold Fourier coefficient is anti-concentrated with respect to $T_{x}^{(\mathrm{avg})}$ over $x\sim\mathcal{D}_{m,n}$, $U\sim\Haar(m)$.
Then, the {\rm ATE} problem in Definition~\ref{def: thresholdBS} is $\#\mathrm{P}$-hard under a $\rm{BPP}$ reduction, with additive imprecision $\varepsilon=\exp(-O(n^{\lambda}))$ and failure probability $\delta=\poly(n)^{-1}$.
\end{theorem}

Note that the above imprecision bound is comparable to the bound achieved in~\cite{bouland2025exponential} for BosonSampling in the dilute regime, under a similar type of anti-concentration assumption that remains open.
A precise statement of the anti-concentration assumption and a detailed proof of the conditional result are provided in Appendix~\ref{appendix: leading coefficient extraction}.

\begin{proof}[Proof Sketch of Theorem~\ref{thm:main-leading}]
The proof proceeds as in the TBS reduction developed in Sec.~\ref{Section: threshold BS}, except that endpoint extrapolation by polynomial interpolation is replaced with leading-coefficient extraction.
Given a worst-case instance $X_0$ of size $k_0=O(n^{\lambda/2})$, sample a typical threshold outcome $x\in\mathcal T_{m,n}$ with $k=|x|$, and choose the partial-rotation decomposition
\begin{align}
    \supp(x)=Q\sqcup L,
    \qquad
    |Q|=2k-n,
    \qquad
    |L|=n-k,
\end{align}
with the rotated set $R=L\sqcup E$ as in Sec.~\ref{subsec: threshold fourier}.
Choose a subset $H\subseteq Q$ with $|H|=k_0$, set
\begin{align}
    n':=n-k_0,
    \qquad
    Q':=Q\setminus H,
\end{align}
and let $x'\in\mathcal T_{m,n'}$ be the residual threshold outcome obtained by removing the clicks in $H$ from $x$.
Then
\begin{align}
    2n'-2|x'|
    =
    2n-2|x|,
\end{align}
so the original and residual threshold probabilities use the same top Fourier frequency under the same rotated set $R$.

For a worst-case matrix $X_0\in\{0,1\}^{k_0\times k_0}$, introduce a sparse perturbation $A_{X_0}\in\mathbb C^{n\times n}$ supported only on the rows indexed by $H$ and on $k_0$ selected input columns.
For the Haar-random block $A=U_{W,[n]}$ appearing in the TBS Fourier coefficient, set $A(t)=A+tA_{X_0}$ and let $U(t)$ be a unitary completion of $A(t)$.
Define $N:=2n-2k$.
The single-permanent-collapse lemma in Appendix~\ref{appendix: leading coefficient extraction} then gives
\begin{align}
    [t^{2k_0}]\,\tau_{2N}(x,U(t),R)
    =
    m^{-k_0}\Per(X_0)^2
    \tau_{2N}(x',U,R).
\label{eq:main-tbs-leading-collapse}
\end{align}
Hence, to compute $\Per(X_0)^2$, we first estimate $\tau_{2N}(x,U(t),R)$ from average-case estimates of $T_x(V_R(\theta)U(t))$ via Fourier-coefficient extraction, and then estimate the coefficient of $t^{2k_0}$ via polynomial-coefficient extraction.
We separately estimate the denominator $\tau_{2N}(x',U,R)$ from the residual threshold instance $x'$.
The threshold residual-law lemma shows that the actual residual instance is dominated by the corresponding typical $n'$-photon TBS ensemble $\mathcal{D}_{m,n'}$.
Moreover, the Haar-average weights satisfy ${T_x^{(\mathrm{avg})}}/{T_{x'}^{(\mathrm{avg})}} = \alpha^{-k_0}e^{o(1)}$ for the typical $x \sim \mathcal{D}_{m,n}$ used in the reduction.
Together with the threshold Fourier-coefficient anti-concentration conjecture, these estimates imply that the ratio of the extracted leading coefficient and the residual Fourier coefficient estimates $m^{-k_0}\Per(X_0)^2$ with overall error amplification at most $\exp(O(k_0\log n))$.
Consequently, $\varepsilon=\exp(-O(n^\lambda))$ suffices to compute $\Per(X_0)^2$ exactly.
\end{proof}

For completeness, the same leading-coefficient method also gives the following analogous result for PBS under the parity Fourier-coefficient anti-concentration conjecture stated in Appendix~\ref{appendix: leading coefficient extraction}.

\begin{theorem}[Average-case hardness for PBS under anti-concentration]\label{thm:main-leading-parity}
Fix any constant $0 < \lambda < 1$, and let $m=\alpha n$ for a fixed $\alpha>2$.
Assume that the Fourier coefficient of PBS is anti-concentrated with respect to $P_{x}^{(\mathrm{avg})}$ over $x\sim\mathcal{Y}_{m,n}$ and $U\sim\Haar(m)$.
Then, the {\rm APE} problem in Definition~\ref{def: parityBS} is $\#\mathrm{P}$-hard under a $\rm{BPP}$ reduction, with additive imprecision $\varepsilon=\exp(-O(n^{\lambda}))$ and failure probability $\delta=\poly(n)^{-1}$.
\end{theorem}

\begin{proof}[Proof Sketch]
The proof uses the sparse parity perturbation developed in Appendix~\ref{appendix: leading coefficient extraction}.
After selecting $k_0$ odd/even pairs and setting $n'=n-k_0$, the coefficient of $t^{2k_0}$ in the residual parity Fourier coefficient factorizes as
\begin{align}
    [t^{2k_0}]
    c_{2n'}(x,U(t),R')
    =
    m^{-k_0}
    \Per(X_0)^2
    c_{2n'}(x',U,R').
\end{align}
The remaining argument applies the same elementary Fourier- and polynomial-coefficient extraction procedures, the parity residual-law lemma, and the parity Fourier-coefficient anti-concentration conjecture.
The resulting error amplification has the same scaling $\exp(O(k_0\log n))$, and therefore $\varepsilon=\exp(-O(n^\lambda))$ suffices.
\end{proof}


\section{Conclusion}\label{section: conclusion}

In this work, we established average-case hardness results for BosonSampling with threshold and parity measurements in the $m > 2n$ linear-mode regime, via a worst-to-average-case reduction based on Fourier-coefficient extraction.
We showed that the corresponding average-case output-probability estimation problems for TBS and PBS are \#P-hard up to an additive imprecision of $\exp(-3n\log n-O(n))$.
This scale matches the best known unconditional bound for standard PNR BosonSampling in the saturated regime~\cite{bouland2026complexity}.
Together with the generalized Stockmeyer reduction we established, these results identify the remaining imprecision gap: improving the average-case imprecision bound to inverse-polynomial scale would yield the standard approximate-sampling hardness consequences for TBS and PBS.
Finally, assuming anti-concentration of the relevant Fourier coefficient, we develop an alternative approach that strengthens the average-case hardness results for TBS and PBS by improving the tolerable additive imprecision to $\exp(-O(n^\lambda))$ for any fixed $0<\lambda<1$.

A crucial open problem is to improve the additive imprecision level for \#P-hardness from $\exp(-O(n\log n))$ to $\poly(n)^{-1}$, which remains open even for the standard BosonSampling setting.
As discussed in Sec.~\ref{section: stockmeyer}, such an improvement would close the gap between the presented average-case hardness results and the full approximate-sampling hardness result.
The leading-coefficient extraction approach in Theorem~\ref{thm:main-leading} and Theorem~\ref{thm:main-leading-parity} provides a viable route that could exponentially improve the tolerable additive imprecision from $\exp(-O(n\log n))$ to $\exp(-O(n^{\lambda}))$.
However, these results require model-specific anti-concentration guarantees for the threshold and parity Fourier coefficients, respectively, which remain crucial open problems.

Realistic noise in linear optical systems, such as photon loss and partial distinguishability of photons, can render BosonSampling classically simulable, as evidenced by numerous prior studies~\cite{rahimi2016sufficient, renema2018efficient,  moylett2019classically, shchesnovich2019noise,  garcia2019simulating, oszmaniec2018classical,  brod2020classical, oh2021classical, oh2023classical, oh2025classical, renema2018classical, renema2020simulability, van2024efficient, villalonga2021efficient, bulmer2022boundary, liu2023simulating, oh2024classical, umanskii2024classical, shi2022effect, bulmer2024simulating}. 
Hence, it would be interesting to develop classical simulation algorithms for threshold and parity BosonSampling under these types of noise, or to identify noise thresholds under which their hardness evidence persists, analogous to~\cite{aaronson2016bosonsampling, go2025quantum, go2025sufficient, go2026hardness}.

Shallow-depth circuits also offer a promising route for mitigating the accumulation of realistic noise.
To date, the classical complexity of shallow-depth BosonSampling has been studied extensively, with numerous works identifying conditions under which it becomes classically simulable~\cite{vidal2004efficient, garcia2019simulating, deshpande2018dynamical, qi2022efficient, oh2022classical, kolarovszki2023simulating, oh2025classical}, as well as regimes exhibiting evidence of classical hardness~\cite{go2024exploring, go2026computational}.
Extending these simulability and hardness results to threshold and parity measurements is a natural direction for future work.

\bibliographystyle{unsrt}
\bibliography{Reference.bib}

\clearpage

\onecolumngrid

\appendix

\section{Click-count concentration}\label{appendix: section: click count concentration}

We present an explicit proof of Lemma~\ref{lemma: most threshold outcomes have many clicks}, which states that the number of clicks concentrates near $\alpha n /(\alpha + 1)$ for $x \sim \mathcal{D}_{m,n}$ when $m = \alpha n$.

\begin{proof}[Proof of Lemma~\ref{lemma: most threshold outcomes have many clicks}]
Let $K:=|x|$. 
By Eq.~\eqref{eq:main-threshold-average}, for any $1\le k\le n$, we have
\begin{align}
\Pr[K=k]
=\binom{m}{k}\frac{\binom{n-1}{k-1}}{\binom{m+n-1}{n}}.
\end{align}
Hence, up to the global normalization factor, the probability mass at click count $k$ is given by
\begin{align}
B_k:=\binom{m}{k}\binom{n-1}{k-1}.
\end{align}
For $1\le k\le n-1$, a direct calculation gives
\begin{align}
\frac{B_{k+1}}{B_k}
=
\frac{m-k}{k+1}\frac{n-k}{k}.
\label{eq:threshold-ratio}
\end{align}
Writing $k=\gamma n$ and $m=\alpha n$, define
\begin{align}
    g_\alpha(\gamma):=\frac{(\alpha-\gamma)(1-\gamma)}{\gamma^2}.
\end{align}
Then Eq.~\eqref{eq:threshold-ratio} satisfies
\begin{align}
\frac{B_{k+1}}{B_k}
=\frac{g_\alpha(k/n)}{1+1/k}
\leq g_\alpha(k/n).
\end{align}
Since  $ g_\alpha(\gamma)-1 = \frac{\alpha-(\alpha+1)\gamma}{\gamma^2}$, the unique solution to $g_\alpha(\gamma)=1$ in $(0,1)$ is $\gamma_\star={\alpha}/({\alpha+1})$.
Also, we have  $g_\alpha'(\gamma) =\frac{(\alpha+1)\gamma-2\alpha}{\gamma^3}<0 $ for every $\gamma\in(0,1)$, since $\alpha>1$. 
Hence, $g_\alpha$ is strictly decreasing on $(0,1)$.

We first prove the upper-tail bound. 
Let $\gamma_1:=\gamma_\star+\frac{\eta}{2}$, $\gamma_0 := \gamma_\star +\eta$, and 
\begin{align}\label{eq:tbs-def-q-alpha-eta}
    q_{\alpha,\eta}:=g_\alpha(\gamma_1).
\end{align}
Since $\gamma_1>\gamma_\star$, we have $q_{\alpha,\eta} <1$, and for all $k\ge \gamma_1 n$, Eq.~\eqref{eq:threshold-ratio} gives $\frac{B_{k+1}}{B_k}\le q_{\alpha,\eta}$.
Let $k_1:=\lceil \gamma_1n\rceil$ and $k_0:=\lceil \gamma_0n\rceil$. 
Then we have $B_{k_0}\le q_{\alpha,\eta}^{k_0-k_1}B_{k_1}$, and therefore
\begin{align}
\Pr[K\ge \gamma_0 n]
&\le
\frac{\sum_{k=k_0}^{n}B_k}{\sum_{k=1}^{n}B_k}
\le
\frac{B_{k_0}}{(1-q_{\alpha,\eta})B_{k_1}}
\le
\frac{q_{\alpha,\eta}^{k_0-k_1}}{1-q_{\alpha,\eta}}
\le
\frac{1}{q_{\alpha,\eta}(1-q_{\alpha,\eta})}
\exp\left(-\frac{\eta}{2}|\log q_{\alpha,\eta}|n\right),
\label{eq:tbs-upper-tail}
\end{align}
where the last inequality uses $k_0-k_1\ge \eta n/2-1$.

We next prove the lower-tail bound. 
For $2\le k\le n$, another direct calculation gives
\begin{align}
\frac{B_{k-1}}{B_k}
=
\frac{k}{m-k+1}\frac{k-1}{n-k+1}
\le
\frac{(k/n)^2}{(\alpha-k/n)(1-k/n)}.
\label{eq:threshold-backward-ratio}
\end{align}
Let $\gamma_1':=\gamma_\star-\frac{\eta}{2}  $, $   \gamma_0':=\gamma_\star-\eta $, and 
\begin{align}\label{eq:tbs-def-r-alpha-eta}
    r_{\alpha,\eta}:=\frac{1}{g_\alpha(\gamma_1')}.
\end{align}
Since $\gamma_1'<\gamma_\star$, we have $r_{\alpha,\eta} <1$, and for every $k\le \gamma_1'n$, Eq.~\eqref{eq:threshold-backward-ratio} gives $\frac{B_{k-1}}{B_k}\le r_{\alpha,\eta}$.
Let $k_1':=\lfloor \gamma_1'n\rfloor$ and $k_0':=\lfloor \gamma_0'n\rfloor$. 
Then we have $B_{k_0'}\le r_{\alpha,\eta}^{k_1'-k_0'}B_{k_1'}$, and hence
\begin{align}
\Pr[K\le \gamma_0'n]
&\le
\frac{\sum_{k=1}^{k_0'}B_k}{\sum_{k=1}^{n}B_k}
\le
\frac{B_{k_0'}}{(1-r_{\alpha,\eta})B_{k_1'}}
\le
\frac{r_{\alpha,\eta}^{k_1'-k_0'}}{1-r_{\alpha,\eta}}
\le
\frac{1}{r_{\alpha,\eta}(1-r_{\alpha,\eta})}
\exp\left(-\frac{\eta}{2}|\log r_{\alpha,\eta}|n\right),
\label{eq:tbs-lower-tail}
\end{align}
where we used $k_1'-k_0'\ge \eta n/2-1$.

Finally, using the union bound on Eq.~\eqref{eq:tbs-upper-tail} and Eq.~\eqref{eq:tbs-lower-tail}, we obtain
\begin{align}
    \Pr_{x\sim \mathcal{D}_{m,n}}
    \left[ \left| |x| - \frac{\alpha}{\alpha+1}n \right| \geq \eta n\right]
    \le C_{\alpha,\eta}e^{-c_{\alpha,\eta}n}  ,
\end{align}
for the constants given by 
\begin{align}
C_{\alpha,\eta} = 2 \max \left\{ \frac{1}{q_{\alpha,\eta}(1-q_{\alpha,\eta})} ,  \frac{1}{r_{\alpha,\eta}(1-r_{\alpha,\eta})}\right\},
\quad
c_{\alpha,\eta} = \frac{\eta}{2}\,\min\left\{ |\log q_{\alpha,\eta} | , | \log r_{\alpha,\eta}| \right\} ,
\end{align}
with $q_{\alpha,\eta} \in (0,1)$ and $r_{\alpha,\eta} \in (0,1)$ defined in Eq.~\eqref{eq:tbs-def-q-alpha-eta} and Eq.~\eqref{eq:tbs-def-r-alpha-eta}, respectively. 
This concludes the proof.

\end{proof}

\section{Upper bound on Haar-averaged threshold output probabilities}\label{subsec: threshold Haar upper}

\begin{lemma}[Upper bound on the Haar-averaged threshold output probability]\label{lem:typical-threshold-Haar-weight-upper}
Let $m=\alpha n$ with $\alpha>1$ and $\alpha\le \poly(n)$. 
For every $x\in\mathcal{T}_{m,n}$,
\begin{align}
    T_x^{(\mathrm{avg})}
    \le
    \exp(-n\log\alpha+O(n)).
\label{eq:tbs-Haar-upper}
\end{align}
\end{lemma}

\begin{proof}
Let $k=|x|$. 
By Eq.~\eqref{eq:main-threshold-average},
\begin{align}
    T_x^{(\mathrm{avg})}
    =
    \frac{\binom{n-1}{k-1}}{\binom{m+n-1}{n}}.
\end{align}
First, the numerator satisfies
\begin{align}
    \binom{n-1}{k-1}\le 2^n.
\label{eq:tbs-numerator-upper}
\end{align}
For the denominator, using $m=\alpha n$ gives
\begin{align}
\binom{m+n-1}{n}
&=
\frac{m}{m+n}\binom{m+n}{n} \\
&=
\frac{\alpha}{\alpha+1}\binom{(\alpha+1)n}{n} \\
&\ge
\frac{1}{(\alpha+1)n+1}
\exp\left(n\left((\alpha+1)\log(\alpha+1)-\alpha\log\alpha\right)\right)  ,
\end{align}
where we use the standard binomial coefficient bounds in terms of entropy~\cite{polyanskiy2025information}:
\begin{align}
\frac{1}{a+1}
\exp\left(
a H\left(\frac{b}{a}\right)
\right) \leq \binom{a}{b}
\le
\exp\left(
a H\left(\frac{b}{a}\right)
\right)    ,
\label{eq:entropy-upper-binomial}
\end{align}
for the binary entropy function $H(p):=-p\log p-(1-p)\log(1-p).$
Now using 
\begin{align}
(\alpha+1)\log(\alpha+1)-\alpha\log\alpha
=
\log\alpha+(\alpha+1)\log\left(1+\frac{1}{\alpha}\right)
\ge
\log\alpha  ,
\end{align}
we obtain
\begin{align}
\binom{m+n-1}{n}
\ge
\frac{\alpha}{\alpha+1}\frac{1}{(\alpha+1)n+1}e^{n\log\alpha}.
\label{eq:tbs-denominator-lower}
\end{align}
Combining Eq.~\eqref{eq:tbs-numerator-upper} and Eq.~\eqref{eq:tbs-denominator-lower}, it holds that 
\begin{align}
T_x^{(\mathrm{avg})}
&\le
2^n\frac{\alpha+1}{\alpha}((\alpha+1)n+1)e^{-n\log\alpha}
\le
\exp(-n\log\alpha+O(n)),
\end{align}
where the last step uses $\alpha\le\poly(n)$, so that the polynomial prefactor is absorbed into the $O(n)$ term.
\end{proof}

\section{Odd-parity concentration}\label{appendix: section: odd parity concentration}

We provide an explicit proof of Lemma~\ref{lemma: most parity outcomes have many odd parities}, showing that, for $x\sim\mathcal{Y}_{m,n}$ with $m=\alpha n$, the number of odd parities is concentrated near $\alpha n/(\alpha+2)$.

\begin{proof}[Proof of Lemma~\ref{lemma: most parity outcomes have many odd parities}]

Let $K = |x|$.
By definition of $\mathcal{Y}_{m,n}$, for any $k \leq n$ with $ k \equiv n \; (\text{mod  2}) $, we have
\begin{align}
    \Pr(K=k)
    =
    \binom{m}{k}P_{x}^{(\mathrm{avg})} 
    =
    \frac{
    \binom{m}{k}
    \binom{m+\frac{n-k}{2}-1}{\frac{n-k}{2}}
    }{
    \binom{m+n-1}{n}
    }.
\end{align}
Rewrite $k = n - 2t$ for $t\in\{0,1,\dots,\lfloor n/2\rfloor\}$, and define
\begin{align}
    B_t:=\binom{m}{n-2t}\binom{m+t-1}{t}
    \quad \Rightarrow  \quad
    \Pr(K=n-2t)=\frac{B_t}{\binom{m+n-1}{n}}.
\end{align}

We first examine the lower-tail bound. 
A direct computation gives, for $t \geq 1$, 
\begin{align}\label{eq: fraction of A_t+1 and A_t}
    \frac{B_{t+1}}{B_t}
=
\frac{(n-2t)(n-2t-1)(m+t)}{(t+1)(m-n+2t+1)(m-n+2t+2)}
\le
    \frac{(n-2t)^2(m+t)}{t(m-n+2t)^2}.
\end{align}
As $m=\alpha n$, the right-hand side can be written as $f_\alpha\!\left(t/n\right)$ for 
\begin{align}\label{eq: def of f_alpha}
f_\alpha(\beta):=
\frac{(1-2\beta)^2(\alpha+\beta)}{\beta(\alpha-1+2\beta)^2}   ,
\end{align}
which further implies
\begin{align}\label{eq: f_alpha - 1}
f_\alpha(\beta)-1
=
-\frac{\alpha\bigl((\alpha+2)\beta-1\bigr)}
{\beta(\alpha-1+2\beta)^2}.
\end{align}
By the above, it is clear that $f_\alpha(\beta)<1 $ whenever $\beta> \beta_\star := \frac{1}{\alpha+2}$.
Now set $\beta_1:=\beta_\star+\frac{\eta}{4}$ and $
\beta_0:=\beta_\star+\frac{\eta}{2}$.
Because $\eta < \alpha/(\alpha + 2)$, we have $\beta_1<\beta_0<1/2$, and the event
\begin{align}
K = n - 2t \le \left(\frac{\alpha}{\alpha+2}-\eta\right)n
\end{align}
is equivalent to $t\ge \beta_0 n$.
From Eq.~\eqref{eq: def of f_alpha}, define
\begin{align}\label{eq: def of q_alpha-eta}
q_{\alpha,\eta}
:=
\max_{\beta\in[\beta_1,\,1/2]} f_\alpha(\beta).
\end{align}
Because $f_\alpha(\beta)<1$ for all $\beta>\beta_\star=1/(\alpha+2)$, we have $0<q_{\alpha,\eta}<1$. 
Also, for every integer $t\ge \beta_1 n$, we have $\frac{B_{t+1}}{B_t}\le q_{\alpha,\eta}$.
Let $t_1:=\left\lceil \beta_1 n\right\rceil$ and $t_0:=\left\lceil \beta_0 n\right\rceil$.
When $(1/\alpha + \eta/2)n \geq 2$ as we assumed in Lemma~\ref{lemma: most parity outcomes have many odd parities}, we have $t_1\geq 1$ (required in Eq.~\eqref{eq: fraction of A_t+1 and A_t}), thus $B_{t_0}\le q_{\alpha,\eta}^{\,t_0-t_1}B_{t_1}$.
Summing the tail terms above $t = t_0$ gives
\begin{align}
\sum_{t=t_0}^{\lfloor n/2\rfloor} B_t
\le
\frac{B_{t_0}}{1-q_{\alpha,\eta}}
\le
\frac{q_{\alpha,\eta}^{\,t_0-t_1}}{1-q_{\alpha,\eta}}\,B_{t_1}.
\end{align}
Since $B_{t_1} \leq \sum_{s = 0}^{\lfloor n/2\rfloor} B_s = \binom{m+n-1}{n} $, it follows that
\begin{align}\label{eq: lower bound of |x| parity}
\Pr\!\left[
K\le \left(\frac{\alpha}{\alpha+2}-\eta\right)n
\right]
\leq \frac{\sum_{t=t_0}^{\lfloor n/2\rfloor} B_t}{\binom{m+n-1}{n}}
\leq
\frac{q_{\alpha,\eta}^{\,t_0-t_1}}{1-q_{\alpha,\eta}}
\le
\frac{1}{q_{\alpha,\eta}(1-q_{\alpha,\eta})}
\exp\!\left(
-\frac{\eta}{4}\,|\log q_{\alpha,\eta}|\,n
\right) ,
\end{align}
using $t_0-t_1\ge (\eta/4)n-1$ and $0<q_{\alpha,\eta}<1$.


We similarly examine the upper-tail bound. 
A direct calculation gives, for $t < n/2$, 
\begin{align}
\frac{B_{t-1}}{B_t}
=
\frac{t(m-n+2t)(m-n+2t-1)}{(n-2t+2)(n-2t+1)(m+t-1)}
\le
\frac{t(m-n+2t)^2}{(n-2t)^2(m+t)} 
=
f_{\alpha}\!\left(\frac{t}{n}\right)^{-1}.
\end{align}
for $f_\alpha$ defined in Eq.~\eqref{eq: def of f_alpha}. 
From Eq.~\eqref{eq: f_alpha - 1}, we have
\begin{align}
f_\alpha(\beta)^{-1}-1
=
\frac{\alpha\bigl((\alpha+2)\beta-1\bigr)}
{(1-2\beta)^2(\alpha+\beta)} ,
\end{align}
such that $f_\alpha(\beta)^{-1} < 1 $ whenever $\beta < 1/({\alpha+2}) = \beta_\star$.
We newly set $\beta_1':=\beta_\star - \frac{\eta}{4} $ and $\beta_0':=\beta_\star - \frac{\eta}{2} $, where by the constraint $\eta < 2/(\alpha + 2)$ we have $0 < \beta_0' < \beta_1'$.
Then, the event
\begin{align}
K = n - 2t \geq \left(\frac{\alpha}{\alpha+2} + \eta\right) n
\end{align}
is equivalent to $t \leq  \beta_0' n$. 
From Eq.~\eqref{eq: def of f_alpha}, let
\begin{align}\label{eq: def of r_alpha-eta}
r_{\alpha,\eta}
:=
\max_{\beta\in[0,\,\beta_1']} f_\alpha(\beta) ^{-1} ,
\end{align}
such that $0 < r_{\alpha,\eta}<1$. 
Also, for every integer $t \leq  \beta_1' n$, it follows that $\frac{B_{t - 1}}{B_t}\le r_{\alpha,\eta}$.
Now let $t_1':=\left\lfloor \beta_1' n\right\rfloor$ and $t_0':=\left\lfloor \beta_0' n\right\rfloor $, such that $B_{t_0'}\le r_{\alpha,\eta}^{\,t_1'-t_0'}B_{t_1'}$.
This implies 
\begin{align}
\Pr\!\left[
K \geq \left(\frac{\alpha}{\alpha+2} + \eta\right)n
\right]
\leq 
\frac{\sum_{t=0}^{t_0'} B_t}{\binom{m+n-1}{n}} 
\leq
\frac{1}{1-r_{\alpha,\eta}}\frac{B_{t_0'}}{\binom{m+n-1}{n}} 
\leq
\frac{r_{\alpha,\eta}^{\,t_1'-t_0'}}{1-r_{\alpha,\eta}}\frac{B_{t_1'}} {\binom{m+n-1}{n}} 
\leq
\frac{r_{\alpha,\eta}^{\,t_1'-t_0'}}{1-r_{\alpha,\eta}} ,
\end{align}
and using $t_1'-t_0'\ge (\eta/4)n-1$ and $0<r_{\alpha,\eta}<1$, we have
\begin{align}\label{eq: upper bound of |x| parity}
\Pr\!\left[
K \geq \left(\frac{\alpha}{\alpha+2} + \eta\right)n
\right] 
\le
\frac{1}{r_{\alpha,\eta}(1-r_{\alpha,\eta})}
\exp\!\left(
-\frac{\eta}{4}\,|\log r_{\alpha,\eta}|\,n
\right) .
\end{align}

Combining Eq.~\eqref{eq: lower bound of |x| parity} and Eq.~\eqref{eq: upper bound of |x| parity} using union bounds, it holds that 
\begin{align}
    \Pr_{x\sim \mathcal{Y}_{m,n}}
    \left[ \left| |x| - \frac{\alpha}{\alpha+2}n \right| \geq \eta n\right]
    \le C_{\alpha,\eta}e^{-c_{\alpha,\eta}n}  .
\end{align}
The constants are given by 
\begin{align}
C_{\alpha,\eta} = 2 \max \left\{ \frac{1}{q_{\alpha,\eta}(1-q_{\alpha,\eta})} ,  \frac{1}{r_{\alpha,\eta}(1-r_{\alpha,\eta})}\right\},
\quad
c_{\alpha,\eta} = \frac{\eta}{4}\,\min\left\{ |\log q_{\alpha,\eta} | , | \log r_{\alpha,\eta}| \right\} ,
\end{align}
for $q_{\alpha,\eta}$ and $r_{\alpha,\eta}$ in Eq.~\eqref{eq: def of q_alpha-eta} and Eq.~\eqref{eq: def of r_alpha-eta}, respectively.

\end{proof}

\section{Proof of Lemma~\ref{lem:Q-nkk-lower-bound}}\label{appendix: section: single permanent into coefficient - PBS}

We lower-bound $Q_{n,k,k_0}$ by restricting the summation in Eq.~\eqref{eq:Q-nkk-definition} to a single family of occupation patterns.
Choose $a_j = 0$ for every $j = 1,\ldots,k-k_0$, such that the constraint gives
\begin{align}
\sum_{j=1}^{k-k_0} b_j = \frac{n-k}{2}.
\end{align}
We then take
\begin{align}
b_j=
\begin{cases}
1, & j=1,\ldots,\frac{n-k}{2},\\
0, & j= \frac{n-k}{2} +1,\ldots,k-k_0.
\end{cases}
\end{align}
which is valid provided that $(n-k)/2 \le k-k_0$.
Specifically, when $k$ satisfies Eq.~\eqref{eq: k lwer bound}, we have 
\begin{align}\label{eq: upper bound for M}
\frac{n-k}{2}
\le
\left(
\frac{1}{\alpha+2}+\frac{\eta}{2}
\right)n,
\end{align}
whereas
\begin{align}
k-k_0
\ge
\left(
\frac{\alpha}{\alpha+2}-\eta
\right)n - o(n) .
\end{align}
Since $\alpha>2$, choosing 
\begin{align}
    \eta \leq \frac{1}{6} < \frac{2}{3}\frac{\alpha - 1}{\alpha + 2}
\end{align}
ensures $(n-k)/2 \le k-k_0$ for all sufficiently large $n$.

Provided $(n-k)/2 \le k-k_0$, the corresponding factorial term for the above pattern is
\begin{align}
\prod_{j=1}^{k-k_0}
\frac{1}{(2a_j+1)!(2b_j)!} 
= \prod_{j=1}^{(n-k)/2}\frac{1}{1!\,2!}
\prod_{j=(n-k)/2+1}^{k-k_0}\frac{1}{1!\,0!} 
=
2^{-(n-k)/2}.
\end{align}
Since all terms in $Q_{n,k,k_0}$ are nonnegative, we have
\begin{align}
Q_{n,k,k_0}
\ge
2^{-(n-k)/2}.
\end{align}
Finally, using Eq.~\eqref{eq: upper bound for M} gives
\begin{align}
Q_{n,k,k_0}
\ge
\exp\left[
-\left(
\frac{1}{\alpha+2}+\frac{\eta}{2}
\right)(\log 2)n
\right],
\end{align}
proving the claim.

\section{Upper bound on Haar-averaged parity output probabilities}\label{appendix: section: upper bound of avg parity}

\begin{lemma}
\label{lem:typical-parity-Haar-weight-upper}
Let $m=\alpha n$ with $\alpha>2$. Let $x\in \mathcal{P}_{m,n}$ be a parity outcome with
$k=|x|$, and suppose that
\begin{align}
\left|
k-\frac{\alpha}{\alpha+2}n
\right|
\le
\eta n ,
\label{eq:typical-k-assumption-Haar-weight}
\end{align}
where $0<\eta< 2 / (\alpha + 2)$ is a constant. 
Then there exists a constant $C_\eta>0$  such that
\begin{align}
P_x^{(\mathrm{avg})}
\le
\exp\left(
-(1-\eta)n\log\alpha
+
C_\eta n
+
\log((\alpha+1)n+1)
\right).
\label{eq:typical-parity-Haar-weight-upper-explicit}
\end{align}
In particular, as far as $\alpha \leq \poly(n)$, we have
\begin{align}
P_x^{(\mathrm{avg})}
\le
\exp\left(
-(1-\eta)n\log\alpha
+
O(n)
\right).
\label{eq:typical-parity-Haar-weight-upper-simple}
\end{align}
\end{lemma}

\begin{proof}
For a parity outcome $x$ with $k=|x|$, write $t:=({n-k})/{2}$ as in the proof of Lemma~\ref{lemma: most parity outcomes have many odd parities}.
By the formula for the Haar-averaged parity output probability, we have
\begin{align}
P_x^{(\mathrm{avg})}
=
\frac{
\binom{m+t-1}{t}
}{
\binom{m+n-1}{n}
}.
\label{eq:Haar-parity-weight-start}
\end{align}
Since $m=\alpha n$, the assumption Eq.~\eqref{eq:typical-k-assumption-Haar-weight} implies
\begin{align}\label{eq:upperbound of r}
t
=
\frac{n-k}{2}
\le
\frac{n}{\alpha+2}+\frac{\eta n}{2}.
\end{align}
Keeping the notation used in the proof of Lemma~\ref{lemma: most parity outcomes have many odd parities}, let us denote  $\beta:=\frac{t}{n}$ and $\beta_0:=\frac{1}{\alpha+2}+\frac{\eta}{2} $, such that Eq.~\eqref{eq:upperbound of r} implies $\beta\le \beta_0 $.

First, since $m+t=(\alpha+\beta)n$, the standard binomial coefficient bounds~\cite{polyanskiy2025information} in Eq.~\eqref{eq:entropy-upper-binomial}
gives
\begin{align}
\binom{m+t-1}{t}
\le
\binom{m+t}{t}
\le
\exp\left(
(\alpha+\beta)n
H\left(
\frac{\beta}{\alpha+\beta}
\right)
\right)
=
\exp\left(
n\left[
(\alpha+\beta)\log(\alpha+\beta)
-\beta\log\beta
-\alpha\log\alpha
\right]
\right).
\label{eq:numerator-Haar-weight-bound}
\end{align}
Second,
\begin{align}
\binom{m+n-1}{n}
=
\frac{m}{m+n}
\binom{m+n}{n}
=
\frac{\alpha}{\alpha+1}
\binom{(\alpha+1)n}{n}.
\end{align}
Therefore, similarly using the binomial coefficient bound in Eq.~\eqref{eq:entropy-upper-binomial}, we have 
\begin{align}
\binom{m+n-1}{n}
&\ge
\frac{\alpha}{\alpha+1}
\frac{1}{(\alpha+1)n+1}
\exp\left(
(\alpha+1)n
H\left(
\frac{1}{\alpha+1}
\right)
\right)
\nonumber\\
&=
\frac{\alpha}{\alpha+1}
\frac{1}{(\alpha+1)n+1}
\exp\left(
n\left[
(\alpha+1)\log(\alpha+1)
-\alpha\log\alpha
\right]
\right).
\label{eq:denominator-Haar-weight-bound}
\end{align}
Combining Eqs.~\eqref{eq:Haar-parity-weight-start},
\eqref{eq:numerator-Haar-weight-bound}, and
\eqref{eq:denominator-Haar-weight-bound}, we obtain
\begin{align}
P_x^{(\mathrm{avg})}
\le
\frac{\alpha+1}{\alpha}
\left((\alpha+1)n+1\right)
\exp\left(
nF_\alpha(\beta)
\right),
\label{eq:Haar-weight-F-alpha}
\end{align}
where, for convenience, we define
\begin{align}
F_\alpha(\beta)
:=
(\alpha+\beta)\log(\alpha+\beta)
-\beta\log\beta
-(\alpha+1)\log(\alpha+1).
\label{eq:F-alpha-definition}
\end{align}
Here, taking the derivative gives 
\begin{align}\label{eq: derivative of F-alpha}
F_\alpha'(\beta)
=
\log\left(\frac{\alpha+\beta}{\beta}\right)>0  ,
\end{align}
implying that $F_\alpha$ is increasing.
Since $\beta\le \beta_0 $, we have $F_\alpha(\beta)\le F_\alpha(\beta_0)$, leading to 
\begin{align}
P_x^{(\mathrm{avg})}
\le
\frac{\alpha+1}{\alpha}
\left((\alpha+1)n+1\right)
\exp\left(
nF_\alpha(\beta_0)
\right).
\label{eq:Haar-weight-rho-plus}
\end{align}

Now let $\beta_\star := {1}/(\alpha+2) = \beta_0 - \eta/2$. 
At $\beta_\star$, a direct simplification gives $-F_\alpha(\beta_\star) = \chi_\alpha$ for 
\begin{align}
\chi_\alpha
:=
\alpha\log(\alpha+2)
-
\frac{\alpha(\alpha+1)}{\alpha+2}\log(\alpha+1).
\label{eq:chi-alpha-definition}
\end{align}
By Eq.~\eqref{eq: derivative of F-alpha}, for any $\beta\in[\beta_\star,\beta_0]$, we have 
\begin{align}
F_\alpha'(\beta)
\le
F_\alpha'(\beta_\star)
=
\log\left(
\frac{
\alpha+\frac{1}{\alpha+2}
}{
\frac{1}{\alpha+2}
}
\right)
=
2\log(\alpha+1).
\end{align}
Since the derivative $F_\alpha'(\beta)$ is decreasing, this leads to 
\begin{align}
F_\alpha(\beta_0)
\le
F_\alpha(\beta_\star)
+
\frac{\eta}{2}\cdot 2\log(\alpha+1)
=
-\chi_\alpha+\eta\log(\alpha+1).
\label{eq:F-alpha-rho-plus-bound}
\end{align}
Substituting Eq.~\eqref{eq:F-alpha-rho-plus-bound} into
Eq.~\eqref{eq:Haar-weight-rho-plus} gives
\begin{align}
P_x^{(\mathrm{avg})}
\le
\frac{\alpha+1}{\alpha}
\left((\alpha+1)n+1\right)
\exp\left(
-\left[
\chi_\alpha-\eta\log(\alpha+1)
\right]n
\right).
\label{eq:Haar-weight-chi-bound}
\end{align}

It remains to lower bound $\chi_\alpha$ in terms of $\log\alpha$. Rewriting
Eq.~\eqref{eq:chi-alpha-definition},
\begin{align}
\chi_\alpha
&=
\alpha\log\alpha
+
\alpha\log\left(1+\frac{2}{\alpha}\right)
-
\frac{\alpha(\alpha+1)}{\alpha+2}
\left[
\log\alpha+\log\left(1+\frac{1}{\alpha}\right)
\right]
\\
&=
\frac{\alpha}{\alpha+2}\log\alpha
+
\alpha\log\left(1+\frac{2}{\alpha}\right)
-
\frac{\alpha(\alpha+1)}{\alpha+2}
\log\left(1+\frac{1}{\alpha}\right).
\label{eq:chi-alpha-rewrite}
\end{align}
The second term in the right-hand side of Eq.~\eqref{eq:chi-alpha-rewrite} is nonnegative. 
Also, the last term is bounded by
\begin{align}
\frac{\alpha(\alpha+1)}{\alpha+2}
\log\left(1+\frac{1}{\alpha}\right)
\le
\frac{\alpha(\alpha+1)}{\alpha+2}\cdot \frac{1}{\alpha}
=
\frac{\alpha+1}{\alpha+2}
\le 1.
\end{align}
Hence $\chi_\alpha$ is lower bounded by
\begin{align}
\chi_\alpha
\ge
\frac{\alpha}{\alpha+2}\log\alpha-1
=
\log\alpha-\frac{2}{\alpha+2}\log\alpha - 1.
\label{eq:chi-alpha-lower-first}
\end{align}
We also have 
\begin{align}
\frac{2}{\alpha+2}\log\alpha \leq W(2/e)  ,
\end{align}
for $W$ denoting the Lambert W function, which gives
\begin{align}
\chi_\alpha\ge \log\alpha-1 - W(2/e) .
\label{eq:chi-alpha-log-alpha-lower}
\end{align}
Furthermore, observe that 
\begin{align}
\eta\log(\alpha+1)
\le
\eta\log\alpha+\eta\log 2.
\label{eq:eta-log-alpha-plus-one}
\end{align}
Combining Eqs.~\eqref{eq:chi-alpha-log-alpha-lower} and
\eqref{eq:eta-log-alpha-plus-one}, we get
\begin{align}
\chi_\alpha-\eta\log(\alpha+1)
\ge
(1-\eta)\log\alpha-(1 + W(2/e) +\eta\log 2).
\label{eq:chi-eta-final-lower}
\end{align}
Substituting Eq.~\eqref{eq:chi-eta-final-lower} into
Eq.~\eqref{eq:Haar-weight-chi-bound} gives
\begin{align}
P_x^{(\mathrm{avg})}
&\le
\frac{\alpha+1}{\alpha}
\left((\alpha+1)n+1\right)
\exp\left(
-(1-\eta)n\log\alpha
+
(1 + W(2/e) +\eta\log 2)n
\right)
\nonumber\\
&\le
\exp\left(
-(1-\eta)n\log\alpha
+
(1 + W(2/e) +\eta\log 2)n
+
\log((\alpha+1)n+1)
+
\log 2
\right),
\label{eq:Haar-weight-final-explicit}
\end{align}
where in the last line we used $(\alpha+1)/\alpha\le 2$ for $\alpha>1$.

This proves Eq.~\eqref{eq:typical-parity-Haar-weight-upper-explicit}, for example with $C_\eta = 3+\eta\log 2$. 
If $\log((\alpha+1)n)=O(n)$, then the final logarithmic prefactor in Eq.~\eqref{eq:Haar-weight-final-explicit} is absorbed into $O(n)$, giving Eq.~\eqref{eq:typical-parity-Haar-weight-upper-simple}.
\end{proof}



\section{Robust Fourier coefficient extraction}\label{appendix: section: robust Fourier extraction}

We introduce the robust Fourier-coefficient extraction algorithm stated in Theorem~\ref{thm:probabilistic-robust-fourier-extraction} of the main text.

\begin{proof}[Proof of Theorem~\ref{thm:probabilistic-robust-fourier-extraction}]
Let
\begin{align}
\rho_L:=
\frac{L-2d-1}{2L}.
\label{eq:rho-F-def}
\end{align}
Using $L=C_Fd$ for $C_F=\left\lceil\max\left\{6,3/(1-8p_{\mathrm{fail}})\right\}\right\rceil$ and $d>1$, we have $2d+1<3d$ and
\begin{align}
    \frac{p_{\mathrm{fail}}}{\rho_L}
    <
    \frac{2L}{L - 3d}p_{\mathrm{fail}}
    =
    \frac{p_{\mathrm{fail}}}{\frac{1}{2} - \frac{3}{2C_F}} 
    \leq 
    \frac{p_{\mathrm{fail}}}{\frac{1}{2} - \frac{3}{2}\frac{1 - 8 p_{\mathrm{fail}} }{3} } 
    = 
    \frac{1}{4} .
\end{align}

We now use the convention $z := e^{i\theta}$, and define $w_j:=z_j^{d}y_j$ for $z_j:=e^{i\theta_j}$. 
Since $|z_j|=1$, Eq.~\eqref{eq:fourier-input-failure-prob} implies
\begin{align}
\Prob\left[
|w_j-q(z_j)|>\epsilon
\right]\le p_{\mathrm{fail}}   ,
\end{align}
for $q(z)=z^{d}f(\theta)$. 
Let $B_j$ be the indicator of the event $|w_j-q(z_j)|>\epsilon$, and set
\begin{align}
B:=\frac{1}{L}\sum_{j=0}^{L-1}B_j.
\end{align}
Then, $\E[B]\le p_{\mathrm{fail}}$, and Markov's inequality gives 
\begin{align}
\Prob[B>\rho_L]\le \frac{p_{\mathrm{fail}}}{\rho_L}<\frac14.
\label{eq:fourier-Markov}
\end{align}
On the event $B\le \rho_L$, the inlier set $G_F:=\{j:|w_j-q(z_j)|\le\epsilon\}$ has size at least
\begin{align}
|G_F|
\ge
L-\rho_L L
=
\frac{L+2d+1}{2}.
\label{eq:fourier-inlier-count}
\end{align}

Using an NP oracle, we can construct a polynomial $\widetilde q(z)$ of degree at most $2d$ and a set $\widetilde G_F\subseteq\{0,\ldots, L-1\}$ such that
\begin{align}
|\widetilde G_F|
\ge
\frac{L+2d+1}{2}   , 
\end{align}
while for every $j\in \widetilde G_F$,
\begin{align}
|w_j-\widetilde q(z_j)|\le \epsilon   .
\end{align}
Note that finding such a pair $(\widetilde q,\widetilde G_F)$ is an NP problem, since a proposed witness can be verified efficiently.
Moreover, such a witness surely exists because the true polynomial $q$ and the true inlier set $G_F$ satisfy the same conditions.


By construction, the intersection $G_F\cap \widetilde G_F$ has size at least $2d+1$.
For every $j$ in this intersection, triangle inequality gives
\begin{align}
|q(z_j)-\widetilde q(z_j)| \leq |q(z_j) - w_j| + |w_j-\widetilde q(z_j)|\le 2\epsilon.
\end{align}
Moreover, $C_F\geq6$ and $d>1$ imply $L=C_Fd\geq6d\geq4d+2=2(2d+1)$, so the sampling condition in Lemma~\ref{lem:roots-conditioning} holds when the lemma is applied to a degree-$2d$ polynomial.
Applying Lemma~\ref{lem:roots-conditioning} with $d$ replaced by $2d$ therefore yields
\begin{align}
\left|
[z^{2d}]\left(q(z)-\widetilde q(z)\right)
\right|
\le
2\epsilon
\left(\frac{eL}{2d}\right)^{2d}
\le
2\epsilon\left(\frac{eC_F}{2}\right)^{2d}.
\end{align}
Since the coefficient of $z^{2d}$ in $q$ is $c_{d}$, the coefficient of $z^{2d}$ in $\widetilde q$ yields the desired estimate. 
The only probabilistic failure comes from the event in Eq.~\eqref{eq:fourier-Markov}, whose probability is smaller than $1/4$.
This proves the claim.

\end{proof}


We conclude this section by presenting a key lemma used in the proof of Theorem~\ref{thm:probabilistic-robust-fourier-extraction}.

\begin{lemma}
\label{lem:roots-conditioning}
Let $d\ge 1$, let $L\ge 2(d+1)$, and let ${u}=e^{2\pi i/L}$. 
Let $I\subseteq\{0,1,\ldots,L-1\}$ be any subset with $|I|=d+1$.
If a polynomial $h(z)\in\mathbb C[z]$ has degree at most $d$, then
\begin{align}
\left|[z^d]h(z)\right|
\le
\max_{j\in I}|h({u}^j)|
\left(\frac{eL}{d}\right)^d,
\label{eq:roots-conditioning}
\end{align}
where $[z^d]\,h(z)$ denotes the coefficient of $z^d$ in the power-series expansion of $h(z)$.
\end{lemma}

\begin{proof}
Since $|I|=d+1$ and $\deg(h)\le d$, the values of $h$ on the points
$\{{u}^j:j\in I\}$ determine $h$ uniquely. 
Using the Lagrange representation, we have 
\begin{align}
h(z)
=
\sum_{j\in I}
h({u}^j)
\prod_{\substack{\ell\in I\\ \ell\ne j}}
\frac{z-{u}^\ell}{{u}^j-{u}^\ell}.
\label{eq:lagrange-roots-unity}
\end{align}
The coefficient of $z^d$ in the $j$-th Lagrange basis polynomial is 
\begin{align}
\left[z^d\right]
\prod_{\substack{\ell\in I\\ \ell\ne j}}
\frac{z-{u}^\ell}{{u}^j-{u}^\ell}
=
\frac{1}{
\prod_{\substack{\ell\in I\\ \ell\ne j}}
({u}^j-{u}^\ell)
}  ,
\label{eq:leading-coeff-lagrange-basis}
\end{align}
which implies 
\begin{align}
[z^d]h(z)
=
\sum_{j\in I}
\frac{h({u}^j)}
{
\prod_{\substack{\ell\in I\\ \ell\ne j}}
({u}^j-{u}^\ell)
}   ,
\label{eq:leading-coeff-lagrange}
\end{align}
and taking absolute values gives
\begin{align}
\left|[z^d]h(z)\right|
\le
\max_{j\in I}|h({u}^j)|
\sum_{j\in I}
\frac{1}{
\prod_{\substack{\ell\in I\\ \ell\ne j}}
|{u}^j-{u}^\ell|
}.
\label{eq:leading-coeff-bound-start}
\end{align}

It remains to lower bound the product of $|{u}^j-{u}^\ell|$. 
Fix $j\in I$. For $\ell\ne j$, define the circular grid distance between
$j$ and $\ell$ by
\begin{align}
\varrho(j,\ell):=\min\{|j-\ell|,L-|j-\ell|\}.
\end{align}
By definition, we have $1\le \varrho(j,\ell)\le \frac{L}{2}$.
Order the $d$ possible grid distances $\varrho(j,\ell)$, with $\ell\in I\setminus\{j\}$, as $\varrho_1\le \varrho_2\le \cdots\le \varrho_d$, such that $\varrho_1 \geq 1$ and $\varrho_d \leq L/2$.

Indeed, around the point $j$, there are at most two grid points at each
positive circular distance: one clockwise and one counterclockwise.
Therefore, the number of grid points at circular distance strictly smaller than $r/2$ is at most
\begin{align}
2\left\lfloor \frac{r-1}{2}\right\rfloor
\le r-1   ,
\end{align}
and thus it is impossible for the $r$-th closest point to have distance strictly
smaller than $r/2$. 
Hence, we have the bound $\varrho_r \ge {r}/{2}$ for every $r = 1,\ldots,d$.

Now, for a fixed $\ell \neq j$, 
\begin{align}
|{u}^j-{u}^\ell|
=
\left|1-e^{2\pi i\varrho(j,l)/L}\right|
=
2\sin\left(\frac{\pi \varrho(j,l)}{L}\right).
\end{align}
Moreover, since $0\le \varrho/L\le 1/2$, using the relation $\sin(\pi x)\ge 2x$ for $0\le x\le 1/2$, we have 
\begin{align}
|{u}^j-{u}^\ell|
\ge
\frac{4\varrho(j,l)}{L}  .
\label{eq:chord-lower-clear}
\end{align}
Therefore, the product of $|{u}^j-{u}^\ell|$ over all $\ell\in I\setminus\{j\}$ satisfies 
\begin{align}
\prod_{\substack{\ell\in I\\ \ell\ne j}}
|{u}^j-{u}^\ell|
\geq 
\prod_{\substack{\ell\in I\\ \ell\ne j}}\frac{4\varrho(j,l)}{L} 
\ge
\prod_{r=1}^{d}\frac{2r}{L}
=
\frac{2^d d!}{L^d}   ,
\label{eq:distance-product-lower-clear}
\end{align}
and equivalently, 
\begin{align}
\frac{1}
{\prod_{\substack{\ell\in I\\ \ell\ne j}}|{u}^j-{u}^\ell|}
\le
\frac{L^d}{2^d d!}.
\end{align}
Substituting this into Eq.~\eqref{eq:leading-coeff-bound-start} gives
\begin{align}
\left|[z^d]h(z)\right|
&\le
\max_{j\in I}|h({u}^j)|
(d+1)\frac{L^d}{2^d d!}.
\end{align}
Using $d!\ge (d/e)^d$ and $d+1\le 2^d$ for $d\ge 1$, we obtain
\begin{align}
(d+1)\frac{L^d}{2^d d!}
\le
\left(\frac{eL}{d}\right)^d  ,
\end{align}
which finally gives the desired bound
\begin{align}
\left|[z^d]h(z)\right|
\le
\max_{j\in I}|h({u}^j)|
\left(\frac{eL}{d}\right)^d   ,
\end{align}
proving the claim. 
\end{proof}

\section{Generalized Stockmeyer's reduction}\label{appendix: section: stockmeyer}

In this appendix, we generalize Stockmeyer's reduction described in~\cite{bouland2026complexity}, particularly to our TBS setting; the corresponding PBS result follows by the same argument.
While the proof closely parallels that of~\cite{bouland2026complexity}, we present it here for completeness and to make the analysis self-contained.

Throughout, we prove that the existence of an approximate classical sampler for TBS implies approximate estimation of threshold output probabilities for most $U$ and $x$, where $U \sim \Haar(m)$ and $x \sim \mathcal{D}_{m,n}$.
Note that, unlike the standard BosonSampling case treated in~\cite{bouland2026complexity}, no reverse-embedding lemma is required here, since the inputs to ATE in Definition~\ref{def: thresholdBS} are already $U$ and $x$.

We begin by formally defining a classical sampler for TBS, which follows the definition commonly used in the literature~\cite{aaronson2011computational, bouland2019complexity, movassagh2023hardness, bouland2022noise, bouland2026complexity, bouland2025exponential, go2026computational}.

\begin{definition}[Classical sampler for TBS]
\label{def:threshold-classical-sampler}
Let $T_x(U)$ denote the threshold output probability distribution.
We define a classical sampler for TBS as a probabilistic algorithm that, on input a unitary matrix $U \in \mathrm{U}(m)$ and an error parameter $\xi \ge 0$, outputs $x \in \mathcal{T}_{m,n}$ from a probability distribution $\widetilde{T}_x(U)$ satisfying
\begin{align}
\frac{1}{2}\sum_{x \in \mathcal{T}_{m,n}} |T_x(U) - \widetilde{T}_x(U)| \le \xi  ,
\end{align}
within $\poly(n,\xi^{-1})$ time.
\end{definition}

We show that, if a classical sampler as in Definition~\ref{def:threshold-classical-sampler} exists, then the ATE problem in Definition~\ref{def: thresholdBS} can be solved in $\rm{BPP^{NP}}$, provided that the total variation distance bound $\xi$ is sufficiently small compared with the error parameters of ATE.

\begin{theorem}[Stockmeyer reduction for TBS]
\label{thm:threshold-stockmeyer}
If there exists a classical sampler as in Definition~\ref{def:threshold-classical-sampler}, then there exists a $\mathrm{BPP}^{\mathrm{NP}}$ procedure that solves $\rm{ATE}$ problem to additive imprecision
\begin{align}
\varepsilon
=
\frac{1}{\delta}
\left(
4\xi + \frac{1}{ f(n)}
\right)  ,
\end{align}
for desired $f(n) = \poly(n)$, with failure probability at most $\delta$.
\end{theorem}

Before proving Theorem~\ref{thm:threshold-stockmeyer}, we describe how Theorem~\ref{thm:threshold-stockmeyer} implies the TBS statement in Theorem~\ref{thm: hardness of APE implies sampling hardness}.
Suppose that the ATE problem in Definition~\ref{def: thresholdBS} is shown to be \#P-hard for $\varepsilon, \delta = \poly(n)^{-1}$ under a finite-PH reduction (say, $\Sigma_{k}^{\mathrm{P}}$).
Then, if a classical sampler as in Definition~\ref{def:threshold-classical-sampler} exists, Theorem~\ref{thm:threshold-stockmeyer} implies that ATE is also solvable in $\mathrm{BPP}^{\mathrm{NP}}$.
Combined with Toda's theorem $\rm{PH} \subseteq {\rm P^{\# P}}$~\cite{toda1991pp}, this yields $\rm{PH} \subseteq \rm{BPP}^{\Sigma_{k}^{\mathrm{P}}}$, implying that PH collapses to a finite level.
In other words, unless PH collapses, no classical sampler satisfying Definition~\ref{def:threshold-classical-sampler} can exist.

\begin{proof}[Proof of Theorem~\ref{thm:threshold-stockmeyer}]
Fix $U \in \mathrm{U}(m)$.
Suppose there exists a classical sampler as in Definition~\ref{def:threshold-classical-sampler}, whose output probability distribution $\widetilde{T}_x(U)$ over $x\in \mathcal{T}_{m,n}$ satisfies
\begin{align}
\sum_{x \in \mathcal{T}_{m,n}} |T_x(U) - \widetilde{T}_x(U)| \le 2\xi  .
\end{align}
We then have
\begin{align}
\mathbb{E}_{x \sim \mathcal{D}_{m,n}}
\left[
\frac{|T_x(U) - \widetilde{T}_x(U)|}{T_x^{(\mathrm{avg})}}
\right]
=
\sum_{x \in \mathcal{T}_{m,n}} |T_x(U) - \widetilde{T}_x(U)|
\le 2\xi.
\end{align}
which, by Markov's inequality, for any $k_1 > 0$, implies
\begin{align}
\Pr_{x \sim \mathcal{D}_{m,n}}
\left[
|T_x(U) - \widetilde{T}_x(U)| \ge 2\xi k_1\, T_x^{(\mathrm{avg})}
\right]
\le \frac{1}{k_1}.
\label{eq:threshold-markov-1}
\end{align}

Now given a classical sampler with output probability distribution $\widetilde{T}_x(U)$, for any $x \in \mathcal{T}_{m,n}$, we apply Stockmeyer's approximate counting algorithm~\cite{stockmeyer1985approximation} with relative accuracy $1/(2f(n))$.
This gives an estimate $\overline{T}_x(U)$ satisfying
\begin{align}
|\widetilde{T}_x(U) - \overline T_x(U)|
\le
\frac{1}{2 f(n) }\, \widetilde{T}_x(U)  ,
\label{eq:threshold-stockmeyer}
\end{align}
for a desired polynomial $f(n) = \poly(n)$, in complexity class $\mathrm{BPP}^{\mathrm{NP}}$.
To convert this multiplicative estimate into an additive one, observe that
\begin{align}
\mathbb{E}_{x \sim \mathcal{D}_{m,n}}
\left[
\frac{\widetilde{T}_x(U) }{T_x^{(\mathrm{avg})}}
\right]
=
\sum_{x \in \mathcal{T}_{m,n}} \widetilde{T}_x(U)
=
1.
\end{align}
Therefore, again by Markov's inequality, for any $k_2 > 0$, we get
\begin{align}
\Pr_{x \sim \mathcal{D}_{m,n}}
\left[
\widetilde{T}_x(U) \ge k_2\, T_x^{(\mathrm{avg})}
\right]
\le \frac{1}{k_2}.
\label{eq:threshold-markov-2}
\end{align}
Combining Eq.~\eqref{eq:threshold-stockmeyer} and Eq.~\eqref{eq:threshold-markov-2}, we obtain
\begin{align}
\Pr_{x \sim \mathcal{D}_{m,n}}
\left[
| \widetilde{T}_x(U) - \overline T_x(U) |
\ge
\frac{k_2}{2 f(n) }\, T_x^{(\mathrm{avg})}
\right]
\le \frac{1}{k_2}.
\label{eq:threshold-stockmeyer-additive}
\end{align}

Finally, by the triangle inequality,
\begin{align}
|T_x(U) - \overline T_x(U)|
\le
|T_x(U) - \widetilde{T}_x(U)| + |\widetilde{T}_x(U) - \overline T_x(U)|.
\end{align}
Hence, by the union bound together with
Eq.~\eqref{eq:threshold-markov-1} and Eq.~\eqref{eq:threshold-stockmeyer-additive},
\begin{align}
\Pr_{x \sim \mathcal{D}_{m,n}}
\left[
|T_x(U) - \overline T_x(U)|
\ge
\left(
2\xi k_1 + \frac{k_2}{2 f(n)}
\right) T_x^{(\mathrm{avg})}
\right]
\le
\frac{1}{k_1} + \frac{1}{k_2}.
\end{align}
Setting $k_1 = k_2 = 2/\delta$ yields
\begin{align}
\Pr_{x \sim \mathcal{D}_{m,n}}
\left[
|T_x(U) - \overline T_x(U)|
\ge
\frac{1}{\delta}
\left(
4\xi + \frac{1}{f(n)}
\right)T_x^{(\mathrm{avg})}
\right]
\le \delta.
\end{align}
Since the above bound holds for every fixed $U$, it holds that
\begin{align}
\Pr_{\substack{U \sim \Haar(m) \\ x \sim \mathcal{D}_{m,n} }}
\left[
|T_x(U) - \overline T_x(U)|
\ge
\frac{1}{\delta}
\left(
4\xi + \frac{1}{f(n)}
\right)T_x^{(\mathrm{avg})}
\right]
\le \delta,
\end{align}
which proves the claim.
\end{proof}

We now define the analogous computational problem for PBS and state the corresponding result.

\begin{definition}[Classical sampler for PBS]
\label{def:parity-classical-sampler}
Let $P_x(U)$ denote the parity output probability distribution.
We define a classical sampler for PBS as a probabilistic algorithm that, on input a unitary matrix $U \in \mathrm{U}(m)$ and an error parameter $\xi \ge 0$, outputs $x \in \mathcal{P}_{m,n}$ from a probability distribution $\widetilde{P}_x(U)$ satisfying
\begin{align}
\frac{1}{2}\sum_{x \in \mathcal{P}_{m,n}} |P_x(U) - \widetilde{P}_x(U)| \le \xi  ,
\end{align}
within $\poly(n,\xi^{-1})$ time. 
\end{definition}

\begin{corollary}[Stockmeyer reduction for PBS]
\label{thm:parity-stockmeyer}
If there exists a classical sampler as in Definition~\ref{def:parity-classical-sampler}, then there exists a $\mathrm{BPP}^{\mathrm{NP}}$ procedure that solves $\rm{APE}$ problem to additive imprecision
\begin{align}
\varepsilon
=
\frac{1}{\delta}
\left(
4\xi + \frac{1}{ f(n)}
\right)  ,
\end{align}
for desired $f(n) = \poly(n)$, with failure probability at most $\delta$.
\end{corollary}

\begin{proof}
The proof follows from that of Theorem~\ref{thm:threshold-stockmeyer} by the substitutions
$T_x \mapsto P_x$, $\widetilde{T}_x \mapsto \widetilde{P}_x$, $T_x^{(\mathrm{avg})} \mapsto P_x^{(\mathrm{avg})}$, $\mathcal{D}_{m,n} \mapsto \mathcal{Y}_{m,n}$, $\mathcal{T}_{m,n} \mapsto \mathcal{P}_{m,n}$, and $\rm{ATE} \mapsto \rm{APE}$.
\end{proof}

Likewise, suppose that the APE problem in Definition~\ref{def: parityBS} is shown to be \#P-hard for $\varepsilon, \delta = \poly(n)^{-1}$ under a finite-PH reduction.
Then Corollary~\ref{thm:parity-stockmeyer}, together with Toda's theorem, implies that the existence of a classical sampler satisfying Definition~\ref{def:parity-classical-sampler} would collapse PH to a finite level.
This gives the PBS statement in Theorem~\ref{thm: hardness of APE implies sampling hardness}.

\section{Shift-and-scale stability of rectangular Haar-random submatrices}
\label{appendix: section: translation invariance}

In this appendix, we prove the shift-and-scale stability of rectangular Haar-random unitary submatrices. 
The proof is a rectangular analog of the corresponding argument for square Haar submatrices, as presented in~\cite{bouland2026complexity}.

Specifically, in Appendix~\ref{appendix: section: invariance: density}, we derive the density of rectangular Haar submatrices, denoted by $Z$ hereafter, with respect to the complex Lebesgue measure. 
In Appendix~\ref{appendix: section: invariance: perturbed density}, we analyze how this density transforms under the affine map $Z\mapsto (1-t)Z+tH$ for fixed $H$. 
In Appendix~\ref{appendix: section: invariance: boundary}, we establish the boundary-tail bound required for the shift-and-scale stability argument. 
In Appendix~\ref{appendix: section: invariance: general}, we prove the general shift-and-scale lemma.
Finally, in Appendix~\ref{appendix: section: invariance: main text}, we prove the specific shift-and-scale lemmas used in the main text.

\subsection{Density of the rectangular Haar submatrices}\label{appendix: section: invariance: density}

We first recall the density of a rectangular truncation of a Haar unitary (see~\cite{jiang2009approximation} for details).


\begin{lemma}[Density of rectangular Haar submatrices~\cite{jiang2009approximation}]
\label{lem:rect-density-appD}
Let $Z = U_{[p],[q]}\in\mathbb C^{p\times q}$ be the top-left $p\times q$ block of $U \sim \Haar(m)$, with $p+q <  m$.
Then with respect to Lebesgue measure on $\mathbb C^{p\times q}$, $Z$ has density
\begin{align}\label{eq:unscaled-density-appD}
f(Z)
=
C\,
\det(I_q-Z^\dagger Z)^{m-p-q}
\mathbf 1{\{Z^\dagger Z\prec I_q\}}  ,
\end{align}
for a normalizing constant $C >0$.
\end{lemma}

Due to the permutation stability of $U \sim \Haar(m)$, Lemma~\ref{lem:rect-density-appD} holds for $U_{R,[q]}$ with arbitrarily chosen rows $R$ (with $|R| = p$) provided that $m - p - q >0$, as required in our hardness analysis.

\subsection{Shift-and-scaled density}\label{appendix: section: invariance: perturbed density}

We now derive the shift-and-scaled density, introduced by the affine map $Z \mapsto (1-t)Z + tH$ for a fixed $H$.

\begin{lemma}[Shift-and-scaled density]
\label{lem:shift-scale-density}
Let $\mu$ be the distribution of $Z = U_{[p],[q]}\in\mathbb C^{p\times q}$ for $U \sim \Haar(m)$, and  let $\mu_t$ be the distribution of $(1-t)Z + tH$ for a fixed $H \in \mathbb{C}^{p \times q}$.
If $f$ is the density of $\mu$, then the density $f_t$ of $\mu_t$ is given by
\begin{align}
f_t(B)
=
(1-t)^{-2pq}
f\left(\frac{B-tH}{1-t}\right).
\label{eq:shift-scale-density}
\end{align}
\end{lemma}

\begin{proof}
The map $T_t(Z) := (1-t)Z+tH $ has inverse
\begin{align}
T_t^{-1}(B)=\frac{B-tH}{1-t}.
\end{align}
Regarding $\mathbb{C}^{p\times q}$ as the real vector space
$\mathbb{R}^{2pq}$, multiplication by $(1-t)^{-1}$ scales each of the $2pq$ real coordinates by $(1-t)^{-1}$, and therefore contributes a factor of $(1-t)^{-2pq}$ to the Jacobian determinant.
Also, the translation by $-tH$ has a unit Jacobian.
Hence, the pushforward density is given by 
\begin{align}
f_t(B)
=
(1-t)^{-2pq}
f\left(T_t^{-1}(B)\right)
=
(1-t)^{-2pq}
f\left(\frac{B-tH}{1-t}\right).
\end{align}
\end{proof}

\subsection{Interior event and boundary tail}\label{appendix: section: invariance: boundary}

Note that, the density of Haar-submatrices in Eq.~\eqref{eq:unscaled-density-appD} is supported on $Z^\dagger Z \prec I_q$, implying that all singular values of $Z$ are strictly smaller than $1$.
Since our proof of shift-and-scale stability requires a strict margin from this boundary, we introduce the following definition.

\begin{definition}[Interior event]\label{def: interior event}
For a fixed $\zeta\in(0,1)$, we define the interior event as
\begin{align}
\mathcal E_\zeta
:=
\left\{
Z\in\mathbb C^{p\times q}:
\|Z\|_{\mathrm{op}}\le 1-\zeta
\right\}.
\end{align}
\end{definition}

We now present a lemma establishing that, whenever $\zeta$ lies below a certain threshold, the complement event $\mathcal{E}_{\zeta}^c$ occurs with exponentially small probability.

\begin{lemma}[Boundary tail bound]\label{lem:boundary-tail-appD}
Let $Z=U_{[p],[q]} \in \mathbb{C}^{p\times q}$ be the top-left $p\times q$ block of $U \sim \Haar(m)$.
Assume that there exist fixed constants $c,C_{\mathrm{dim}}>0$ such that $p\leq C_{\mathrm{dim}}q$ and $m-p-q\geq cq$.
Then there exist constants $\zeta^*,c'>0$, depending only on $c$ and $C_{\mathrm{dim}}$, such that, for every fixed $\zeta\in(0,\zeta^*)$,
\begin{align}
\mathbb P[\|Z\|_{\mathrm{op}}>1-\zeta]
\le
e^{-c' (m-p-q)}.
\end{align}
\end{lemma}

As similarly argued, Lemma~\ref{lem:boundary-tail-appD} holds for $U_{R, [n]}$ with arbitrarily chosen mode set $R$, due to the permutation invariance of $U \sim \Haar(m)$.

\begin{proof}[Proof of Lemma~\ref{lem:boundary-tail-appD}]
We use the Gaussian representation of the first $q$ columns of a Haar unitary.
Let
\begin{align}
X=
\begin{pmatrix}
X_1\\
X_2
\end{pmatrix}
\in\mathbb C^{m\times q},
\end{align}
where $X_1 \sim \mathcal{N}(0,1)_{\mathbb{C}}^{p \times q}$ and $X_2 \sim \mathcal{N}(0,1)_{\mathbb{C}}^{(m-p) \times q}$ are i.i.d. complex Gaussian matrices.
Then the first $q$ columns of a Haar unitary have the same distribution as~\cite{mezzadri2006generate}
\begin{align}
Q=X(X^\dagger X)^{-1/2} ,
\end{align}
and accordingly, its top $p\times q$ submatrix is distributed as
\begin{align}
Z
=
X_1(X_1^\dagger X_1+X_2^\dagger X_2)^{-1/2}.
\end{align}

To proceed, for any nonzero vector $v\in\mathbb C^q$, write
\begin{align}
w:=
(X_1^\dagger X_1+X_2^\dagger X_2)^{-1/2}v  ,
\end{align}
such that $Zv=X_1w$.
Moreover, observe that
\begin{align}
\|v\|_2^2
=
w^\dagger(X_1^\dagger X_1+X_2^\dagger X_2)w 
=
\|X_1w\|_2^2+\|X_2w\|_2^2  ,
\end{align}
and thus we have 
\begin{align}
\frac{\|Zv\|_2^2}{\|v\|_2^2}
=
\frac{\|X_1w\|_2^2}
{\|X_1w\|_2^2+\|X_2w\|_2^2}.
\end{align}
By definition, taking the supremum over $v\ne0$, equivalently over $w\ne0$, gives
\begin{align}
\label{eq:rayleigh-quotient-bound}
\|Z\|_{\mathrm{op}}^2
=
\sup_{w\ne0}
\frac{\|X_1w\|_2^2}
{\|X_1w\|_2^2+\|X_2w\|_2^2}.
\end{align}

Next, we use standard Gaussian singular-value bounds in~\cite{vershynin2012introduction}. 
Concretely, with probability at least $1- e^{-C u^2}$ for a constant $C > 0$, we simultaneously have
\begin{align}
\|X_1\|_{\mathrm{op}}
\le
\sqrt p+\sqrt q+u,
\qquad
s_{\min}(X_2)
\ge
\sqrt{m-p}-\sqrt q-u  ,
\end{align}
where $s_{\min}(\cdot)$ denotes the smallest singular value of the matrix argument.
Conditioned on these events, for every $w\ne0$, we have
\begin{align}
\|X_1w\|_2^2
\le
(\sqrt p+\sqrt q+u)^2\|w\|_2^2, \qquad
\|X_2w\|_2^2
\ge
(\sqrt{m-p}-\sqrt q-u)^2\|w\|_2^2  ,
\end{align}
where we assumed $\sqrt{m-p}-\sqrt q > u$ here. 
Substituting these inequalities into Eq.~\eqref{eq:rayleigh-quotient-bound} implies
\begin{align}
\label{eq:A-bound-L1-L2}
\|Z\|_{\mathrm{op}}^2
\leq 
\frac{\|X_1w\|_2^2}{\|X_1w\|_2^2 + (\sqrt{m-p}-\sqrt q-u)^2\|w\|_2^2} 
\le
\frac{(\sqrt p+\sqrt q+u)^2}{(\sqrt p+\sqrt q+u)^2+ (\sqrt{m-p}-\sqrt q-u)^2} .
\end{align}
We now choose
\begin{align}
    u := \frac{\sqrt{m-p} - \sqrt{q}}{2}  .
\end{align}
Using $p\leq C_{\mathrm{dim}}q$ and $m-p-q\geq cq$, we have
\begin{align}
\frac{\sqrt{m-p}-\sqrt q}{2\sqrt p+\sqrt{m-p}+\sqrt q}
=
\frac{\sqrt{(m-p)/q}-1}{2\sqrt{p/q}+\sqrt{(m-p)/q}+1}
\geq
\frac{\sqrt{1+c}-1}{2\sqrt{C_{\mathrm{dim}}}+\sqrt{1+c}+1}
=:
\gamma_{c,C_{\mathrm{dim}}}>0.
\end{align}
Combining this with Eq.~\eqref{eq:A-bound-L1-L2}, define
\begin{align}\label{eq: def of zeta}
\zeta^*
:=
1-\frac{1}{\sqrt{1+\gamma_{c,C_{\mathrm{dim}}}^2}}
>0.
\end{align}
Then, on the same singular-value event,
\begin{align}
\|Z\|_{\mathrm{op}}
\leq
1-\zeta^*.
\end{align}
Consequently, the same bound holds with $1-\zeta$ on the right-hand side for every fixed $\zeta\in(0,\zeta^*)$.

Given $m - p - q \geq cq$ for a constant $c > 0$, let $m-p = rq$ for $r \geq 1 + c$. 
Then
\begin{align}
    \frac{\sqrt{m-p} - \sqrt{q}}{\sqrt{m - p - q} } = \frac{\sqrt{r} - 1}{\sqrt{r-1}} = \sqrt{\frac{\sqrt{r} - 1}{\sqrt{r} + 1}}   ,
\end{align}
which is increasing with $r$ and thus minimized when $r = 1 + c$. 
Therefore, we have 
\begin{align}
    \sqrt{m-p} - \sqrt{q} \geq \frac{\sqrt{1+c} - 1}{\sqrt{c}} \sqrt{m - p - q} 
    = \frac{\sqrt{c}}{\sqrt{1+c} + 1} \sqrt{m - p - q} , 
\end{align}
implying that 
\begin{align}
    \exp(-C u^2) \leq \exp(- \frac{C}{4} \left(\frac{\sqrt{c}}{\sqrt{1+c} + 1}\right)^2 (m - p - q) ) = \exp(-c' (m-p-q)) , 
\end{align}
for 
\begin{align}\label{eq: def of c_zeta}
    c' := \frac{C}{4}\left(\frac{\sqrt{c}}{\sqrt{1+c} + 1}\right)^2 > 0  .
\end{align}
Combining all the arguments, for every fixed $\zeta\in(0,\zeta^*)$ we obtain
\begin{align}
\mathbb P[\|Z\|_{\mathrm{op}}>1-\zeta]
\le
e^{-c' (m - p - q)},
\end{align}
where $\zeta^*$ and $c'$ are defined in Eq.~\eqref{eq: def of zeta} and Eq.~\eqref{eq: def of c_zeta}, respectively.

\end{proof}

\subsection{General shift-and-scale stability lemma}\label{appendix: section: invariance: general}

We now prove the core lemma for the shift-and-scale stability required in the main hardness argument. 

\begin{lemma}[General rectangular shift-and-scale stability]
\label{lem:general-translation-appD}
Let $\mu$ be the distribution of $Z = U_{[p],[q]}$ for $U \sim \Haar(m)$, and let $\mu_t$ be the distribution of $(1-t)Z + t H$ for a fixed $H\in\mathbb C^{p\times q}$. 
Assume that there exist fixed constants $c,C_{\mathrm{dim}}>0$ such that $p\leq C_{\mathrm{dim}}q$ and $m-p-q\geq cq$.
Let $\zeta \in (0, \zeta^*/2)$ be any constant for $\zeta^*$ satisfying the boundary tail bound in Lemma~\ref{lem:boundary-tail-appD}.
Then there exist constants $c', c_\zeta, C_\zeta>0$ such that, whenever
\begin{align}
0\le t\le
\min\left\{
\frac{\zeta}{4(1+\|H\|_{\mathrm{op}})},
\frac{c_\zeta}{
pq+(m-p-q)\sqrt q(\sqrt q+\|H\|_{\mathrm{F}})
}
\right\},
\label{eq:shift-scale-t-bound}
\end{align}
the total variation distance between $\mu$ and $\mu_t$ is bounded by
\begin{align}
\label{eq:general-translation-bound-unscaled}
\|\mu_t-\mu\|_{\mathrm{TVD}}
\le
C_\zeta \left[
pq+(m-p-q)\sqrt q(\sqrt q+\|H\|_{\mathrm{F}})
\right]t
+
e^{-c' (m-p-q)}   .
\end{align}
\end{lemma}

\begin{proof}

By Lemma~\ref{lem:shift-scale-density}, the density of $\mu_t$ is
\begin{align}
f_t(B)
=
(1-t)^{-2pq}
f\left(\frac{B-tH}{1-t}\right).
\end{align}
Fix $\zeta\in(0,\zeta^*/2)$ as in the statement of the lemma, and define the likelihood ratio only on the interior event $\mathcal E_\zeta$ by
\begin{align}
R_t(B)
:=
\frac{f_t(B)}{f(B)},
\qquad
B\in\mathcal E_\zeta.
\end{align}
Since the affine transformation can place $\mu_t$-mass outside the support of $\mu$, we do not identify the total variation distance with an expectation of $|R_t-1|$ over the entire space.
Instead, we split the integral into the interior region and its complement:
\begin{align}
\|\mu_t-\mu\|_{\mathrm{TVD}}
=
\frac12\int_{\mathbb C^{p\times q}} |f_t(B)-f(B)|\,dB
\leq
\frac12\int_{\mathcal E_\zeta}f(B)|R_t(B)-1|\,dB
+
\frac12\mu(\mathcal E_\zeta^c)
+
\frac12\mu_t(\mathcal E_\zeta^c).
\label{eq: exp of likelihood}
\end{align}
The likelihood-ratio calculation below is used only for the first term, while the last two terms are controlled by the boundary-tail bound.

Now define
\begin{align}
\Phi(W)
:=
\log\det\!\left(I_q-W^\dagger W\right).
\end{align}
From the density formula in Eq.~\eqref{eq:unscaled-density-appD}, on the region where the likelihood ratio is well-defined, we have
\begin{align}\label{eq: log likelihood}
\log R_t(B)
=
-2pq\log(1-t)
+
(m-p-q)
\left[
\Phi\left(\frac{B-tH}{1-t}\right)-\Phi(B)
\right].
\end{align}
Here, for $t\le 1/2$,
\begin{align}
|-2pq\log(1-t)|\le 4pq\,t.
\label{eq:jacobian-term-bound}
\end{align}
Let us define the path
\begin{align}
W_s:=\frac{B-sH}{1-s} ,
\qquad
0\le s\le t   ,
\end{align}
so that $W_0=B$ and $W_t=(B-tH)/(1-t)$. 
Then we have 
\begin{align}\label{eq: dPhi integral}
\Phi\left(\frac{B-tH}{1-t}\right)-\Phi(B)
=
\int_0^t \frac{d}{ds}\Phi(W_s)\,ds.
\end{align}
Accordingly, the problem reduces to bounding the derivative of $\Phi(W_s)$.
To simplify the analysis, let 
\begin{align}
    M_s:=I_q-W_s^\dagger W_s   ,
\end{align}
such that $\Phi(W_s)=\log\det M_s$.

Now assume $B \in\mathcal E_\zeta$ for $\mathcal{E}_\zeta$ in Definition~\ref{def: interior event}, such that $\|B\|_{\mathrm{op}}\le 1-\zeta$. 
If $t$ is bounded by 
\begin{align}
t\le \frac{\zeta}{4(1+\|H\|_{\mathrm{op}})},
\label{eq:path-stays-interior}
\end{align}
then for every $0\le s\le t$,
\begin{align}
\|W_s\|_{\mathrm{op}}
\le
\frac{\|B\|_{\mathrm{op}}+s\|H\|_{\mathrm{op}}}{1-s}
\le
1-\frac{\zeta}{2}.
\label{eq:Ws-interior}
\end{align}
Therefore, $\| W_s^\dagger W_s\|_{\mathrm{op}} = \|W_s\|_{\mathrm{op}}^2 \leq (1 - \zeta/2)^2 $, and thus we have 
\begin{align}
M_s = I_q-W_s^\dagger W_s
\succeq
c_\zeta I_q
\end{align}
for a constant
\begin{align}\label{eq: def of c_zeta 2}
    c_\zeta = \zeta -  \zeta^2/4 > 0   ,
\end{align}
which also implies
\begin{align}\label{eq: def of C_zeta}
\left\|M_s^{-1}\right\|_{\mathrm{op}}
\le 
\frac{1}{\zeta -  \zeta^2/4} 
=: 
C_\zeta = c_\zeta^{-1}.
\end{align}

Provided that $M_s$ is invertible, the standard matrix-calculus identity gives~\cite{magnus2019matrix, higham2008functions}
\begin{align}
\frac{d}{ds}\Phi(W_s)
=
\mathrm{Tr}\!\left(M_s^{-1}\frac{dM_s}{ds}\right)  .
\end{align}
Using
\begin{align}
\frac{dM_s}{ds}
= -\frac{d}{ds}(W_s^\dagger W_s)
= -\left(\frac{d}{ds}W_s\right)^\dagger W_s - W_s^\dagger \left( \frac{d}{ds}W_s \right)   ,
\end{align}
we obtain
\begin{align}
\left|\frac{d}{ds}\Phi(W_s)\right|
&\leq 
\left| \mathrm{Tr}\!\left[
M_s^{-1}
\left( \left(\frac{d}{ds}W_s\right)^\dagger W_s + W_s^\dagger \left( \frac{d}{ds}W_s \right)  \right)
\right]\right|  \\
&\leq
\left| \mathrm{Tr}\!\left[
M_s^{-1}\left(\frac{d}{ds}W_s\right)^\dagger W_s  \right]  \right| 
+ \left| \mathrm{Tr}\!\left[ M_s^{-1}W_s^\dagger \left( \frac{d}{ds}W_s \right)  \right] \right| \\
&\leq
\left\| \frac{d}{ds}W_s  \right\|_{\mathrm F} \left(  \| W_s M_s^{-1} \|_{\mathrm F}  + \| M_s^{-1}W_s^\dagger \|_{\mathrm F} \right) \\
&\leq 
2\left\| \frac{d}{ds}W_s  \right\|_{\mathrm F} \left\| W_s  \right\|_{\mathrm F} \| M_s^{-1} \|_{\mathrm{op}}  \\
&\leq
2 C_\zeta  \left\| \frac{d}{ds}W_s  \right\|_{\mathrm F} \left\| W_s  \right\|_{\mathrm F} ,
\end{align}
where we used Cauchy--Schwarz inequality, sub-additivity of the Frobenius norm, and also used the identity:
\begin{align}
    \| A B \|_{\mathrm F}^2 = \Tr(B^{\dag}A^{\dag}AB)  \leq 
    \| B \|_{\mathrm{op}}^2 \Tr(A^{\dag}A) 
    =
    \| B \|_{\mathrm{op}}^2 \| A \|_{\mathrm F}^2  ,
\end{align}
for any complex matrices $A,B$.
On the event $B\in \mathcal E_\zeta$ and under Eq.~\eqref{eq:path-stays-interior},
\begin{align}
\|W_s\|_{\mathrm{F}}\le \sqrt q\|W_s\|_{\mathrm{op}}\le \sqrt q,
\end{align}
by using Eq.~\eqref{eq:Ws-interior}.
Also, for $s\le t\le 1/2$, we have 
\begin{align}
\left\|\frac{dW_s}{ds}\right\|_{\mathrm F}
=
\left\| \frac{B-H}{(1-s)^2} \right\|_{\mathrm F}
\le
4(\|B\|_{\mathrm F} + \|H\|_{\mathrm F})
\le
4(\sqrt{q}\|B\|_{\mathrm{op}} + \|H\|_{\mathrm F})
\le
4(\sqrt q+\|H\|_{\mathrm F})  .
\end{align}
Hence,
\begin{align}
\left|
\frac{d}{ds}\Phi(W_s)
\right|
\le
8C_\zeta\sqrt q(\sqrt q+\|H\|_{\mathrm{F}})  ,
\label{eq:Phi-derivative-bound-final}
\end{align}
and as given in Eq.~\eqref{eq: dPhi integral}, integrating over $s\in[0,t]$ gives
\begin{align}
\left|
\Phi\left(\frac{B-tH}{1-t}\right)-\Phi(B)
\right|
\le
8C_\zeta\sqrt q(\sqrt q+\|H\|_{\mathrm{F}})t  .
\label{eq:Phi-difference-shift-scale}
\end{align}
Combining Eq.~\eqref{eq: log likelihood},
Eq.~\eqref{eq:jacobian-term-bound}, and
Eq.~\eqref{eq:Phi-difference-shift-scale}, we obtain, for $B\in \mathcal E_\zeta$,
\begin{align}
|\log R_t(B)|
\le
8 C_\zeta
\left[
pq+(m-p-q)\sqrt q(\sqrt q+\|H\|_{\mathrm{F}})
\right]t   ,
\label{eq:log-ratio-shift-scale-final}
\end{align}
where we used $C_\zeta \geq 4/3$.
If $t$ also satisfies 
\begin{align}\label{eq: bound of t 2}
 t\le
\frac{c_\zeta / 8}{ 
pq+(m-p-q)\sqrt q(\sqrt q+\|H\|_{\mathrm{F}})
},
\end{align}
then the right-hand side of Eq.~\eqref{eq:log-ratio-shift-scale-final} is smaller than $1$.
Hence
\begin{align}
|R_t(B)-1|
=
|e^{\log R_t(B)}-1|
\le
2|\log R_t(B)|
\le
16C_\zeta
\left[
pq+(m-p-q)\sqrt q(\sqrt q+\|H\|_{\mathrm{F}})
\right]t   ,
\label{eq:R-minus-one-shift-scale}
\end{align}
because $|e^x-1|\le 2|x|$ for $|x| \le 1$.

By Eq.~\eqref{eq:R-minus-one-shift-scale}, the interior contribution in Eq.~\eqref{eq: exp of likelihood} satisfies
\begin{align}\label{eq: splitting exp of log likelihood}
\frac12\int_{\mathcal E_\zeta}f(B)|R_t(B)-1|\,dB
\leq
8C_\zeta
\left[
pq+(m-p-q)\sqrt q(\sqrt q+\|H\|_{\mathrm F})
\right]t.
\end{align}

We next bound the two exterior terms in Eq.~\eqref{eq: exp of likelihood}.
Lemma~\ref{lem:boundary-tail-appD} gives
\begin{align}
\mu(\mathcal E_\zeta^c)
\leq
e^{-c_0(m-p-q)}
\end{align}
for a constant $c_0>0$.
Since $\zeta<\zeta^*/2$, set $\zeta_+:=2\zeta<\zeta^*$.
The first upper bound on $t$ in Eq.~\eqref{eq:shift-scale-t-bound} gives $t(1+\|H\|_{\mathrm{op}})\leq\zeta/4$.
Therefore, for every $B\in\mathcal E_{\zeta_+}$,
\begin{align}
\|(1-t)B+tH\|_{\mathrm{op}}
\leq
(1-t)(1-2\zeta)+t\|H\|_{\mathrm{op}}
\leq
1-2\zeta+t(1+\|H\|_{\mathrm{op}})
\leq
1-\zeta.
\end{align}
Thus $(1-t)B+tH\in\mathcal E_\zeta$, and taking the contrapositive yields
\begin{align}
\mu_t(\mathcal E_\zeta^c)
=
\mu\!\left(\{B:(1-t)B+tH\in\mathcal E_\zeta^c\}\right)
\leq
\mu(\mathcal E_{2\zeta}^c)
\leq
e^{-c_+(m-p-q)}
\end{align}
for a constant $c_+>0$.
After updating constants, we obtain
\begin{align}\label{eq: def of c'}
\frac12\mu(\mathcal E_\zeta^c)
+
\frac12\mu_t(\mathcal E_\zeta^c)
\leq
e^{-c'(m-p-q)}.
\end{align}

Substituting Eqs.~\eqref{eq: splitting exp of log likelihood} and~\eqref{eq: def of c'} into Eq.~\eqref{eq: exp of likelihood} yields
\begin{align}
\|\mu_t-\mu\|_{\mathrm{TVD}}
\le
C_\zeta
\left[
pq+(m-p-q)\sqrt q(\sqrt q+\|H\|_{\mathrm{F}})
\right]t 
+
e^{-c' (m-p-q)},
\end{align}
whenever $t$ satisfies 
\begin{align}
0\le t\le
\min\left\{
\frac{\zeta}{4(1+\|H\|_{\mathrm{op}})},
\frac{c_\zeta}{
pq+(m-p-q)\sqrt q(\sqrt q+\|H\|_{\mathrm{F}})
}
\right\},
\end{align}
where, for notational simplicity, we have redefined the constants $C_\zeta \rightarrow C_\zeta/16$ and $c_\zeta \rightarrow 8c_\zeta$.
This completes the proof.

\end{proof}

\subsection{Applications of Lemma~\ref{lem:general-translation-appD}: Proof of shift-and-scale stability lemmas in the main text}\label{appendix: section: invariance: main text}

By Lemma~\ref{lem:general-translation-appD}, the shift-and-scale stability lemmas in the main text, Lemma~\ref{lem:tbs-TV} and Lemma~\ref{lem:sparse-TV}, follow straightforwardly.

We first prove Lemma~\ref{lem:tbs-TV}. 
Apply Lemma~\ref{lem:general-translation-appD} with $p=q=n$ and $H=A_\star$, where $A_\star$ is constructed in Lemma~\ref{lem:tbs-block-permanent-encoding} and satisfies 
\begin{align}
    \|A_\star\|_{\mathrm{F}} \leq C\kappa n 
\end{align}
for an absolute constant $C>0$. 
Substituting $p=q=n$ and this Frobenius-norm bound into Lemma~\ref{lem:general-translation-appD} gives
\begin{align}
    n^2+(m-2n)\sqrt n(\sqrt n+\|A_\star\|_{\mathrm{F}})
    \le C'mn(1+\kappa\sqrt n)  ,
\end{align}
for another constant $C'>0$. 
This yields Eq.~\eqref{eq:tbs-t-bound} and Eq.~\eqref{eq:tbs-TV-bound} in Lemma~\ref{lem:tbs-TV} after redefining constants.

We next prove Lemma~\ref{lem:sparse-TV}. 
Apply Lemma~\ref{lem:general-translation-appD} with $p=2k$, $q=n$, and $H=A_\star$, where $A_\star$ is the matrix constructed in Lemma~\ref{lem:large-block-permanent-encoding} and satisfies
\begin{align}
    \|A_\star\|_{\mathrm F}\le 2\kappa\sqrt{nk}.
\end{align}
Substituting this bound into Lemma~\ref{lem:general-translation-appD} yields
\begin{align}
    2kn+(m-n-2k)\sqrt n(\sqrt n+\|A_\star\|_{\mathrm F})
    \le
    C mn(1+\kappa\sqrt{k})   ,
\end{align}
for some absolute constant $C>0$.
The claim in Lemma~\ref{lem:sparse-TV} then follows by absorbing constants.

\section{Random conditional Haar completion}
\label{appendix: section: conditional-haar-completion}

This appendix shows how to lift a distance between distributions on unitary submatrices, as characterized in Appendix~\ref{appendix: section: translation invariance}, to the corresponding distance between distributions on full unitary matrices.

Keeping the notation used in Appendix~\ref{appendix: section: translation invariance}, let $Z := U_{[p],[q]} \in \mathbb{C}^{p \times q}$ be a fixed $p \times q$ submatrix of $U \in \mathrm{U}(m)$, and assume $m - p - q > 0$, which aligns well with the regime relevant to our main analysis.
Let $\pi: \mathrm{U}(m)\to \mathbb{C}^{p\times q}$ denote the projection $\pi(U):= Z$, such that $\mu := \pi_{\#}\Haar(m)$ denotes the distribution of the $p \times q$ block of a Haar-random unitary. 
Also, let $\mu_t$ denote the distribution of the perturbed Haar submatrices introduced in Appendix~\ref{appendix: section: translation invariance}, namely, the distribution of $(1-t)Z + tH$ when $Z$ is distributed according to $\mu$.

We employ the notion of interior event defined in Definition~\ref{def: interior event}, that is, 
\begin{align}
\mathcal{E}_\zeta := \{Z\in\mathbb{C}^{p\times q}:\|Z\|_{\mathrm{op}}\le 1-\zeta\},
\end{align}
for a constant $\zeta \in (0,1)$.
We also use the notion 
\begin{align}\label{eq: interior - strict}
\mathcal{E} := \{Z\in\mathbb{C}^{p\times q}:\|Z\|_{\mathrm{op}} < 1\} .
\end{align}

In Appendix~\ref{appendix: section: completion: basic}, we construct the basic unitary-completion procedure.
In Appendix~\ref{appendix: section: completion: exact}, we show that this procedure exactly follows the conditional Haar distribution. 
In Appendix~\ref{appendix: section: completion: tvd}, we prove that unitary completion does not increase the total variation distance between the original block distributions. 
Finally, in Appendix~\ref{appendix: section: completion: full}, we construct the full unitary-completion procedure.

\subsection{The completion procedure}\label{appendix: subsec: completion procedure}\label{appendix: section: completion: basic}

We now describe a randomized polynomial-time procedure $\mathcal{C}(Z)$ that, on input $Z\in \mathcal{E}$ in Eq.~\eqref{eq: interior - strict}, outputs a unitary matrix $U$ whose top-left $p \times q$ submatrix is exactly $Z$ (i.e., $\pi(U) = Z$).

The first $q$ columns of the completed unitary $U$ are constructed as follows.
Set
\begin{align}
M_{Z}:=(I_q-Z^\dagger Z)^{1/2}.
\end{align}
Since $\|Z\|_{\mathrm{op}}\le 1-\zeta$, the matrix $I_q-Z^\dagger Z$ is positive
definite, thus $M_{Z}$ is well-defined.

Next, we sample $W \in \mathbb{C}^{(m-p)\times q}$ uniformly from the complex Stiefel manifold $V_{q}(\mathbb{C}^{m-p})$. 
This can be achieved, for example, by drawing a Haar-random unitary from $\mathrm{U}(m-p)$ and taking its first $q$ columns~\cite{muirhead2009aspects, mezzadri2006generate, tropp2012comparison, chikuse2003statistics}. 
Define
\begin{align}\label{eq: how C(A) construct Q}
Y:= WM_{Z},
\qquad
Q :=
\begin{pmatrix}
Z\\
Y
\end{pmatrix}
\in \mathbb{C}^{m\times q}.
\end{align}
By construction, 
\begin{align}
Q^\dagger Q
=
Z^\dagger Z + M_{Z}^\dagger W^\dagger W M_{Z}
=
Z^\dagger Z + M_{Z}^2
=
I_q,
\end{align}
so $Q \in V_{q}(\mathbb{C}^{m})$ is composed of $q$ orthonormal columns.

For the remaining $m - q$ columns, we fix an arbitrary but deterministic orthonormal basis $O_Q \in\mathbb{C}^{m\times(m-q)}$ over $\mathrm{col}(Q)^{\bot}$, the subspace orthogonal to all columns of $Q$, for example by applying Gram--Schmidt orthonormalization.
Then we sample $H \sim \mathrm{Haar}(m-q)$ independently, and finally output
\begin{align}
U = [\,Q\ \ O_QH\,].
\label{eq:conditional-completion-output}
\end{align}
By construction, $U^{\dag}U = I_{m}$, which implies $U \in \mathrm{U}(m)$ and $\pi(U)= Z$ as required.

\subsection{The completion follows the exact Haar distribution}\label{appendix: section: completion: exact}

The following lemma claims that the above procedure in Sec.~\ref{appendix: subsec: completion procedure} is exactly the right way to sample a Haar-random unitary conditioned on its $[p]\times[q]$ block.

\begin{lemma}[Conditional Haar completion]
\label{lem:conditional-haar-completion}
For every $Z \in \mathcal{E}$, the output of $\mathcal{C}(Z)$ has the same
distribution as a Haar-random unitary conditioned on the event $\pi(U) = Z$.
Equivalently, if $Z$ is sampled from the Haar block distribution $\mu$, then the sampled $U$ from $\mathcal{C}(Z)$ is exactly Haar-random in $\mathrm{U}(m)$.
\end{lemma}

\begin{proof}
We prove the latter statement, namely, that if $U$ is sampled according to $\mathcal C(Z)$ with $Z$ drawn from $\mu$, then $U$ is Haar-random.
The former statement then follows immediately, since $Z$ drawn from $\mu$ lies in $\mathcal{E}$ with probability $1$.

To proceed, consider Haar-random $U_0\sim\Haar(m)$, and let $Q_0$ be its first $q$ columns. 
Writing
\begin{align}
Q_0=
\begin{pmatrix}
Z_0\\
Y_0
\end{pmatrix} ,
\end{align}
by definition, $Z_0\in\mathbb{C}^{p\times q}$ has distribution $\mu$. 
Also, since $Q_0^\dagger Q_0=I_q$, the bottom block $Y_0 \in\mathbb{C}^{(m-p)\times q}$ satisfies $Y_0^\dagger Y_0=I_q-Z_0^\dagger Z_0$. 
Hence, after fixing $Z_0 = Z$, every possible lower block has the form
\begin{align}
Y = W(I_q-Z^\dagger Z)^{1/2},
\end{align}
for $W\in V_q(\mathbb{C}^{m-p})$, as also argued in the proof of Lemma~\ref{lem:rect-density-appD}.


Here, left multiplication of $U_0$ by a block unitary of the form
\begin{align}
\begin{pmatrix}
I_p & 0\\
0 & G
\end{pmatrix},
\qquad
G\in \mathrm{U}(m-p),
\end{align}
leaves $Z$ unchanged while mapping $Y$ to $GY$. 
Since Haar measure is left-invariant, the conditional distribution of $Y$ is unchanged under all such rotations $Y \mapsto GY$.
Accordingly, for $U_0 \sim \Haar(m)$, it follows that the conditional distribution of $W$ must be uniform on $V_q(\mathbb{C}^{m-p})$~\cite{chikuse2003statistics, mardia1977uniform}, which is precisely the distribution from which $\mathcal C(Z)$ chooses $W$ in Eq.~\eqref{eq: how C(A) construct Q}.
Therefore, when $Z$ is drawn from $\mu$, the first $q$ columns
\begin{align}
Q =
\begin{pmatrix}
Z\\
W(I_q- Z^\dagger Z)^{1/2}
\end{pmatrix}
\end{align}
generated by $\mathcal{C}(Z)$ have the same distribution as the first $q$ columns of a Haar-random unitary. 

Next, for a fixed $Q$, choose one orthonormal basis $O_{Q} \in \mathbb{C}^{m \times (m - q)}$ over $\mathrm{col}(Q)^{\perp}$.
Then, $O_Q^\dagger O_Q=I_{m-q}$ and $Q^\dagger O_Q =0$, implying that $[Q\ O_Q]\in \mathrm{U}(m)$. 
Now consider another way to complete $Q$ to a unitary $[Q\  \widetilde O_Q]$ for $\widetilde O_Q \in\mathbb{C}^{m\times(m-q)}$ being another orthonormal basis
over $\mathrm{col}(Q)^{\perp}$. 
Since $O_Q$ and $\widetilde O_Q$ are two orthonormal bases of the same subspace, there is a unique unitary matrix $H \in \mathrm{U}(m-q)$ such that $\widetilde O_Q = O_QH$ (or equivalently $H := O_Q^\dagger \widetilde O_Q$). 
Hence, every unitary whose first $q$ columns are $Q$ can be written uniquely as
\begin{align}
[\,Q\ \ O_QH\,],
\qquad
H\in \mathrm{U}(m-q).
\label{eq:completion-form}
\end{align}

Now consider a Haar-random unitary conditioned on its first $q$ columns equal to $Q$. 
Then, Eq.~\eqref{eq:completion-form} induces a probability distribution on $H\in U(m-q)$, which must, in fact, be right-invariant.
To make this explicit, consider right-multiplying the full unitary by the block unitary
\begin{align}
\begin{pmatrix}
I_q & 0\\
0 & G
\end{pmatrix},
\qquad
G\in \mathrm{U}(m-q),
\end{align}
which leaves the first $q$ columns $Q$ unchanged while mapping $O_QH \mapsto O_QHG$. 
Hence, the right invariance of the full unitary induces right invariance of $H$, implying that $H$ is Haar distributed.
Therefore, conditioned on the first $q$ columns $Q$, the remaining columns of a Haar-random unitary can be generated by fixing an arbitrary orthonormal basis $O_{Q} \in \mathbb{C}^{m \times (m - q)}$, sampling $H \sim \Haar(m-q)$, and outputting $[Q\ \ O_QH]$.
This coincides exactly with the completion rule used by $\mathcal C(Z)$ in  Eq.~\eqref{eq:conditional-completion-output}.

To summarize, applying $\mathcal C(Z)$ to $Z$ drawn from $\mu$ reproduces the distribution of a Haar-random unitary. 
Equivalently, for any fixed $Z \in \mathcal{E}$, the output of $\mathcal C(Z)$ has the same distribution as a Haar-random unitary conditioned on $\pi(U)=Z$.

\end{proof}

\subsection{Completion does not increase total variation distance}\label{appendix: section: completion: tvd}

The following lemma is the general statement required in the proof of Lemma~\ref{lem:tbs-completion} and Lemma~\ref{lem:completion}: 
if two block distributions are close in total variation distance, then the induced distributions on their unitary completions are no farther apart in total variation distance.

\begin{lemma}
\label{lem:completion-postprocessing}
Let $\varsigma$ and $\varsigma'$ be two distributions supported on $\mathcal{E}_{\zeta}$.
Let $\upsilon$ and $\upsilon'$ denote the induced distributions on $\mathrm{U}(m)$ obtained by pushing forward $\varsigma$ and $\varsigma'$, respectively, through the same randomized completion procedure $\mathcal{C}$.
Then
\begin{align}
\|\upsilon-\upsilon'\|_{\mathrm{TVD}}
\le
\|\varsigma-\varsigma'\|_{\mathrm{TVD}}.
\end{align}
\end{lemma}

\begin{proof}
Let us write the random seed of the completion procedure $\mathcal{C}$ as $\omega$, such that $\mathcal C(Z;\omega)$ is deterministic once $\omega$ is fixed.

Now fix any event $B \subseteq \mathrm{U}(m)$. 
For each fixed value of $\omega$, define 
\begin{align}
E_{\omega,B}
:=
\{Z:\mathcal C(Z;\omega)\in B\}.
\end{align}
Then, for a fixed $\omega$, we have 
\begin{align}
\Pr_{Z\sim\varsigma}[\mathcal C(Z;\omega)\in B]
=
\varsigma(E_{\omega,B}),
\qquad 
\Pr_{Z\sim\varsigma'}[\mathcal C(Z;\omega)\in B]
=
\varsigma'(E_{\omega,B}).
\end{align}
By the definition of total variation distance,
\begin{align}
|\varsigma(E_{\omega,B})-\varsigma'(E_{\omega,B})|
\le
\|\varsigma-\varsigma'\|_{\mathrm{TVD}}.
\end{align}

Next, note that 
\begin{align}
    \upsilon (B) = \Pr_{Z,\omega}[\mathcal{C}(Z;\omega) \in B] = \E_{\omega}\left[ \Pr_{Z \sim \varsigma }[ \mathcal{C}(Z;\omega) \in B ] \right] ,
\end{align}
and the same identity holds for $\upsilon'(B)$ upon replacing $\varsigma$ with $\varsigma'$.
Accordingly, we have 
\begin{align}
|\upsilon(B)-\upsilon'(B)|
&=
\left|
\mathbb E_{\omega}
\left[
\varsigma(E_{\omega,B})-\varsigma'(E_{\omega,B})
\right]
\right|\\
&\le
\mathbb E_{\omega}
\left[
|\varsigma(E_{\omega,B})-\varsigma'(E_{\omega,B})|
\right]\\
&\le
\|\varsigma-\varsigma'\|_{\mathrm{TVD}}.
\end{align}
Since this holds for every event $B\subseteq \mathrm{U}(m)$, taking the supremum over
$B$ gives the desired bound 
\begin{align}
\|\upsilon-\upsilon'\|_{\mathrm{TVD}}
\le
\|\varsigma-\varsigma'\|_{\mathrm{TVD}} ,
\end{align}
concluding the proof. 
\end{proof}

\subsection{Applying the completion lemma to the perturbed block distribution}\label{appendix: section: completion: full}

We now define an extended completion procedure $\widetilde{\mathcal{C}}_{\zeta}$ on all of $Z \in \mathbb{C}^{p\times q}$ as follows:
\begin{itemize}
\item[(a)] if $Z\in \mathcal{E}_{\zeta}$, output $\mathcal{C}(Z)$;
\item[(b)] if $Z\notin \mathcal{E}_{\zeta}$, output a fixed arbitrary unitary $U_0 \in \mathrm{U}(m)$.
\end{itemize}

Then, applying the corresponding completion procedure to $\mu_t$ yields a distribution that is close to the Haar measure $\Haar(m)$, as shown in the following lemma.

\begin{lemma}
\label{lem:lifting-to-unitaries}
Let $\mu := \pi_{\#}\Haar(m)$ denote the Haar block distribution, and let $\mu_t$ denote the perturbed Haar block distribution from $\mu$.
Define
\begin{align}
\mathcal{H}_t:=(\widetilde{\mathcal{C}}_{\zeta})_{\#}\mu_t
\end{align}
to be the distribution on $\mathrm{U}(m)$ obtained by pushing forward $\mu_t$ through the randomized completion procedure $\widetilde{\mathcal{C}}_{\zeta}$ depicted above.
Then
\begin{align}
\|\mathcal{H}_t-\Haar(m)\|_{\mathrm{TVD}}
\le
\|\mu_t-\mu\|_{\mathrm{TVD}}
+
\mu(\mathcal{E}_{\zeta}^c).
\end{align}
Moreover, if $A \sim \mu_t$ and $U = \widetilde C_\zeta(A)$, then
\begin{align}\label{eq: app completion fail}
\Pr[\pi(U)\neq A]\le \mu_t (\mathcal E_\zeta^c).
\end{align}

\end{lemma}

\begin{proof}
Apply $\widetilde{\mathcal{C}}_{\zeta}$ to both $\mu_t$ and $\mu$.
By Lemma~\ref{lem:completion-postprocessing}, we obtain
\begin{align}
\|(\widetilde{\mathcal{C}}_{\zeta})_{\#}\mu_t-(\widetilde{\mathcal{C}}_{\zeta})_{\#}\mu\|_{\mathrm{TVD}}
\le
\|\mu_t-\mu\|_{\mathrm{TVD}}.
\end{align}
Now compare $(\widetilde{\mathcal{C}}_{\zeta})_{\#}\mu$ with the exact Haar measure $\Haar(m)$.
On the event $\mathcal{E}_{\zeta}$, $\widetilde{\mathcal{C}}_{\zeta}$ agrees with the true conditional Haar completion by Lemma~\ref{lem:conditional-haar-completion}, so it produces exactly a Haar-random unitary.
The only discrepancy comes from the complement event $\mathcal{E}_{\zeta}^c$, which yields the bound
\begin{align}
\|(\widetilde{\mathcal{C}}_{\zeta})_{\#}\mu-\Haar(m)\|_{\mathrm{TVD}}
\le
\Pr_{Z \sim \mu}[\widetilde{\mathcal{C}}_{\zeta}(Z) \neq \mathcal{C}(Z)]
\leq
\Pr_{Z \sim \mu}[Z \in  \mathcal{E}_{\zeta}^c]
\leq
\mu(\mathcal{E}_{\zeta}^c).
\end{align}
The triangle inequality gives
\begin{align}
\|\mathcal{H}_t-\Haar(m)\|_{\mathrm{TVD}}
&=
\|(\widetilde{\mathcal{C}}_{\zeta})_{\#}\mu_t-\Haar(m)\|_{\mathrm{TVD}}\\
&\le
\|(\widetilde{\mathcal{C}}_{\zeta})_{\#}\mu_t-(\widetilde{\mathcal{C}}_{\zeta})_{\#}\mu\|_{\mathrm{TVD}}
+
\|(\widetilde{\mathcal{C}}_{\zeta})_{\#}\mu-\Haar(m)\|_{\mathrm{TVD}}\\
&\le
\|\mu_t-\mu\|_{\mathrm{TVD}}
+
\mu(\mathcal{E}_{\zeta}^c).
\end{align}

It remains to prove the bound in Eq.~\eqref{eq: app completion fail}.
If $A \sim \mu_t$ satisfies $A\in \mathcal E_\zeta$, then the completion procedure is exact and satisfies $\pi(\widetilde C_\zeta(A))=A$.
Therefore, the failure event $\{\pi(U)\neq A\}$ is contained in $\{A\in \mathcal E_\zeta^c\}$, leading to
\begin{align}
\Pr[\pi(U)\neq A]\le \Pr[A\in \mathcal E_\zeta^c]= \mu_t( \mathcal E_\zeta^c).
\end{align}
This proves Eq.~\eqref{eq: app completion fail}.

\end{proof}

For a suitable choice of the constant $\zeta$, both $\mu(\mathcal{E}_\zeta^c)$ and $\mu_t(\mathcal{E}_\zeta^c)$ are bounded by $\exp(-\Omega(m-p-q))$, as shown in Lemma~\ref{lem:boundary-tail-appD} and Lemma~\ref{lem:general-translation-appD}, where the latter holds for $t$ satisfying Eq.~\eqref{eq:shift-scale-t-bound}.
Combining these bounds with Lemma~\ref{lem:lifting-to-unitaries} yields the desired statement.

\begin{lemma}\label{lem:lifting-to-unitaries fin}

Let $\mu$ be the distribution of the block $Z = U_{[p],[q]} \in \mathbb{C}^{p \times q}$ for $U \sim \Haar(m)$, and let $\mu_t$ be the distribution of the block $(1-t)Z + tH \in \mathbb{C}^{p \times q}$.
Assume that there exist fixed constants $c,C_{\mathrm{dim}}>0$ such that $p\leq C_{\mathrm{dim}}q$ and $m-p-q\geq cq$, and assume that $t$ satisfies Eq.~\eqref{eq:shift-scale-t-bound}.
There exists a randomized polynomial-time procedure $\widetilde{\mathcal{C}}$ that, with probability at least $1 - \exp({-\Omega(m-p-q)})$, completes the input matrix to a unitary containing it as the prescribed submatrix. 
Moreover, the corresponding pushforward measure $\mathcal{H}_t = \widetilde{\mathcal{C}}_{\#}\mu_t$ satisfies the total variation distance bound
\begin{align}
\|\mathcal{H}_t-\Haar(m)\|_{\mathrm{TVD}}
\le
\|\mu_t-\mu\|_{\mathrm{TVD}}
+
e^{-\tilde c(m-p - q)} ,
\end{align}
for a constant $\tilde c > 0$.
\end{lemma}

\section{Alternative approach for average-case hardness: leading coefficient extraction}
\label{appendix: leading coefficient extraction}

Inspired by the recent average-case hardness proposal~\cite{bouland2025exponential}, we present an alternative route to average-case hardness for TBS, namely the $\#\mathrm{P}$-hardness of ATE in Definition~\ref{def: thresholdBS} under a threshold Fourier-coefficient anti-concentration conjecture.
Compared with the main proof of Theorem~\ref{thm: threshold average case main}, this approach reduces the interpolation degree from $2n$ to $2k_0$, where $k_0=o(n)$ is the size of the worst-case permanent instance.
This leads to a substantially stronger additive-precision requirement.
The trade-off is that the proof requires an anti-concentration assumption on the threshold Fourier coefficient, as described in Conjecture~\ref{ass:tbs-tfca}.
We subsequently present the analogous PBS result using the same leading-coefficient method.

The main idea is as follows.
Instead of encoding $\Per(X_0)^2$ into the value of a Fourier coefficient at the endpoint $t=1$, we perturb a Haar-random block by a ``dilute'' matrix $A_{X_0}$ supported only on a $k_0\times k_0$ block, and encode $\Per(X_0)^2$ into the leading coefficient of the resulting $2k_0$-degree polynomial in the perturbation parameter $t$.
We then estimate this leading coefficient via an elementary coefficient-extraction procedure, which incurs an imprecision blowup of at most $\exp(O(k_0\log n))$, smaller than the $\exp(O(n\log n))$ blowup appearing in the main proofs.

\subsection{Threshold sparse perturbation and single-permanent collapse}
\label{subsec: appendix tbs leading collapse}

Fix photon number $n$ and mode number $m=\alpha n$ for a fixed $\alpha>2$.
Let $x\in\mathcal T_{m,n}$ be a threshold outcome with weight $|x|=k>n/2$, and let $O:=\supp(x)$.
Let $k_0<2k-n$ be the size of the worst-case permanent instance.

For the residual construction below, first choose a subset $H\subseteq O$ uniformly at random among all $k_0$-subsets of $O$.
Next, choose a subset 
\begin{align}
Q'\subseteq O\setminus H 
\end{align}
uniformly at random with $|Q'|=2k-n-k_0$, and define
\begin{align}
    Q:=H\sqcup Q',
    \qquad
    L:=O\setminus Q  ,
\label{eq:tbs-leading-QL}
\end{align}
so that $|Q|=2k-n$ and $|L|=n-k$.
Choose a set $E\subseteq[m]\setminus O$ uniformly at random with $|E|=n-k$, choose a uniformly random pairing between $L$ and $E$, and write
\begin{align}
    R:=L\sqcup E.
\label{eq:tbs-leading-R}
\end{align}
Finally, define
\begin{align}
    W:=Q\sqcup L\sqcup E,
    \qquad
    W':=Q'\sqcup L\sqcup E=W\setminus H  ,
\label{eq:tbs-leading-W}
\end{align}
implying that $|W|=n$.

Now set $n':=n-k_0$, and define the residual threshold outcome $x'\in\mathcal T_{m,n'}$ by
\begin{align}
    x'_i=
    \begin{cases}
    1, & i\in O\setminus H,\\
    0, & \text{otherwise}.
    \end{cases}
\label{eq:tbs-leading-xprime}
\end{align}
Then $|x'|=k-k_0$, $|Q'| = 2|x'|-n'$, and $ |L| = n'-|x'|$.
In particular, the same sets $Q'$, $L$, $E$, and $R$ define the partial-rotation construction for the $n'$-photon threshold outcome $x'$.
Moreover, defining $N:=2n-2k$, we have the consistency relation
\begin{align}
    2n'-2|x'|
    =
    2(n-k_0)-2(k-k_0)
    =
    2n-2k
    =
    N.
\label{eq:tbs-leading-frequency-consistency}
\end{align}
Thus, the original and residual threshold probabilities have the same top Fourier frequency $2N$ under the same rotated set $R$.

We now introduce the sparse perturbation that carries the worst-case instance.
Let $X_0\in\{0,1\}^{k_0\times k_0}$ be arbitrary, and fix an ordering of the rows in $H$.
For a normalization parameter $\kappa>0$, define the sparse matrix $A_{X_0}\in\mathbb C^{n\times n}$, with rows indexed by $W$, by
\begin{align}
    (A_{X_0})_{H,\{n'+1,\ldots,n\}}
    =
    \kappa X_0,
\label{eq:tbs-leading-sparse-A}
\end{align}
with all other entries equal to zero.
Let $A:=U_{W,[n]}$ for $U\sim\Haar(m)$, and define
\begin{align}
    A(t):=A+tA_{X_0}.
\label{eq:tbs-leading-translated-A}
\end{align}
Whenever $U(t)\in{\rm U}(m)$ is a unitary satisfying
\begin{align}
    U(t)_{W,[n]}=A(t),
\end{align}
we can evaluate the threshold Fourier coefficient $\tau_{2N}(x,U(t),R)$.

\begin{lemma}[Threshold single-permanent collapse for sparse translation]
\label{lem:tbs-leading-collapse}
Let $U(t)\in{\rm U}(m)$ be any unitary satisfying $(U(t))_{W,[n]}=A(t)$, with $A(t)$ defined in Eq.~\eqref{eq:tbs-leading-translated-A}.
Then $\tau_{2N}(x,U(t),R)$ is a polynomial in $t$ of degree at most $2k_0$, and its leading coefficient satisfies
\begin{align}
    [t^{2k_0}]\,\tau_{2N}(x,U(t),R)
    =
    \kappa^{2k_0}\Per(X_0)^2
    \tau_{2N}(x',U,R).
\label{eq:tbs-leading-collapse-identity}
\end{align}
Here, the coefficient on the right-hand side is the $n'$-photon threshold Fourier coefficient associated with the residual data $(x',Q',L,E,R)$.
\end{lemma}

\begin{proof}
By Lemma~\ref{lem:tbs-rowfilter}, the top threshold Fourier coefficient can be written as
\begin{align}
\tau_{2N}(x,U(t),R)
&=
\sum_{\substack{S\in\mathcal S_{m,n}\\ \supp(S)=O\\ s_i=1\;\forall i\in Q}}
\frac{1}{\prod_i s_i!}
\Per\left(\left((\Pi_Q+V_R^+)U(t)\right)_{\nu(S),[n]}\right)
\Per\left(\left((\Pi_Q+V_R^-)U(t)\right)_{\nu(S),[n]}\right)^*.
\label{eq:tbs-leading-collapse-start}
\end{align}
Since $H\subseteq Q$, every occupation pattern contributing to Eq.~\eqref{eq:tbs-leading-collapse-start} satisfies $s_i=1$ for all $i \in H$.
Deleting these $k_0$ fixed single photons gives a unique occupation pattern $S'\in\mathcal S_{m,n'}$ satisfying
\begin{align}
    \supp(S')=O\setminus H,
    \qquad
    s'_i=1
    \quad
    \forall i\in Q'.
\label{eq:tbs-leading-residual-occupation}
\end{align}
Conversely, every $S'$ satisfying Eq.~\eqref{eq:tbs-leading-residual-occupation} determines a unique $S$ contributing to Eq.~\eqref{eq:tbs-leading-collapse-start} by adding one photon to each mode in $H$.
The factorial factor is unchanged because the deleted occupations are all equal to one.

The perturbation $A_{X_0}$ is supported only on the rows indexed by $H$ and the columns $\{n'+1,\ldots,n\}$.
Moreover, because $H\subseteq Q$, the rows indexed by $H$ are unaffected by $V_R^\pm$ and are selected through the projector $\Pi_Q$.
Therefore, each permanent in Eq.~\eqref{eq:tbs-leading-collapse-start} has degree at most $k_0$ in $t$.
Expanding the permanent along the $k_0$ rows indexed by $H$ gives
\begin{align}
    \Per(M)
    =
    \sum_{\substack{I\subseteq[n]\\ |I|=k_0}}
    \Per(M_{H,I})
    \Per(M_{H^c,[n]\setminus I}).
\label{eq:tbs-leading-laplace}
\end{align}
To obtain the coefficient of $t^{k_0}$, the perturbation must be selected from every row in $H$.
Since the perturbation is supported only on the last $k_0$ columns, only the term $I=\{n'+1,\ldots,n\}$ contributes.
Thus, for each sign $\pm$ and each pair $(S,S')$ related by Eq.~\eqref{eq:tbs-leading-residual-occupation},
\begin{align}
&[t^{k_0}]
\Per\left(\left((\Pi_Q+V_R^\pm)U(t)\right)_{\nu(S),[n]}\right)
=
\kappa^{k_0}\Per(X_0)
\Per\left(\left((\Pi_{Q'}+V_R^\pm)U\right)_{\nu(S'),[n']}\right).
\label{eq:tbs-leading-one-permanent-collapse}
\end{align}
The coefficient of $t^{2k_0}$ in the product of the two permanents in Eq.~\eqref{eq:tbs-leading-collapse-start} is therefore the product of the two leading coefficients in Eq.~\eqref{eq:tbs-leading-one-permanent-collapse}.
Summing over $S'$, we obtain
\begin{align}
[t^{2k_0}]\,\tau_{2N}(x,U(t),R)
&=
\kappa^{2k_0}\Per(X_0)^2
\sum_{\substack{S'\in\mathcal S_{m,n'}\\ \supp(S')=O\setminus H\\ s'_i=1\;\forall i\in Q'}}
\frac{1}{\prod_i s'_i!}
\Per\left(\left((\Pi_{Q'}+V_R^+)U\right)_{\nu(S'),[n']}\right)
\Per\left(\left((\Pi_{Q'}+V_R^-)U\right)_{\nu(S'),[n']}\right)^*.
\end{align}
By Eq.~\eqref{eq:tbs-leading-frequency-consistency} and Lemma~\ref{lem:tbs-rowfilter} applied to the $n'$-photon residual instance, the summation is precisely $\tau_{2N}(x',U,R)$.
This proves Eq.~\eqref{eq:tbs-leading-collapse-identity}.
\end{proof}

\subsection{Threshold sparse translation and conditional completion}
\label{subsec: appendix tbs sparse translation}

The perturbation in Eq.~\eqref{eq:tbs-leading-translated-A} is a pure sparse translation rather than the affine path $(1-t)A+tA_\star$ used in the main proof.
The density argument in Appendix~\ref{appendix: section: translation invariance} nevertheless gives an analogous total variation distance bound for the square $n\times n$ threshold block.

\begin{lemma}[Threshold sparse translation and completion]
\label{lem:tbs-leading-translation}
Let $\mu$ be the distribution of $A=U_{W,[n]}$ for $U\sim\Haar(m)$, and let $\mu_t$ be the distribution of $A(t)=A+tA_{X_0}$, with $A_{X_0}$ defined in Eq.~\eqref{eq:tbs-leading-sparse-A} and $\kappa=m^{-1/2}$.
Suppose $m-2n=\Omega(n)$ and $m=\alpha n$.
Then there exist constants $c_0,c_1,c_2>0$ such that, whenever
\begin{align}
    0\leq t\leq\frac{c_0}{\sqrt\alpha\,n k_0},
\label{eq:tbs-leading-t-range}
\end{align}
we have
\begin{align}
    \|\mu_t-\mu\|_{\mathrm{TVD}}
    \leq
    c_1\sqrt\alpha\,n k_0t
    +
    e^{-c_2(m-2n)}.
\label{eq:tbs-leading-tv}
\end{align}
Moreover, there exists a randomized polynomial-time unitary-completion procedure that succeeds with probability at least $1-\exp(-\Omega(m-2n))$, whose output distribution $\mathcal H_t$ satisfies
\begin{align}
    \|\mathcal H_t-\Haar(m)\|_{\mathrm{TVD}}
    \leq
    c_1\sqrt\alpha\,n k_0t
    +
    e^{-\Omega(m-2n)}.
\label{eq:tbs-leading-completion}
\end{align}
\end{lemma}

\begin{proof}
Let $f$ denote the density of the square Haar block and let $f_t(B)=f(B-tA_{X_0})$ denote the density of its translated distribution.
The assumptions of Lemma~\ref{lem:boundary-tail-appD} hold with $p=q=n$.
Fix $\zeta\in(0,\zeta^*/2)$ and define $R_t(B):=f_t(B)/f(B)$ only for $B\in\mathcal E_\zeta$.
As in the proof of Lemma~\ref{lem:general-translation-appD},
\begin{align}
\|\mu_t-\mu\|_{\mathrm{TVD}}
&\leq
\frac12\int_{\mathcal E_\zeta}f(B)|R_t(B)-1|\,dB
+
\frac12\mu(\mathcal E_\zeta^c)
+
\frac12\mu_t(\mathcal E_\zeta^c).
\label{eq:tbs-leading-support-split}
\end{align}

On the interior event, let $W_s:=B-sA_{X_0}$ for $0\leq s\leq t$.
After decreasing $c_0$ if necessary, the range in Eq.~\eqref{eq:tbs-leading-t-range} ensures $t\|A_{X_0}\|_{\mathrm{op}}\leq\zeta/2$, and hence $\|W_s\|_{\mathrm{op}}\leq1-\zeta/2$.
For $\Phi(W):=\log\det(I_n-W^\dagger W)$, the same derivative calculation as in Appendix~\ref{appendix: section: translation invariance}, now with $dW_s/ds=-A_{X_0}$, gives
\begin{align}
\left|\frac{d}{ds}\Phi(W_s)\right|
\leq
C\sqrt n\,\|A_{X_0}\|_{\mathrm F}.
\end{align}
Since the pure translation has no Jacobian term,
\begin{align}
\log R_t(B)
=
(m-2n)\bigl[\Phi(B-tA_{X_0})-\Phi(B)\bigr],
\end{align}
and therefore the interior term in Eq.~\eqref{eq:tbs-leading-support-split} is bounded by
\begin{align}
C'(m-2n)\sqrt n\,\|A_{X_0}\|_{\mathrm F}t.
\end{align}

For the exterior terms, Lemma~\ref{lem:boundary-tail-appD} gives $\mu(\mathcal E_\zeta^c)\leq e^{-c(m-2n)}$.
Moreover, if $B\in\mathcal E_{2\zeta}$ and $t\|A_{X_0}\|_{\mathrm{op}}\leq\zeta$, then $B+tA_{X_0}\in\mathcal E_\zeta$.
Thus $\mu_t(\mathcal E_\zeta^c)\leq\mu(\mathcal E_{2\zeta}^c)\leq e^{-c'(m-2n)}$.
Combining these bounds with Eq.~\eqref{eq:tbs-leading-support-split} yields
\begin{align}
    \|\mu_t-\mu\|_{\mathrm{TVD}}
    \leq
    C'(m-2n)\sqrt n\,\|A_{X_0}\|_{\mathrm F}t
    +
    e^{-c_2(m-2n)}.
\label{eq:tbs-leading-tv-intermediate}
\end{align}

By construction, $A_{X_0}$ has at most $k_0^2$ nonzero entries, each of magnitude at most $\kappa=m^{-1/2}$, and hence
\begin{align}
    \|A_{X_0}\|_{\mathrm F}
    \leq
    \kappa k_0
    =
    \frac{k_0}{\sqrt{\alpha n}}.
\end{align}
Since $m-2n\leq m=\alpha n$, Eq.~\eqref{eq:tbs-leading-tv-intermediate} implies Eq.~\eqref{eq:tbs-leading-tv} after absorbing constants.
Finally, Eq.~\eqref{eq:tbs-leading-completion} follows from Lemma~\ref{lem:lifting-to-unitaries} and the conditional-completion procedure in Appendix~\ref{appendix: section: conditional-haar-completion}; the dimension condition $p\leq Cq$ holds here with $p=q=n$.
\end{proof}

\subsection{Coefficient extraction algorithms}
\label{subsec: appendix coefficient extraction}

We list the two elementary extraction procedures used in the reduction.
The key difference from the robust algorithms in Sec.~\ref{section: shared framework} is that the following procedures do not require an NP oracle but do require that all queried values satisfy the stated additive-error bound.
This is sufficient for the present conditional argument because we will take a union bound over all queried points.

\begin{lemma}[Fourier coefficient extraction]
\label{thm: fourier coeff ext}
Let $f(\theta)=\sum_{\ell=-D}^{D}a_\ell e^{i\ell\theta}$ be a trigonometric polynomial of degree at most $D$.
Fix $r\in\{-D,\ldots,D\}$, and let $L\geq D+|r|+1$.
For $j=0,\ldots,L-1$, set $\theta_j=2\pi j/L$.
If $|y_j-f(\theta_j)|\leq\epsilon$ for every $j$, then the estimator
\begin{align}
    \widehat a_r
    :=
    \frac{1}{L}\sum_{j=0}^{L-1}y_je^{-ir\theta_j}
\end{align}
satisfies
\begin{align}
    |\widehat a_r-a_r|\leq\epsilon.
\label{eq: error bound for coef ext}
\end{align}
\end{lemma}

\begin{proof}
Write $\epsilon_j:=y_j-f(\theta_j)$.
We have the orthogonality relation
\begin{align}
    \frac{1}{L}\sum_{j=0}^{L-1}e^{i(\ell-r)\theta_j}
    =
    \begin{cases}
    1, & \ell\equiv r\;({\rm mod}\;L),\\
    0, & \text{otherwise}.
    \end{cases}
\end{align}
Since $|\ell-r|\leq D+|r|<L$ for $-D\leq\ell\leq D$, the only allowed $\ell$ satisfying $\ell\equiv r\;({\rm mod}\;L)$ is $\ell=r$.
Therefore,
\begin{align}
    \frac{1}{L}\sum_{j=0}^{L-1}f(\theta_j)e^{-ir\theta_j}
    =
    a_r,
\end{align}
and accordingly,
\begin{align}
    \widehat a_r-a_r
    =
    \frac{1}{L}\sum_{j=0}^{L-1}\epsilon_je^{-ir\theta_j}.
\end{align}
Then Eq.~\eqref{eq: error bound for coef ext} follows from $|\epsilon_j|\leq\epsilon$.
\end{proof}

\begin{lemma}[Polynomial coefficient extraction]
\label{thm: polynomial coefficient extraction}
Let $p(t)=\sum_{j=0}^{d}p_jt^j\in\mathbb C[t]$ be a polynomial of degree at most $d$, and let $0<\Delta<1$.
Set $h:=\Delta/d$ and $t_i:=ih$ for $i=0,\ldots,d$.
Suppose $|y_i-p(t_i)|\leq\epsilon$ for every $i=0,\ldots,d$.
Then the estimator
\begin{align}
    \widehat p_d
    :=
    \frac{1}{d!}\left(\frac{d}{\Delta}\right)^d
    \sum_{i=0}^{d}(-1)^{d-i}\binom{d}{i}y_i
\label{eq: appendix finite difference estimator}
\end{align}
satisfies
\begin{align}
    |\widehat p_d-p_d|
    \leq
    \epsilon\exp(d\log\Delta^{-1}+O(d)).
\label{eq: appendix finite difference bound}
\end{align}
\end{lemma}

\begin{proof}
The coefficient of $t^d$ in the Lagrange basis polynomial at $t_i=ih$ is
\begin{align}
    [t^d]
    \prod_{\substack{0\leq j\leq d\\ j\neq i}}
    \frac{t-t_j}{t_i-t_j}
    =
    \frac{(-1)^{d-i}}{h^di!(d-i)!}.
\end{align}
Thus Eq.~\eqref{eq: appendix finite difference estimator} is exactly the top coefficient of the polynomial interpolating the noisy data points.
Applying the same identity to the error values $y_i-p(t_i)$ gives
\begin{align}
    |\widehat p_d-p_d|
    &\leq
    \epsilon\sum_{i=0}^{d}\frac{1}{h^di!(d-i)!}
    =
    \frac{\epsilon}{h^dd!}\sum_{i=0}^{d}\binom{d}{i}
    =
    \epsilon\frac{2^d}{d!h^d}
    \nonumber\\
    &=
    \epsilon\frac{(2d)^d}{d!\Delta^d}
    \leq
    \epsilon\exp(d\log\Delta^{-1}+O(d)),
\end{align}
where the last step uses Stirling's inequality for $d!$.
\end{proof}

\subsection{Threshold residual law and Haar-average weight ratio}
\label{subsec: appendix tbs residual law}

We next establish the two facts required in the leading-coefficient reduction.
First, deleting $k_0$ clicked modes from the sampled $x \sim \mathcal{D}_{m,n}$ produces a residual law that is asymptotically close to the $n'$-photon threshold outcome ensemble $\mathcal{D}_{m,n'}$. 
Second, the ratio of the corresponding Haar-average threshold probabilities has the right scale required in the error analysis.

For $1\leq n_0 \leq n$, define
\begin{align}
    \mu_{n_0}^{\rm T}
    :=
    \frac{mn_0}{m+{n_0}}  ,
\label{eq:tbs-leading-mu}
\end{align}
which is in fact the center of the click-count distribution for $y\sim\mathcal D_{m,{n_0}}$.
We introduce this notation to track the center as the photon number $n_0$ varies.

\begin{lemma}[Moderate click-count concentration]
\label{lem:tbs-leading-moderate}
Fix constants $\alpha>2$ and $0<\lambda<1$, let $m=\alpha n$, and let
\begin{align}
    n-n^{\lambda/2}\leq {n_0}\leq n.
\end{align}
For every fixed $\mathcal{Q}>0$, there exists a constant $C_{\mathcal{Q}}>0$ such that
\begin{align}
\Pr_{y\sim\mathcal D_{m,{n_0}}}
\left[
\left||y|-\mu_{n_0}^{\rm T}\right|
>
C_{\mathcal{Q}}\sqrt{n\log n}
\right]
\leq
n^{-\mathcal{Q}}.
\label{eq:tbs-leading-moderate}
\end{align}
\end{lemma}

\begin{proof}
Let $K:=|y|$.
By Eq.~\eqref{eq:main-threshold-average}, up to a global normalization factor, the probability mass at $K=r$ is
\begin{align}
    B_{{n_0},r}
    :=
    \binom mr\binom{{n_0}-1}{r-1},
    \qquad
    1\leq r\leq {n_0} ,
\end{align}
and its consecutive ratio is given by
\begin{align}
    \frac{B_{{n_0},r+1}}{B_{{n_0},r}}
    =
    \frac{(m-r)({n_0}-r)}{r(r+1)}
    <
    \frac{(m-r)({n_0}-r)}{r^2}.
\label{eq:tbs-leading-ratio}
\end{align}
The right-hand side equals $1$ when $r = \mu_{n_0}^{\rm T}$, and strictly decreases as $r$ increases from $\mu_{n_0}^{\rm T}$. 
Therefore, there exist constants $c,C>0$ such that, for every integer $s\geq C$ with $s=o(n)$,
\begin{align}
    \frac{B_{{n_0},\lceil\mu_{n_0}^{\rm T}\rceil+s+1}}
    {B_{{n_0},\lceil\mu_{n_0}^{\rm T}\rceil+s}}
    \leq
    \exp\left(-c\frac{s}{n}\right),
\end{align}
and the analogous bound holds for the reverse ratios on the lower side of the mode.
Iterating these bounds gives
\begin{align}
    B_{{n_0},\lceil\mu_{n_0}^{\rm T}\rceil+s}
    \leq
    B_{{n_0},\lceil\mu_{n_0}^{\rm T}\rceil}
    \exp\left(-c'\frac{s^2}{n}\right)
\end{align}
for a constant $c'>0$, and similarly for the lower tail.
Since the normalizing sum is at least the maximal term and contains at most $n$ terms, we obtain
\begin{align}
\Pr\left[
\left|K-\mu_{n_0}^{\rm T}\right|>u
\right]
\leq
2n\exp\left(-c'\frac{u^2}{n}\right)
\end{align}
for $u$ in the range used here.
Choosing $u=C_{\mathcal{Q}}\sqrt{n\log n}$ with a sufficiently large $C_{\mathcal{Q}}$ proves Eq.~\eqref{eq:tbs-leading-moderate}.
\end{proof}

For a threshold outcome $y\in\mathcal T_{m,N}$ with $h:=|y|>N/2$, we use the following randomized partial-rotation rule.
Choose $Q_y\subseteq\supp(y)$ uniformly with $|Q_y|=2h-N$, set $L_y:=\supp(y)\setminus Q_y$, choose $E_y\subseteq[m]\setminus\supp(y)$ uniformly with $|E_y|=N-h$, and choose a uniformly random pairing between $L_y$ and $E_y$.
We denote the resulting paired set by $R_y:=L_y\sqcup E_y$.
For $|y|\leq N/2$, we assign a fixed dummy value $\bot$ to the partial-rotation data; this event is included among the bad events in Conjecture~\ref{ass:tbs-tfca}.

Let $\mathcal R_{m,n,k_0}^{\rm T}$ denote the joint residual law obtained as follows.
Sample $x\sim\mathcal D_{m,n}$ conditioned on
\begin{align}
    \mathsf E_n^{\rm T}
    :=
    \left\{
    \left||x|-\mu_n^{\rm T}\right|
    \leq
    C_{\mathcal{Q}}\sqrt{n\log n}
    \right\},
\label{eq:tbs-leading-typical-event}
\end{align}
choose mode-index sets $H,Q',Q,L,E,R$ as in Eqs.~\eqref{eq:tbs-leading-QL}--\eqref{eq:tbs-leading-R}, and output the residual outcome $x'$ together with the residual partial-rotation data $(Q',L,E,R)$.

We now show that the residual law is asymptotically close to the threshold outcome ensemble $\mathcal{D}_{m,n'}$.

\begin{lemma}[Threshold residual-law is close to threshold outcome ensemble]
\label{lem:tbs-leading-residual}
Fix constants $\alpha>2$ and $0<\lambda<1$.
Let $m=\alpha n$, let $k_0\leq n^{\lambda/2}$, and set $n'=n-k_0$.
For every event $\mathcal B$ depending on a threshold outcome and its randomized partial-rotation data, it holds that
\begin{align}
\Pr_{(x',Q',L,E,R)\sim\mathcal R_{m,n,k_0}^{\rm T}}
\left[
(x',Q',L,E,R)\in\mathcal B
\right]
&\leq
 e^{o(1)}
\Pr_{\substack{y\sim\mathcal D_{m,n'}\\
(Q_y,L_y,E_y,R_y)}}
\left[
(y,Q_y,L_y,E_y,R_y)\in\mathcal B
\right]
+
 n^{-\mathcal{Q}},
\label{eq:tbs-leading-residual-domination}
\end{align}
for every constant $\mathcal{Q}>0$, where the partial-rotation data on the right-hand side are sampled according to the randomized rule above.
\end{lemma}

\begin{proof}
Let $q_{\rm T}:=\Pr_{x\sim\mathcal D_{m,n}}[\mathsf E_n^{\rm T}]$.
By Lemma~\ref{lem:tbs-leading-moderate}, $q_{\rm T}=1-n^{-\Omega(1)}$, and hence $q_{\rm T}^{-1}=e^{o(1)}$.

Fix a residual threshold string $y\in\mathcal T_{m,n'}$ and write $\ell:=|y|$.
To obtain $y$, the original string $x$ must be obtained by adding $k_0$ clicks outside $\supp(y)$, and there are $\binom{m-\ell}{k_0}$ possibilities.
Each such string has weight $\ell+k_0$ and therefore has $\mathcal D_{m,n}$-probability
\begin{align}
    \frac{\binom{n-1}{\ell+k_0-1}}
    {\binom{m+n-1}{n}}.
\end{align}
Given such an $x$, the set $H$ is chosen uniformly among the $\binom{\ell+k_0}{k_0}$ subsets of $\supp(x)$ of size $k_0$.
Thus
\begin{align}
\mathcal R_{m,n,k_0}^{\rm T}(y)
\leq
q_{\rm T}^{-1}
\frac{\binom{m-\ell}{k_0}}
{\binom{\ell+k_0}{k_0}}
\frac{\binom{n-1}{\ell+k_0-1}}
{\binom{m+n-1}{n}}.
\label{eq:tbs-leading-residual-law}
\end{align}
On the other hand, the random threshold outcome for $y$ has the probability weight
\begin{align}
    \mathcal D_{m,n'}(y)
    =
    \frac{\binom{n'-1}{\ell-1}}
    {\binom{m+n'-1}{n'}}.
\label{eq:tbs-leading-original-law}
\end{align}
Dividing Eq.~\eqref{eq:tbs-leading-residual-law} by Eq.~\eqref{eq:tbs-leading-original-law} gives
\begin{align}
\frac{\mathcal R_{m,n,k_0}^{\rm T}(y)}
{\mathcal D_{m,n'}(y)}
&\leq
q_{\rm T}^{-1}
\frac{\binom{m-\ell}{k_0}}
{\binom{\ell+k_0}{k_0}}
\frac{\binom{n-1}{\ell+k_0-1}}
{\binom{n'-1}{\ell-1}}
\frac{\binom{m+n'-1}{n'}}
{\binom{m+n-1}{n}}
=
q_{\rm T}^{-1}
\prod_{j=1}^{k_0}
\frac{(m-\ell-j+1)(n-j)(n'+j)}
{(\ell+j)(\ell+j-1)(m+n'+j-1)}.
\label{eq:tbs-leading-residual-ratio}
\end{align}

Under the residual construction, $|x'|=|x|-k_0$.
Moreover, because $\mu_{n_0}^{\rm T}=mN/(m+N)$ has derivative bounded by a constant for ${n_0} \in[n-k_0,n]$, we have 
\begin{align}
    \mu_n^{\rm T}-k_0-\mu_{n'}^{\rm T}
    =
    O(k_0).
\end{align}
Since $k_0=o(\sqrt{n\log n})$, Eq.~\eqref{eq:tbs-leading-typical-event} implies
\begin{align}
    |\ell - \mu_{n'}^{\rm T} |
    =
    ||x| - \mu_{n}^{\rm T} + \mu_{n}^{\rm T} - k_0 - \mu_{n'}^{\rm T}|
    \leq
    O(\sqrt{n\log n}).
\label{eq:tbs-leading-residual-window}
\end{align}
For such $\ell$, we define the function
\begin{align}
    F(\ell)
    :=
    \frac{(m-\ell)nn'}{\ell^2(m+n')}.
\end{align}
Each factor in Eq.~\eqref{eq:tbs-leading-residual-ratio} admits the exact decomposition
\begin{align}
&\frac{(m-\ell-j+1)(n-j)(n'+j)}
{(\ell+j)(\ell+j-1)(m+n'+j-1)}
\nonumber\\
&=F(\ell)
\left(1-\frac{j-1}{m-\ell}\right)
\left(1-\frac{j}{n}\right)
\left(1+\frac{j}{n'}\right)
\left(1+\frac{j}{\ell}\right)^{-1}
\left(1+\frac{j-1}{\ell}\right)^{-1}
\left(1+\frac{j-1}{m+n'}\right)^{-1}.
\label{eq:tbs-leading-residual-factorization}
\end{align}
At $\ell=\mu_{n'}^{\rm T}$, we have 
\begin{align}
    F(\mu_{n'}^{\rm T})
    =
    \frac{n}{n'}
    =
    1+O\left(\frac{k_0}{n}\right) ,
\end{align}
and its derivative is given by
\begin{align}
    \frac{\partial}{\partial\ell}\log F(\ell)
    =
    -\frac{1}{m-\ell}
    -\frac{2}{\ell}
    =
    O\left(\frac{1}{n}\right)
\end{align}
inside the window in Eq.~\eqref{eq:tbs-leading-residual-window}.
Since $j\leq k_0=o(n)$ and all denominators in Eq.~\eqref{eq:tbs-leading-residual-factorization} are $\Theta(n)$, the remaining factors contribute $O(k_0/n)$ to the logarithm of each product term.
Therefore,
\begin{align}
\log
\frac{\mathcal R_{m,n,k_0}^{\rm T}(y)}
{\mathcal D_{m,n'}(y)}
&\leq
|\log q_{\rm T}|
+
O\left(
 k_0\sqrt{\frac{\log n}{n}}
 +
 \frac{k_0^2}{n}
\right)
=
 o(1),
\label{eq:tbs-leading-residual-pointwise}
\end{align}
because $k_0\leq n^{\lambda/2}$ and $0<\lambda<1$.

It remains to include the partial-rotation data.
Conditioned on the residual string $y$, the set $Q'$ is uniform among the subsets of $\supp(y)$ of size $2|y|-n'$ by construction.
The removed set $H$ is symmetric over the complement of $\supp(y)$, and, conditioned on $H$, the set $E$ is uniform among the $(n'-|y|)$-subsets disjoint from $H$.
For every fixed candidate $E_0\subseteq[m]\setminus\supp(y)$ of this size, the number of possible sets $H$ disjoint from $E_0$ is the same.
Consequently, after averaging over $H$, the marginal distribution of $E$ is uniform among all $(n'-|y|)$-subsets of $[m]\setminus\supp(y)$.
The matching is also uniform.
Thus, conditioned on $y$, the residual data $(Q',L,E,R)$ follow exactly the randomized partial-rotation rule for the $n'$-photon outcome $y$.
Combining this fact with Eq.~\eqref{eq:tbs-leading-residual-pointwise}, and absorbing the negligible complement of the moderate window by Lemma~\ref{lem:tbs-leading-moderate}, proves Eq.~\eqref{eq:tbs-leading-residual-domination}.
\end{proof}

Next, we characterize the ratio of the Haar-averaged threshold probabilities associated with $x\in\mathcal T_{m,n}$ of weight $k=|x|$ and with $x'\in\mathcal T_{m,n'}$, obtained by deleting $k_0$ clicks from $x$.

\begin{lemma}[Threshold Haar-average weight ratio]
\label{lem:tbs-leading-weight-ratio}
Let $m=\alpha n$ for a fixed $\alpha>2$, let $k_0\leq n^{\lambda/2}$ for a fixed $0 < \lambda < 1$, and set $n'=n-k_0$.
Let $x\in\mathcal T_{m,n}$ have weight $k=|x|$, and let $x'\in\mathcal T_{m,n'}$ be obtained by deleting $k_0$ clicks from $x$.
Then
\begin{align}
    \frac{T_x^{(\mathrm{avg})}}
    {T_{x'}^{(\mathrm{avg})}}
    =
    \prod_{j=1}^{k_0}
    \frac{(n-j)(n'+j)}
    {(k-j)(m+n'+j-1)}.
\label{eq:tbs-leading-weight-ratio-exact}
\end{align}
Moreover, if $\left|k-\mu_n^{\rm T}\right| \leq  C\sqrt{n\log n}$ for a constant $C>0$, then
\begin{align}
    \frac{T_x^{(\mathrm{avg})}}
    {T_{x'}^{(\mathrm{avg})}}
    =
    \alpha^{-k_0}e^{o(1)}.
\label{eq:tbs-leading-weight-ratio-asymptotic}
\end{align}
\end{lemma}

\begin{proof}
By Eq.~\eqref{eq:main-threshold-average},
\begin{align}
    T_x^{(\mathrm{avg})}
    =
    \frac{\binom{n-1}{k-1}}
    {\binom{m+n-1}{n}},
    \qquad
    T_{x'}^{(\mathrm{avg})}
    =
    \frac{\binom{n'-1}{k-k_0-1}}
    {\binom{m+n'-1}{n'}}.
\end{align}
Taking the ratio and expanding the factorials gives Eq.~\eqref{eq:tbs-leading-weight-ratio-exact}.

For the asymptotic form, note that the factor in Eq.~\eqref{eq:tbs-leading-weight-ratio-exact} at $k=\mu_n^{\rm T}=mn/(m+n)$ and $j=0$ gives 
\begin{align}
    \frac{nn'}{\mu_n^{\rm T}(m+n'- 1)}
    =
    \frac{1}{\alpha}
    \frac{(1-k_0/n)(\alpha+1)}
    {\alpha+1-k_0/n - 1/n}
    =
    \frac{1}{\alpha}
    \left(1+O\left(\frac{k_0}{n}\right)\right).
\end{align}
Perturbing $k$ by $O(\sqrt{n\log n})$ and $j$ by at most $k_0$ changes the logarithm of each factor by at most 
\begin{align}
    O\left(
    \sqrt{\frac{\log n}{n}}
    +
    \frac{k_0}{n}
    \right).
\end{align}
Therefore, for each factor in Eq.~\eqref{eq:tbs-leading-weight-ratio-exact}, we have 
\begin{align}
    \log(\frac{(n-j)(n'+j)}
    {(k-j)(m+n'+j-1)})
    =
    -\log\alpha
    +
    O\left(
    \sqrt{\frac{\log n}{n}}
    +
    \frac{k_0}{n}
    \right).
\end{align}
Provided $k_0\leq n^{\lambda/2}$ and $0<\lambda<1$, summing over $j=1,\ldots,k_0$ proves Eq.~\eqref{eq:tbs-leading-weight-ratio-asymptotic}.
\end{proof}

\subsection{Threshold average-case hardness with improved imprecision}
\label{subsec: appendix tbs conditional hardness}

We now prove the average-case hardness of TBS using the leading-coefficient extraction approach.
To this end, we first state the threshold Fourier-coefficient anti-concentration assumption required by the reduction.
The randomized partial-rotation rule in the conjecture is the one defined in Sec.~\ref{subsec: appendix tbs residual law}.

\begin{conjecture}[Threshold Fourier-coefficient anti-concentration]
\label{ass:tbs-tfca}
Let $s$ be a photon number satisfying $m/s>2$ and $m/s=\Theta(1)$.
Draw $y\sim\mathcal D_{m,s}$ and $U\sim\Haar(m)$ independently.
Whenever $|y|>s/2$, draw $(Q_y,L_y,E_y,R_y)$ according to the randomized partial-rotation rule in Sec.~\ref{subsec: appendix tbs residual law}, and define $N_y:=2s-2|y|$.
Then
\begin{align}
\Pr\left[
|y|\leq\frac{s}{2}
\;\;{\rm{or}}\;\;
|\tau_{2N_y}(y,U,R_y)|
\leq
\frac{T_y^{(\mathrm{avg})}}{\poly(s)}
\right]
\leq
\frac{1}{\poly(s)},
\label{eq:tbs-leading-anticoncentration}
\end{align}
where the probability is over $y$, $U$, and the randomized partial-rotation data.
\end{conjecture}

Note that Conjecture~\ref{ass:tbs-tfca} is the threshold Fourier-coefficient analog of the anti-concentration conjecture used in~\cite{bouland2026complexity}.
The selection rule is included explicitly because the threshold Fourier coefficient depends on the decomposition into unrotated and rotated clicked modes.

\begin{theorem}
\label{thm: threshold average case leading}
Fix any constant $0<\lambda<1$.
Let $m=m(n)$ satisfy $m/n=\alpha+o(1)$ for a fixed $\alpha>2$.
Suppose Conjecture~\ref{ass:tbs-tfca} holds.
Then the {\rm ATE} problem in Definition~\ref{def: thresholdBS} is $\#\mathrm P$-hard under a $\rm BPP$ reduction for error bounds satisfying
\begin{align}
    \varepsilon
    =
    \exp(-O(n^\lambda)),
    \qquad
    \delta
    =
    O\left(\frac{1}{n^{1+\lambda}}\right).
\label{eq:tbs-leading-error-parameters}
\end{align}
\end{theorem}

\begin{proof}
Let $\mathcal O$ be an oracle for ATE, available for all input sizes in the reduction.
The residual call below has photon number $n'=n-o(n)$ and the same mode number $m$, so $m/n'=\alpha+o(1)$; throughout, the oracle guarantee is understood uniformly over this asymptotically equivalent fixed-ratio family.
Concretely, on input $x\in\mathcal T_{m,n}$ and $U\in{\rm U}(m)$, it outputs an additive estimate of $T_x(U)$ satisfying
\begin{align}
\Pr_{\substack{x\sim\mathcal D_{m,n}\\ U\sim\Haar(m)}}
\left[
|\mathcal O(x,U)-T_x(U)|
>
\varepsilon T_x^{(\mathrm{avg})}
\right]
<
\delta.
\label{eq:tbs-leading-ATE-oracle}
\end{align}
We show how to compute $\Per(X_0)^2$ exactly for an arbitrary $X_0\in\{0,1\}^{k_0\times k_0}$.

Choose $n:=\left\lceil k_0^{2/\lambda}\right\rceil$ and $ m:=\alpha n$, so that $k_0\leq n^{\lambda/2}$ and $k_0=o(n)$.
Sample $x\sim\mathcal D_{m,n}$ and set $k:=|x|$.
By Lemma~\ref{lem:tbs-leading-moderate}, for every fixed $\mathcal{Q}>0$, with probability at least $1-n^{-\mathcal{Q}}$,
\begin{align}
    \left|k-\mu_n^{\rm T}\right|
    \leq
    C_{\mathcal{Q}}\sqrt{n\log n}.
\label{eq:tbs-leading-main-typical}
\end{align}
For fixed $\alpha>2$, this event implies, $k>\frac n2$, $2k-n=\Theta(n)$, $n-k=\Theta(n)$, and $2k-n>k_0$, for all sufficiently large $n$.
Conditioned on Eq.~\eqref{eq:tbs-leading-main-typical}, construct $H,Q',Q,L,E,R,W,W'$, $n'$, and $x'$ according to Eqs.~\eqref{eq:tbs-leading-QL}--\eqref{eq:tbs-leading-xprime}, including the randomized choices specified there.
Let $\mathcal R_{m,n,k_0}^{\rm T}$ denote the resulting joint residual law.

Given $X_0$, define the sparse perturbation $A_{X_0}$ by Eq.~\eqref{eq:tbs-leading-sparse-A} with $\kappa:=m^{-1/2}$.
Let $A=U_{W,[n]}$ for $U\sim\Haar(m)$, and set $A(t)=A+tA_{X_0}$.
Define $p(t) := \tau_{2N}(x,U(t),R)$ and $N:=2n-2k$, where $U(t)$ is a unitary completion of $A(t)$.
By Lemma~\ref{lem:tbs-leading-collapse}, $p(t)$ is a polynomial of degree at most $2k_0$ and
\begin{align}
    [t^{2k_0}]p(t)
    =
    \kappa^{2k_0}\Per(X_0)^2
    \tau_{2N}(x',U,R).
\label{eq:tbs-leading-main-collapse}
\end{align}

We first estimate $p(t)$ for $t$ in a small interval.
Choose
\begin{align}
    \Delta
    :=
    \frac{c_\Delta\delta}{\sqrt\alpha\,n k_0}
\label{eq:tbs-leading-Delta}
\end{align}
for a sufficiently small constant $c_\Delta>0$.
By Lemma~\ref{lem:tbs-leading-translation}, for every $t\in[0,\Delta]$, the completed unitary $U(t)$ has distribution $\mathcal H_t$ satisfying
\begin{align}
    \|\mathcal H_t-\Haar(m)\|_{\mathrm{TVD}}
    \leq
    O(\delta)+e^{-\Omega(n)}.
\label{eq:tbs-leading-Ht-close}
\end{align}
Since left multiplication by $V_R(\theta)$ preserves Haar measure, each oracle call on $V_R(\theta)U(t)$ fails with probability at most $O(\delta)+e^{-\Omega(n)}$, after conditioning on Eq.~\eqref{eq:tbs-leading-main-typical}.

Let $L_F:=4n+1$ and $\theta_j:=2\pi j/L_F$ for $j=0,\ldots,L_F-1$.
By Lemma~\ref{lem:tbs-rowfilter}, $T_x(V_R(\theta)U(t))$ has trigonometric degree at most $2N\leq2n$.
Therefore, Lemma~\ref{thm: fourier coeff ext}, applied with $D=2n$ and $r=2N$, estimates $p(t)$ to additive error $\varepsilon T_x^{(\mathrm{avg})}$ whenever all $L_F$ oracle calls succeed.
By the union bound, this Fourier extraction fails with probability $O(n\delta)+e^{-\Omega(n)}$.

We evaluate this procedure at the $2k_0+1$ equally spaced points
\begin{align}
    t_i:=\frac{i\Delta}{2k_0},
    \qquad
    i=0,1,\ldots,2k_0.
\end{align}
If all Fourier extractions succeed, Lemma~\ref{thm: polynomial coefficient extraction}, with $d=2k_0$, gives an estimate $\widetilde c_{2k_0}$ of $[t^{2k_0}]p(t)$ satisfying
\begin{align}
    \left|
    \widetilde c_{2k_0}-[t^{2k_0}]p(t)
    \right|
    \leq
    \varepsilon T_x^{(\mathrm{avg})}
    \exp(2k_0\log\Delta^{-1}+O(k_0)).
\label{eq:tbs-leading-numerator-error}
\end{align}
The failure probability of Eq.~\eqref{eq:tbs-leading-numerator-error} is $O(nk_0\delta)+e^{-\Omega(n)}$, which is $o(1)$ under Eq.~\eqref{eq:tbs-leading-error-parameters}.

We next estimate the denominator $\tau_{2N}(x',U,R)$.
The residual pair consisting of $x'$ and its inherited partial-rotation data follows $\mathcal R_{m,n,k_0}^{\rm T}$.
By Lemma~\ref{lem:tbs-leading-residual}, the ATE-oracle failure probability under this residual law is at most $e^{o(1)}\delta+n^{-\mathcal{Q}}$ for any fixed $\mathcal{Q}>0$.
Applying Lemma~\ref{thm: fourier coeff ext} to the $n'$-photon threshold instance gives an estimate $\widehat\tau$ satisfying
\begin{align}
    |\widehat\tau- \tau_{2N}(x',U,R)|
    \leq
    \varepsilon T_{x'}^{(\mathrm{avg})},
\label{eq:tbs-leading-denominator-error}
\end{align}
except with failure probability $O(n\delta+n^{1-\mathcal{Q}})$.

The same residual-law domination transfers Conjecture~\ref{ass:tbs-tfca} to the actual residual construction.
Thus, with probability $1-o(1)$, we have 
\begin{align}
    |\tau_{2N}(x',U,R)|
    \geq
    \frac{T_{x'}^{(\mathrm{avg})}}{\poly(n)}.
\label{eq:tbs-leading-anticonc-event}
\end{align}
By a union bound, Eqs.~\eqref{eq:tbs-leading-main-typical}, \eqref{eq:tbs-leading-numerator-error}, \eqref{eq:tbs-leading-denominator-error}, and \eqref{eq:tbs-leading-anticonc-event} hold simultaneously with probability at least $2/3$ for all sufficiently large $n$.

Now, for brevity, let $a := \kappa^{2k_0}\Per(X_0)^2$ and $B:= \tau_{2N}(x',U,R)$.
By Eq.~\eqref{eq:tbs-leading-main-collapse}, we have $[t^{2k_0}]p(t)=aB$.
Combining Eqs.~\eqref{eq:tbs-leading-numerator-error}, \eqref{eq:tbs-leading-denominator-error}, and \eqref{eq:tbs-leading-anticonc-event}, and using that $\varepsilon=o(\poly(n)^{-1})$, gives
\begin{align}
\left|
\frac{|\widetilde c_{2k_0}|}{|\widehat\tau|}
-a
\right|
&\leq
\frac{
|\widetilde c_{2k_0}-aB|
+a|\widehat\tau-B|
}{|B|-|\widehat\tau-B|}
\leq
\varepsilon
\left(
\frac{T_x^{(\mathrm{avg})}}
{T_{x'}^{(\mathrm{avg})}}
\right)
\poly(n)
\exp(2k_0\log\Delta^{-1}+O(k_0))
+
\varepsilon\poly(n)a.
\label{eq:tbs-leading-ratio-error}
\end{align}
By Lemma~\ref{lem:tbs-leading-weight-ratio}, we have  ${T_x^{(\mathrm{avg})}}/{T_{x'}^{(\mathrm{avg})}} = \alpha^{-k_0}e^{o(1)}$.
Moreover, Eq.~\eqref{eq:tbs-leading-Delta} gives $\Delta^{-1}=\poly(n)$ for fixed $\alpha$, while $\kappa^{2k_0}=m^{-k_0}$ and $\Per(X_0)^2\leq(k_0!)^2$.
Therefore Eq.~\eqref{eq:tbs-leading-ratio-error} implies
\begin{align}
    \left|
    \frac{|\widetilde c_{2k_0}|}{|\widehat\tau|}
    -
    \kappa^{2k_0}\Per(X_0)^2
    \right|
    \leq
    \varepsilon\exp(O(k_0\log n)).
\label{eq:tbs-leading-scaled-error}
\end{align}
Dividing by $\kappa^{2k_0}$ gives an estimate of $\Per(X_0)^2$ with additive error at most
\begin{align}
    \varepsilon\kappa^{-2k_0}
    \exp(O(k_0\log n))
    =
    \varepsilon\exp(O(k_0\log n)).
\end{align}
Since $k_0\leq n^{\lambda/2}$, choosing $\varepsilon=\exp(-O(n^\lambda))$ with a sufficiently large constant in the exponent makes this error strictly smaller than $1/2$.
Rounding therefore recovers $\Per(X_0)^2$ exactly, and repetition amplifies the success probability.
This proves the theorem.
\end{proof}

\subsection{Extension to PBS: sparse collapse and translation}
\label{subsec: appendix parity leading collapse}

We now present the analogous PBS result.
Fix photon number $n$ and mode number $m=\alpha n$.
Let $x\in\mathcal P_{m,n}$ be a parity outcome with weight $|x|=k=\Theta(n)$, let $O:=\supp(x)$, and choose a set $E\subseteq[m]\setminus O$ uniformly with $|E|=k$ together with a uniformly random pairing
\begin{align}
    R:=O\sqcup E
    =
    \{o_1,e_1\}\sqcup\cdots\sqcup\{o_k,e_k\}.
\label{eq: def of R in LC}
\end{align}
Let $k_0<k$.
Choose $J_O\subseteq O$ uniformly among the $k_0$-subsets of $O$, let $J_E\subseteq E$ be its matched subset, and define
\begin{align}
    J:=J_O\sqcup J_E,
    \qquad
    R':=R\setminus J,
    \qquad
    n':=n-k_0.
\end{align}
The residual parity outcome $x'\in\mathcal P_{m,n'}$ is defined by
\begin{align}
    x'_i=
    \begin{cases}
    1, & i\in O\setminus J_O,\\
    0, & \text{otherwise}.
    \end{cases}
\label{eq: appendix def of x prime}
\end{align}
Then $|x'|=k-k_0$ and $n'-|x'|=n-k$ is even.

Let $X_0\in\{0,1\}^{k_0\times k_0}$ be arbitrary.
For a normalization parameter $\kappa>0$, define the sparse matrix $A_{X_0}\in\mathbb C^{2k\times n}$ by
\begin{align}
    (A_{X_0})_{J_O,\{n'+1,\ldots,n\}}
    =
    \kappa X_0,
\label{eq: appendix sparse AX}
\end{align}
with all other entries equal to zero.
Let $A:=U_{R,[n]}$ for $U\sim\Haar(m)$, and define
\begin{align}
    A(t):=A+tA_{X_0}.
\label{eq: appendix translated AX}
\end{align}
Whenever a unitary $U(t)\in{\rm U}(m)$ satisfies $(U(t))_{R,[n]}=A(t)$, we can evaluate $c_{2n'}(x,U(t),R')$.

\begin{lemma}[Parity single-permanent collapse for sparse translation]
\label{lem:collapse}
Let $U(t)\in{\rm U}(m)$ be any unitary satisfying $(U(t))_{R,[n]}=A(t)$, with $A(t)$ defined in Eq.~\eqref{eq: appendix translated AX}.
Then $c_{2n'}(x,U(t),R')$ is a polynomial in $t$ of degree at most $2k_0$, and its leading coefficient satisfies
\begin{align}
    [t^{2k_0}]\,c_{2n'}(x,U(t),R')
    =
    \kappa^{2k_0}\Per(X_0)^2
    c_{2n'}(x',U,R').
\label{eq: appendix collapse identity}
\end{align}
\end{lemma}

\begin{proof}
Write $z=e^{i\theta}$.
By definition,
\begin{align}
    c_{2n'}(x,U(t),R')
    =
    \sum_{\substack{S\in\mathcal S_{m,n}\\ S\;{\rm mod}\;2=x}}
    \frac{1}{\prod_i s_i!}
    \frac{1}{2\pi}
    \int_0^{2\pi}
    \left|\Per\left((V_{R'}(\theta)U(t))_{\nu(S),[n]}\right)\right|^2
    z^{-2n'}d\theta.
\label{eq: appendix collapse start}
\end{align}
For a fixed $S$, let $n_{R'}(S):=\sum_{i\in R'}s_i$.
Only these rows acquire factors $z^{\pm1}$ under $V_{R'}(\theta)$.
Thus the coefficient of $z^{2n'}$ vanishes unless $n_{R'}(S)\geq n'$.
Since the total photon number is $n=n'+k_0$, this condition implies that at most $k_0$ photons lie in $J$.
On the other hand, every mode in $J_O$ has odd occupation and therefore contains at least one photon.
Since $|J_O|=k_0$, every contributing pattern satisfies
\begin{align}
    s_i=1
    \quad
    \forall i\in J_O,
    \qquad
    s_i=0
    \quad
    \forall i\in J_E,
\label{eq: appendix selected occupation}
\end{align}
and all remaining $n'$ photons lie in $R'$.
Deleting the fixed single photons in $J_O$ gives an occupation pattern $S'\in\mathcal S_{m,n'}$ satisfying $S'\;{\rm mod}\;2=x'$.

Similarly to the TBS case, we use the Laplace expansion of permanent
\begin{align}
    \Per(M)
    =
    \sum_{\substack{I\subseteq[n]\\ |I|=k_0}}
    \Per(M_{J_O,I})
    \Per(M_{J_O^c,[n]\setminus I}).
\label{eq: appendix laplace permanent}
\end{align}
Because $A_{X_0}$ is supported only on $J_O\times\{n'+1,\ldots,n\}$, extracting the coefficient of $t^{2k_0}$ selects this perturbation block once in the permanent and once in the conjugate permanent.
Therefore,
\begin{align}
[t^{2k_0}]\,c_{2n'}(x,U(t),R')
&=
\kappa^{2k_0}\Per(X_0)^2
\sum_{\substack{S'\in\mathcal S_{m,n'}\\ S'\;{\rm mod}\;2=x'}}
\frac{1}{\prod_i s'_i!}
\Per\left((V_{R'}^+U)_{\nu(S'),[n']}\right)
\Per\left((V_{R'}^-U)_{\nu(S'),[n']}\right)^*.
\end{align}
The summation is precisely $c_{2n'}(x',U,R')$ by Lemma~\ref{lem:rowfilter} applied to the $n'$-photon residual instance.
This proves Eq.~\eqref{eq: appendix collapse identity}.
\end{proof}

\begin{lemma}[Parity sparse translation and completion]
\label{lem: appendix sparse translation}
Let $\mu$ be the distribution of $A=U_{R,[n]}$ for $U\sim\Haar(m)$ and $|R|=2k$, and let $\mu_t$ be the distribution of $A(t)=A+tA_{X_0}$, with $A_{X_0}$ defined in Eq.~\eqref{eq: appendix sparse AX} and $\kappa=m^{-1/2}$.
Suppose $m-n-2k=\Omega(n)$ and $m=\alpha n$.
Then there exist constants $c_0,c_1,c_2>0$ such that, whenever
\begin{align}
    0\leq t\leq\frac{c_0}{\sqrt\alpha\,n k_0},
\end{align}
we have
\begin{align}
    \|\mu_t-\mu\|_{\mathrm{TVD}}
    \leq
    c_1\sqrt\alpha\,n k_0t
    +
    e^{-c_2(m-n-2k)}.
\label{eq: appendix sparse tv bound}
\end{align}
Moreover, there exists a randomized polynomial-time unitary-completion procedure that succeeds with probability at least $1-\exp(-\Omega(m-n-2k))$, whose output distribution $\mathcal H_t$ satisfies
\begin{align}
    \|\mathcal H_t-\Haar(m)\|_{\mathrm{TVD}}
    \leq
    c_1\sqrt\alpha\,n k_0t
    +
    e^{-\Omega(m-n-2k)}.
\label{eq: appendix sparse completion bound}
\end{align}
\end{lemma}

\begin{proof}
Let $f$ be the density of the $2k\times n$ Haar block and let $f_t(B)=f(B-tA_{X_0})$.
Here $p=2k\leq2n=2q$, so the additional dimension assumption in Lemma~\ref{lem:boundary-tail-appD} is satisfied.
Fix $\zeta\in(0,\zeta^*/2)$ and define $R_t(B):=f_t(B)/f(B)$ only on $\mathcal E_\zeta$.
Then
\begin{align}
\|\mu_t-\mu\|_{\mathrm{TVD}}
&\leq
\frac12\int_{\mathcal E_\zeta}f(B)|R_t(B)-1|\,dB
+
\frac12\mu(\mathcal E_\zeta^c)
+
\frac12\mu_t(\mathcal E_\zeta^c).
\label{eq:parity-leading-support-split}
\end{align}

On $\mathcal E_\zeta$, the path $W_s=B-sA_{X_0}$ remains in the support when $t\|A_{X_0}\|_{\mathrm{op}}$ is sufficiently small.
The same derivative calculation as in the threshold sparse-translation proof gives
\begin{align}
\frac12\int_{\mathcal E_\zeta}f(B)|R_t(B)-1|\,dB
\leq
C'(m-n-2k)\sqrt n\,\|A_{X_0}\|_{\mathrm F}t.
\end{align}
The two exterior probabilities are bounded by Lemma~\ref{lem:boundary-tail-appD}: as before, $\mu(\mathcal E_\zeta^c)\leq e^{-c(m-n-2k)}$, while $B\in\mathcal E_{2\zeta}$ and $t\|A_{X_0}\|_{\mathrm{op}}\leq\zeta$ imply $B+tA_{X_0}\in\mathcal E_\zeta$, yielding $\mu_t(\mathcal E_\zeta^c)\leq e^{-c'(m-n-2k)}$.
Consequently,
\begin{align}
    \|\mu_t-\mu\|_{\mathrm{TVD}}
    \leq
    C'(m-n-2k)\sqrt n\,\|A_{X_0}\|_{\mathrm F}t
    +
    e^{-c_2(m-n-2k)}.
\end{align}

By construction,
\begin{align}
    \|A_{X_0}\|_{\mathrm F}
    \leq
    \kappa k_0
    =
    \frac{k_0}{\sqrt{\alpha n}}.
\end{align}
Since $m-n-2k\leq m=\alpha n$, Eq.~\eqref{eq: appendix sparse tv bound} follows after absorbing constants.
Equation~\eqref{eq: appendix sparse completion bound} follows from Lemma~\ref{lem:lifting-to-unitaries}; the dimension condition $p\leq Cq$ holds here with $p=2k\leq2n=2q$.
\end{proof}

\subsection{Parity residual law and analogous hardness result}
\label{subsec: appendix parity conditional hardness}

For an $N$-photon parity outcome, we similarly define
\begin{align}
    \mu_{n_0}^{\rm P}
    :=
    \frac{m{n_0}}{m+2{n_0}}.
\end{align}

\begin{lemma}[Moderate odd-parity concentration]
\label{lem:parity-leading-moderate}
Fix constants $\alpha>2$ and $0<\lambda<1$, let $m=\alpha n$, and let $n-n^{\lambda/2}\leq {n_0} \leq n$.
For every constant $\mathcal{Q}>0$, there exists a constant $C_\mathcal{Q}>0$ such that
\begin{align}
\Pr_{y\sim\mathcal Y_{m,{n_0}}}
\left[
\left||y|-\mu_{n_0}^{\rm P}\right|
>
C_{\mathcal{Q}}\sqrt{n\log n}
\right]
\leq
n^{-\mathcal{Q}}.
\label{eq:parity-leading-moderate}
\end{align}
\end{lemma}

\begin{proof}
Write $|y|={n_0}-2t$ and, up to normalization, let
\begin{align}
    B_{{n_0},t}
    :=
    \binom{m}{{n_0}-2t}
    \binom{m+t-1}{t}.
\end{align}
As described in the proof of Lemma~\ref{lemma: most parity outcomes have many odd parities}, the ratio $B_{n_0,t+1}/B_{n_0,t}$ is strictly smaller than $1$ whenever
$t>n_0/(\alpha_{n_0}+2)$, while $B_{n_0,t-1}/B_{n_0,t}$ is strictly smaller than $1$ whenever $t<n_0/(\alpha_{n_0}+2)$.
The corresponding central value is $|y| = n_0 - 2t = \mu_{n_0}^{\rm P}$.
The same ratio iteration as in Lemma~\ref{lem:tbs-leading-moderate} therefore gives
\begin{align}
\Pr\left[
\left||y|-\mu_N^{\rm P}\right|>u
\right]
\leq
2n\exp\left(-c\frac{u^2}{n}\right)
\end{align}
for a constant $c>0$.
Choosing $u=C_{\mathcal{Q}}\sqrt{n\log n}$ proves the claim.
\end{proof}

Given a parity outcome $y\in\mathcal P_{m,N}$, the randomized parity-rotation rule chooses a set $E_y\subseteq[m]\setminus\supp(y)$ uniformly with $|E_y|=|y|$ and then chooses a uniformly random matching between $\supp(y)$ and $E_y$.
The resulting paired set is denoted by $R_y$.

Let $\mathcal R_{m,n,k_0}^{\rm P}$ be the joint residual law obtained by sampling $x\sim\mathcal Y_{m,n}$ conditioned on
\begin{align}
    \mathsf E_n^{\rm P}
    :=
    \left\{
    \left||x|-\mu_n^{\rm P}\right|
    \leq
    C_{\mathcal{Q}}\sqrt{n\log n}
    \right\},
\label{eq: appendix typical event}
\end{align}
choosing the parity rotation data and the selected pairs as in Sec.~\ref{subsec: appendix parity leading collapse}, and outputting $(x',R')$.

\begin{lemma}[Parity residual-law domination]
\label{lem:residual}
Fix constants $\alpha>2$ and $0<\lambda<1$.
Let $m=\alpha n$, let $k_0\leq n^{\lambda/2}$, and set $n'=n-k_0$.
For every event $\mathcal B$ depending on a parity outcome and its randomized parity-rotation data,
\begin{align}
\Pr_{(x',R')\sim\mathcal R_{m,n,k_0}^{\rm P}}
[(x',R')\in\mathcal B]
&\leq
 e^{o(1)}
\Pr_{\substack{y\sim\mathcal Y_{m,n'}, R_y}}
[(y,R_y)\in\mathcal B]
+
 n^{-\mathcal{Q}}
\label{eq: app residual law bound}
\end{align}
for every constant $\mathcal{Q}>0$.
\end{lemma}

\begin{proof}
Let $q_{\rm P}:=\Pr_{x\sim\mathcal Y_{m,n}}[\mathsf E_n^{\rm P}]$.
By Lemma~\ref{lem:parity-leading-moderate}, $q_{\rm P}^{-1}=e^{o(1)}$.
Fix a residual parity string $y\in\mathcal P_{m,n'}$ and write $\ell:=|y|$.
To obtain $y$, the original string $x$ is obtained by adding $k_0$ ones outside $\supp(y)$, and there are $\binom{m-\ell}{k_0}$ choices.
Each such string has weight $\ell+k_0$, and
\begin{align}
    \frac{n-(\ell+k_0)}{2}
    =
    \frac{n'-\ell}{2}.
\end{align}
Given such an $x$, the selected set $J_O$ is uniform among the $\binom{\ell+k_0}{k_0}$ subsets of its support.
Consequently, we have 
\begin{align}
\mathcal R_{m,n,k_0}^{\rm P}(y)
\leq
q_{\rm P}^{-1}
\frac{\binom{m-\ell}{k_0}}
{\binom{\ell+k_0}{k_0}}
\frac{
\binom{m+\frac{n'-\ell}{2}-1}{\frac{n'-\ell}{2}}
}{
\binom{m+n-1}{n}
}.
\label{eq: app residual law}
\end{align}
On the other hand, the probability mass of $y$ under the distribution $Y_{m,n'}$ is given by
\begin{align}
    \mathcal Y_{m,n'}(y)
    =
    \frac{
    \binom{m+\frac{n'-\ell}{2}-1}{\frac{n'-\ell}{2}}
    }{
    \binom{m+n'-1}{n'}
    }.
\label{eq: app original law}
\end{align}
Thus, it holds that 
\begin{align}
\frac{\mathcal R_{m,n,k_0}^{\rm P}(y)}
{\mathcal Y_{m,n'}(y)}
&\leq
q_{\rm P}^{-1}
\frac{\binom{m-\ell}{k_0}}
{\binom{\ell+k_0}{k_0}}
\frac{\binom{m+n'-1}{n'}}
{\binom{m+n-1}{n}}
=
q_{\rm P}^{-1}
\prod_{j=1}^{k_0}
\frac{(m-\ell-j+1)(n'+j)}
{(\ell+j)(m+n'+j-1)}.
\label{eq: app ratio}
\end{align}
Following the same argument as in Lemma~\ref{lem:tbs-leading-residual}, under the residual construction, $\ell=\mu_{n'}^{\rm P}+O(\sqrt{n\log n})$ because $|x'|=|x|-k_0$, $\mu_n^{\rm P}-\mu_{n'}^{\rm P}=O(k_0)$, and $k_0=o(\sqrt{n\log n})$. 
Also, expanding the factor above gives
\begin{align}\label{eq: app decomp of product}
\frac{(m-\ell-j+1)(n'+j)}
{(\ell+j)(m+n'+j-1)}
=
\frac{(m-\ell)n'}{\ell(m+n')}
\left(1-\frac{j-1}{m-\ell}\right)
\left(1+\frac{j}{n'}\right) 
\left(1+\frac{j}{\ell}\right)^{-1}
\left(1+\frac{j-1}{m+n'}\right)^{-1}  . 
\end{align}
Since $j\le k_0=o(n)$, one can check that the leading factor of the product term is
\begin{align}
\frac{(m-\ell)n'}{\ell(m+n')}  ,
\end{align}
which is exactly one at $\ell=\mu_{n'}^{\rm P}$, and its logarithmic derivative is
\begin{align}
    -\frac{1}{m-\ell}-\frac{1}{\ell}
    =
    O\left(\frac{1}{n}\right).
\end{align}
The $j$-dependent factors contribute $O(k_0^2/n)$ after taking logarithms and summing over $j$.
Therefore,
\begin{align}
\log
\frac{\mathcal R_{m,n,k_0}^{\rm P}(y)}
{\mathcal Y_{m,n'}(y)}
\leq
|\log q_{\rm P}|
+
O\left(
 k_0\sqrt{\frac{\log n}{n}}
 +
 \frac{k_0^2}{n}
\right)
=
 o(1).
\end{align}

Conditioned on $y$, the residual even-mode set and matching are distributed according to the randomized parity-rotation rule by symmetry: averaging over the deleted odd modes and their matched even partners makes the remaining even-mode set uniform in the complement of $\supp(y)$.
Combining this fact with the pointwise domination above and Lemma~\ref{lem:parity-leading-moderate} proves Eq.~\eqref{eq: app residual law bound}.
\end{proof}

For later use, observe that the Haar-average parity weights satisfy an exact simplification.
Since $x'$ is obtained from $x$ by removing $k_0$ odd parities while reducing the photon number from $n$ to $n'=n-k_0$, we have
\begin{align}
    \frac{n-|x|}{2}
    =
    \frac{n'-|x'|}{2}.
\end{align}
Therefore, it follows that 
\begin{align}
    \frac{P_x^{(\mathrm{avg})}}
    {P_{x'}^{(\mathrm{avg})}}
    =
    \frac{\binom{m+n'-1}{n'}}
    {\binom{m+n-1}{n}}
    =
    \exp(-k_0\log(1+\alpha)+O(k_0^2/n)).
\label{eq: appendix avg ratio}
\end{align}

We now state an analogous anti-concentration conjecture for PBS Fourier coefficients. 

\begin{conjecture}[Parity Fourier-coefficient anti-concentration]
\label{ass:ptca}
Let $s$ be a photon number satisfying $m/s>2$ and $m/s=\Theta(1)$.
Draw $y\sim\mathcal Y_{m,s}$ and $U\sim\Haar(m)$ independently, and draw $R_y$ according to the randomized parity-rotation rule above.
Then
\begin{align}
\Pr\left[
|c_{2s}(y,U,R_y)|
\leq
\frac{P_y^{(\mathrm{avg})}}{\poly(s)}
\right]
\leq
\frac{1}{\poly(s)}.
\label{eq: appendix fourier anticonc}
\end{align}
\end{conjecture}

Finally, we prove the average-case hardness of PBS using the leading-coefficient extraction approach.

\begin{theorem}
\label{thm: parity average case leading}
Fix any constant $0<\lambda<1$.
Let $m=m(n)$ satisfy $m/n=\alpha+o(1)$ for a fixed $\alpha>2$.
Suppose Conjecture~\ref{ass:ptca} holds.
Then the {\rm APE} problem in Definition~\ref{def: parityBS} is $\#\mathrm P$-hard under a $\rm BPP$ reduction for error bounds satisfying
\begin{align}
    \varepsilon
    =
    \exp(-O(n^\lambda)),
    \qquad
    \delta
    =
    O\left(\frac{1}{n^{1+\lambda}}\right).
\end{align}
\end{theorem}

\begin{proof}
Let $\mathcal O$ be an oracle for APE, available for all input sizes in the reduction.
As in the threshold reduction, the residual call has photon number $n'=n-o(n)$ and the same mode number $m$, and the oracle guarantee is understood uniformly over the resulting $m/n'=\alpha+o(1)$ fixed-ratio family.
On input $x\in\mathcal P_{m,n}$ and $U\in{\rm U}(m)$, it outputs an additive estimate satisfying
\begin{align}
\Pr_{\substack{x\sim\mathcal Y_{m,n}\\ U\sim\Haar(m)}}
\left[
|\mathcal O(x,U)-P_x(U)|
>
\varepsilon P_x^{(\mathrm{avg})}
\right]
<
\delta.
\label{eq: appendix APE oracle}
\end{align}
We show how to compute $\Per(X_0)^2$ exactly for an arbitrary $X_0\in\{0,1\}^{k_0\times k_0}$.

Choose $n:=\left\lceil k_0^{2/\lambda}\right\rceil$ and $ m:=\alpha n$.
Sample $x\sim\mathcal Y_{m,n}$ and let $k:=|x|$.
By Lemma~\ref{lem:parity-leading-moderate}, with probability at least $1-n^{-\mathcal{Q}}$,
\begin{align}
    \left|k-\mu_n^{\rm P}\right|
    \leq
    C_{\mathcal{Q}}\sqrt{n\log n}.
\label{eq:parity-leading-main-typical}
\end{align}
For fixed $\alpha>2$, this event implies $k=\Theta(n)$, $k>k_0$, and $m-n-2k=\Omega(n)$, for all sufficiently large $n$.
Conditioned on Eq.~\eqref{eq:parity-leading-main-typical}, construct $O,E,R,J_O,J_E,J,R'$, $n'$, and $x'$ as in Sec.~\ref{subsec: appendix parity leading collapse}, using the randomized choices specified there.
Let $\mathcal R_{m,n,k_0}^{\rm P}$ be the corresponding residual law.

Define the sparse perturbation $A_{X_0}$ by Eq.~\eqref{eq: appendix sparse AX} with $\kappa=m^{-1/2}$.
Let $A=U_{R,[n]}$ for $U\sim\Haar(m)$, set $A(t)=A+tA_{X_0}$, and define $p(t) := c_{2n'}(x,U(t),R')$ for $U(t)$ being a unitary completion of $A(t)$.
By Lemma~\ref{lem:collapse}, $p(t)$ has degree at most $2k_0$ and
\begin{align}
    [t^{2k_0}]p(t)
    =
    \kappa^{2k_0}\Per(X_0)^2
    c_{2n'}(x',U,R').
\label{eq:parity-leading-main-collapse}
\end{align}

Now choose
\begin{align}
    \Delta
    :=
    \frac{c_\Delta\delta}{\sqrt\alpha\,n k_0}
\label{eq: appendix Delta choice}
\end{align}
for a sufficiently small constant $c_\Delta>0$.
By Lemma~\ref{lem: appendix sparse translation}, for every $t\in[0,\Delta]$, the unitary completion $U(t)$ has distribution $\mathcal H_t$ satisfying
\begin{align}
    \|\mathcal H_t-\Haar(m)\|_{\mathrm{TVD}}
    \leq
    O(\delta)+e^{-\Omega(n)}.
\end{align}
Thus each oracle call on $V_{R'}(\theta)U(t)$ fails with probability at most $O(\delta)+e^{-\Omega(n)}$.

Let $L_F:=4n+1$ and $\theta_j:=2\pi j/L_F$.
The parity probability $P_x(V_{R'}(\theta)U(t))$ has degree at most $2n$, so Lemma~\ref{thm: fourier coeff ext}, applied with $D=2n$ and $r=2n'$, estimates $p(t)$ to additive error $\varepsilon P_x^{(\mathrm{avg})}$ whenever all $L_F$ oracle calls succeed.
Applying this procedure at the points
\begin{align}
    t_i:=\frac{i\Delta}{2k_0},
    \qquad
    i=0,1,\ldots,2k_0,
\end{align}
and then applying Lemma~\ref{thm: polynomial coefficient extraction} gives an estimate $\widetilde c_{2k_0}$ satisfying
\begin{align}
    \left|
    \widetilde c_{2k_0}-[t^{2k_0}]p(t)
    \right|
    \leq
    \varepsilon P_x^{(\mathrm{avg})}
    \exp(2k_0\log\Delta^{-1}+O(k_0)).
\label{eq: appendix numerator coeff error}
\end{align}
The total failure probability of this numerator estimate is $O(nk_0\delta)+e^{-\Omega(n)}=o(1)$.

We also separately estimate $c_{2n'}(x',U,R')$.
By Lemma~\ref{lem:residual}, the APE-oracle failure probability under the actual residual law is at most $e^{o(1)}\delta+n^{-\mathcal{Q}}$.
Applying Lemma~\ref{thm: fourier coeff ext} to the $n'$-photon residual instance yields an estimate $\widehat c$ satisfying
\begin{align}
    |\widehat c- c_{2n'}(x',U,R') |
    \leq
    \varepsilon P_{x'}^{(\mathrm{avg})}
\label{eq: appendix denominator coeff error}
\end{align}
except with probability $O(n\delta+n^{1-\mathcal{Q}})$.
The same residual-law domination transfers Conjecture~\ref{ass:ptca}, so with probability $1-o(1)$, we have 
\begin{align}
    |c_{2n'}(x',U,R')|
    \geq
    \frac{P_{x'}^{(\mathrm{avg})}}{\poly(n)}.
\label{eq: appendix anticonc event}
\end{align}

Now, using Eq.~\eqref{eq:parity-leading-main-collapse} and the same ratio estimate as in the threshold proof, we obtain, on the joint good event,
\begin{align}
\left|
\frac{|\widetilde c_{2k_0}|}{|\widehat c|}
- \kappa^{2k_0}\Per(X_0)^2
\right|
&\leq
\varepsilon
\left(
\frac{P_x^{(\mathrm{avg})}}
{P_{x'}^{(\mathrm{avg})}}
\right)
\poly(n)
\exp(2k_0\log\Delta^{-1}+O(k_0))
+
\varepsilon\poly(n) \kappa^{2k_0}\Per(X_0)^2 .
\label{eq: appendix ratio error first}
\end{align}
By Eq.~\eqref{eq: appendix avg ratio}, $\Delta^{-1}=\poly(n)$, $\kappa^{2k_0}=m^{-k_0}$, and $\Per(X_0)^2\leq(k_0!)^2$, the right-hand side is at most $\varepsilon\exp(O(k_0\log n))$, thus dividing by $\kappa^{2k_0}$ gives an estimate of $\Per(X_0)^2$ with additive error at most $\varepsilon\exp(O(k_0\log n))$.
Since $k_0\leq n^{\lambda/2}$, choosing $\varepsilon=\exp(-O(n^\lambda))$ with a sufficiently large constant makes this error smaller than $1/2$.
Rounding and repetition recover $\Per(X_0)^2$ exactly, completing the proof.
\end{proof}

\end{document}